\documentclass[11pt]{amsart}
 \usepackage{amsmath,amssymb,amsfonts,amsthm,amscd,textcomp,times,booktabs}
 \usepackage[dvipdfmx]{graphicx}
  \usepackage{braket}
  \usepackage{enumerate}
  \usepackage{color,float}
  \usepackage{bm}
  \usepackage{multirow}
  \usepackage{arydshln}
  \usepackage[all]{xy}
  \usepackage[normalem]{ulem}
  \usepackage{mathtools}
\newtheorem{theorem}{Theorem}

\newtheorem{proposition}[theorem]{Proposition}
\newtheorem{lemma}[theorem]{Lemma}

\newtheorem{corollary}[theorem]{Corollary}
\newtheorem{remark}[theorem]{Remark}
\newtheorem{example}[theorem]{Example}
\newtheorem{problem}[theorem]{Problem}

\newtheorem{construction}[theorem]{Construction}
\newtheorem{convention}[theorem]{Convention}

\numberwithin{equation}{section}
\numberwithin{theorem}{section}
\newcommand{\bettershortstack}[2][c]{\begin{tabular}[b]{@{}#1@{}}#2\end{tabular}}

\newcommand{\Lt}{\ensuremath{\mathrm{L}}}
\newcommand{\Rt}{\ensuremath{\mathrm{R}}}
\newcommand{\norm}[1]{\ensuremath{\left| #1 \right|}}
\newcommand{\ora}[1]{\ensuremath{\overrightarrow{#1}}}

\newcommand{\nega}{\ensuremath{\mathrm{neg}}}

\newcommand{\proj}{\ensuremath{\mathrm{proj}}}

\begin{document}
%%%%%%%%%%%%%%%%%%%%%%%%%%%%%%%%%%%%%%%%%%%%%%%%%%%%
\title[Negative 3D gadgets in origami extrusions with a supporting triangle on the back side]
{Negative 3D gadgets in origami extrusions\\
with a supporting triangle on the back side}
\author{Mamoru Doi}
\address{11-9-302 Yumoto-cho, Takarazuka, Hyogo 665-0003, Japan}
\email{doi.mamoru@gmail.com}
\maketitle
\noindent{\bfseries Abstract. }
In our previous two papers, we studied (positive) $3$D gadgets in origami extrusions
which create a top face parallel to the ambient paper and two side faces sharing a ridge with two simple outgoing pleats.

Then a natural problem comes up whether it is possible to construct a `negative' $3$D gadget from any positive one
having the same net without changing the outgoing pleats,
that is, to sink the top and two side faces of any positive $3$D gadget to the reverse side without changing the outgoing pleats.
Of course, simply sinking the faces causes a tear of the paper, and thus we have to modify the crease pattern.
There are two known constructions of negative $3$D gadgets before ours,
but they do not solve this problem because their outgoing pleats are different from positive ones.

In the present paper we give an affirmative solution to the above problem.
For this purpose, we present three constructions of negative $3$D gadgets with a supporting triangle on the back side,
which are based on our previous ones of positive $3$D gadgets.
The first two are an extension of those presented in our previous paper, and the third is new.
We prove that our first and third constructions solve the problem.
Our solutions enable us to deal with positive and negative $3$D gadgets on the same basis,
so that we can construct from an origami extrusion constructed with $3$D gadgets its negative using the same pleats
if there are no interferences among the $3$D gadgets.
We also treat repetition/division of negative $3$D gadgets under certain conditions, which reduces their interferences with others.

\section{Introduction}\label{sec:introduction}
As in our previous two papers \cite{Doi19}, \cite{Doi20}, we are concerned with $3$D gadgets in origami extrusions,
where an origami extrusion is a folding of a $3$D object in the middle of a flat piece of paper, 
and $3$D gadgets are ingredients of origami extrusions which create faces with solid angles.
In particular, we focus on $3$D gadgets with two simple outgoing pleats which create a top face parallel to the ambient paper, and two side faces sharing a ridge,
where a simple pleat is one that is formed by a mountain and a valley fold parallel to each other.
There are two known constructions of such $3$D gadgets.
One is developed by Carlos Natan \cite{Natan} as a generalization of the conventional cube gadget,
and these conventional $3$D gadgets have a supporting pyramid on the back side.
The other newer one is presented in our previous papers \cite{Doi19}, \cite{Doi20}.
Our new $3$D gadgets, which have flat back sides above the ambient paper, are completely downward compatible with the conventional ones in the sense
that any conventional one can be replaced by ours with the same outgoing pleats, but the converse is not always possible.
Also, there are several advantages of our new $3$D gadgets over the conventional ones.
In particular, we can change angles of the outgoing pleats under certain conditions.

Let us call the above $3$D gadgets \emph{positive} $3$D gadgets.
For the development of a given positive $3$D gadget, if we forget the crease patten and only consider the net with the same orientation,
we may expect as the resulting $3$D object not only the original one extruded with the $3$D gadget, but also its \emph{negative}, that is,
the $3$D object which is constructed by sinking all faces of the original one to the reverse side.
We shall say that a $3$D gadget is a \emph{negative} one if it extrudes the negative of the $3$D object which a positive one does.

Then the following problem naturally comes up.
\begin{problem}\rm\label{prob:existence_negative}
We focus on positive $3$D gadgets with two simple outgoing pleats which creates a top face parallel to the ambient paper and two side faces sharing a ridge,
which are all presented in \cite{Doi20}.
Then for any given positive $3$D gadget, is it possible to construct a negative one with the same net as the positive one without changing the outgoing pleats?
\end{problem}

For an individual positive $3$D gadget, there may be a desirable negative one.
For example, let us consider a simple case of cube gadgets as shown in Figure $\ref{fig:CP_cubes}$,
where `CP' stands for `crease pattern' throughout this paper.
Then there correspond negative $3$D gadgets of two types as shown in Figure $\ref{fig:CP_cubes_negative_2}$.
Although there are two known constructions of negative $3$D gadgets before ours which apply to any net of positive $3$D gadgets,
these constructions do not solve Problem $\ref{prob:existence_negative}$.
Indeed, negative cube gadgets constructed by the two known constructions, which are shown in Figure $\ref{fig:CP_cubes_negative_1}$,
have different outgoing pleats from those of the original ones.
In these figures, the faces of the outgoing simple pleats bounded by a mountain and a valley fold are shaded.
\begin{figure}[htbp]
\addtocounter{theorem}{1}
\centering\includegraphics[width=0.48\hsize]{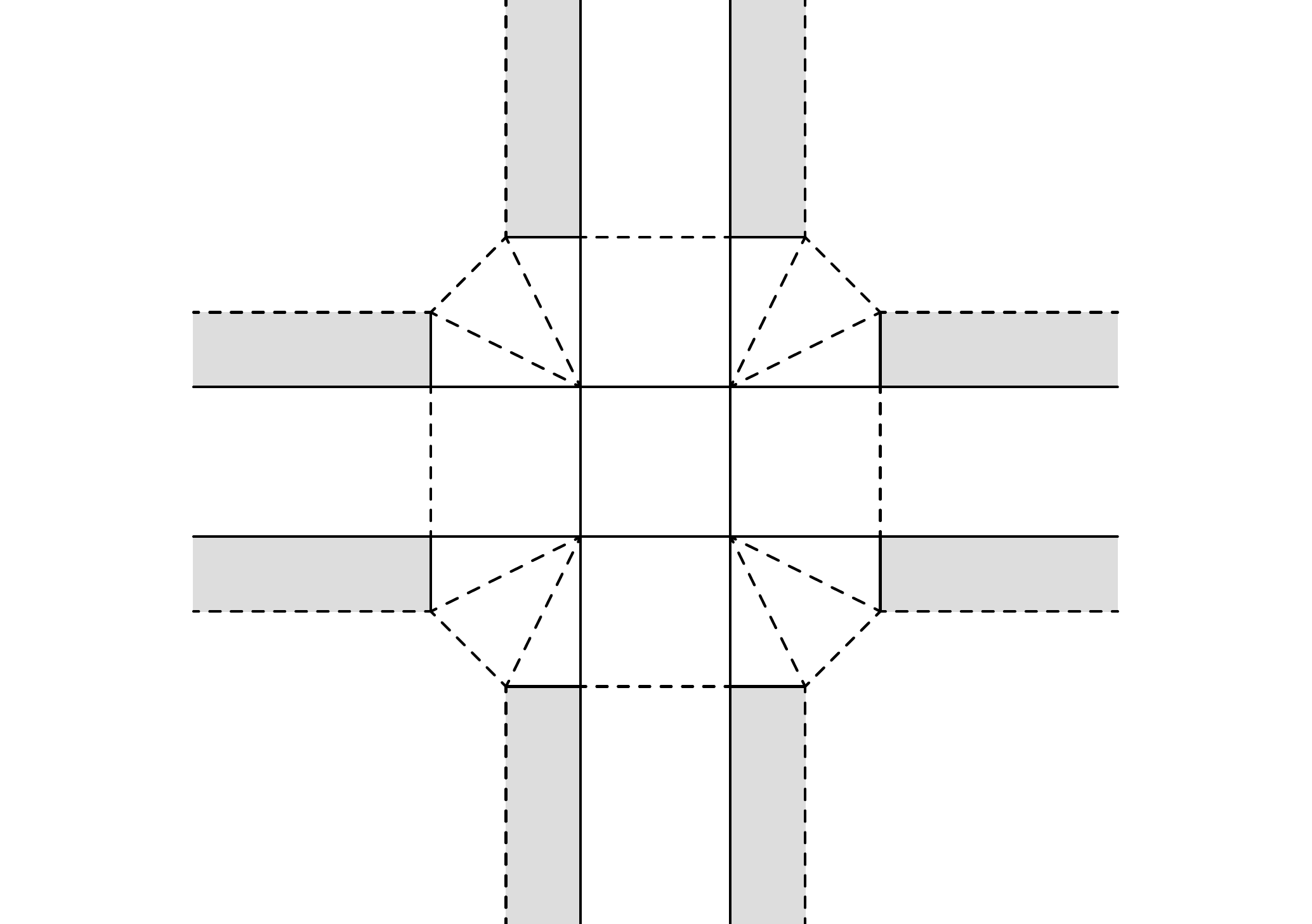}
\includegraphics[width=0.48\hsize]{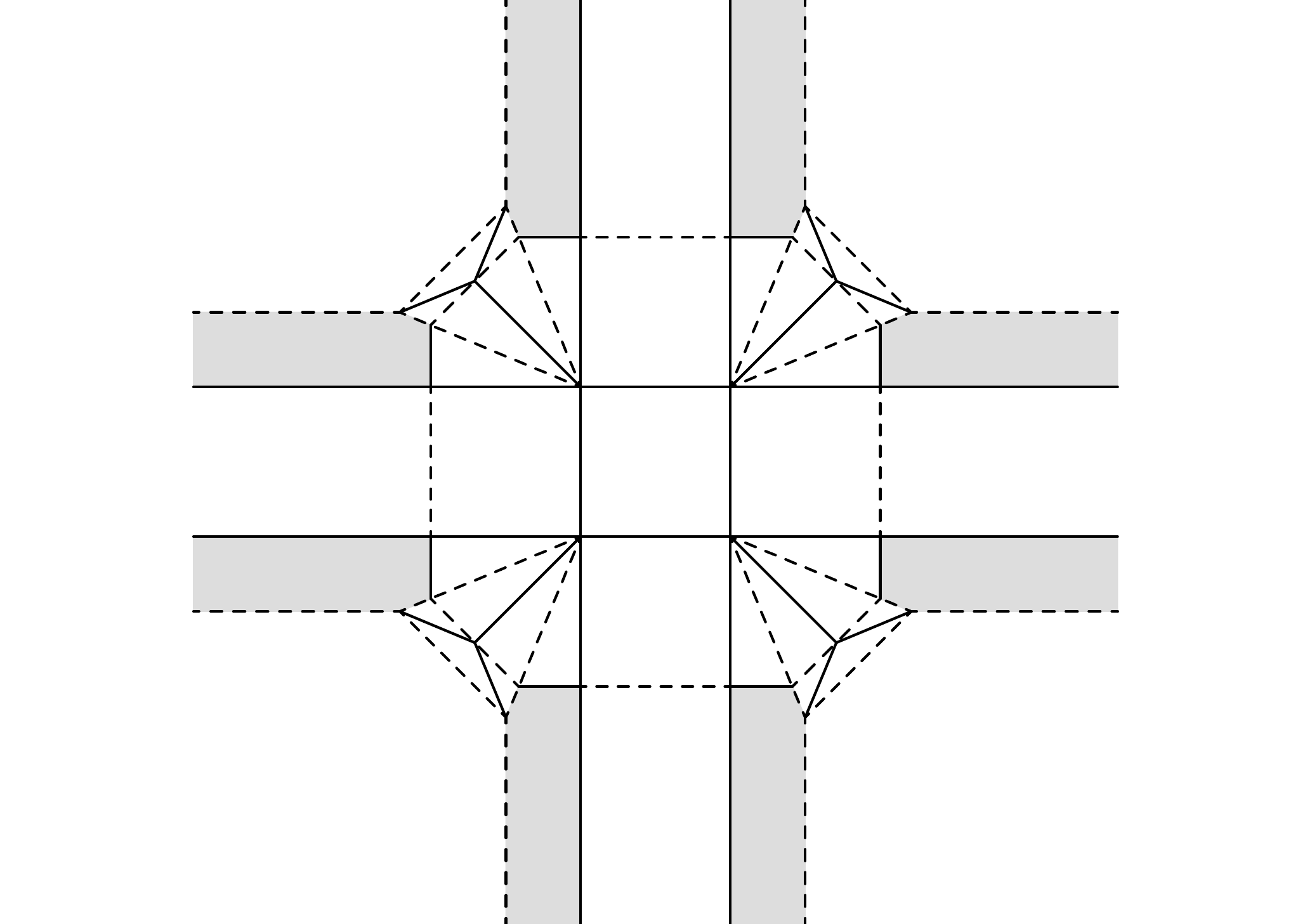}
    \caption{CPs of a cube extruded with the conventional and our cube gadgets}
    \label{fig:CP_cubes}
\end{figure}
\begin{figure}[htbp]
\addtocounter{theorem}{1}
\centering\includegraphics[width=0.48\hsize]{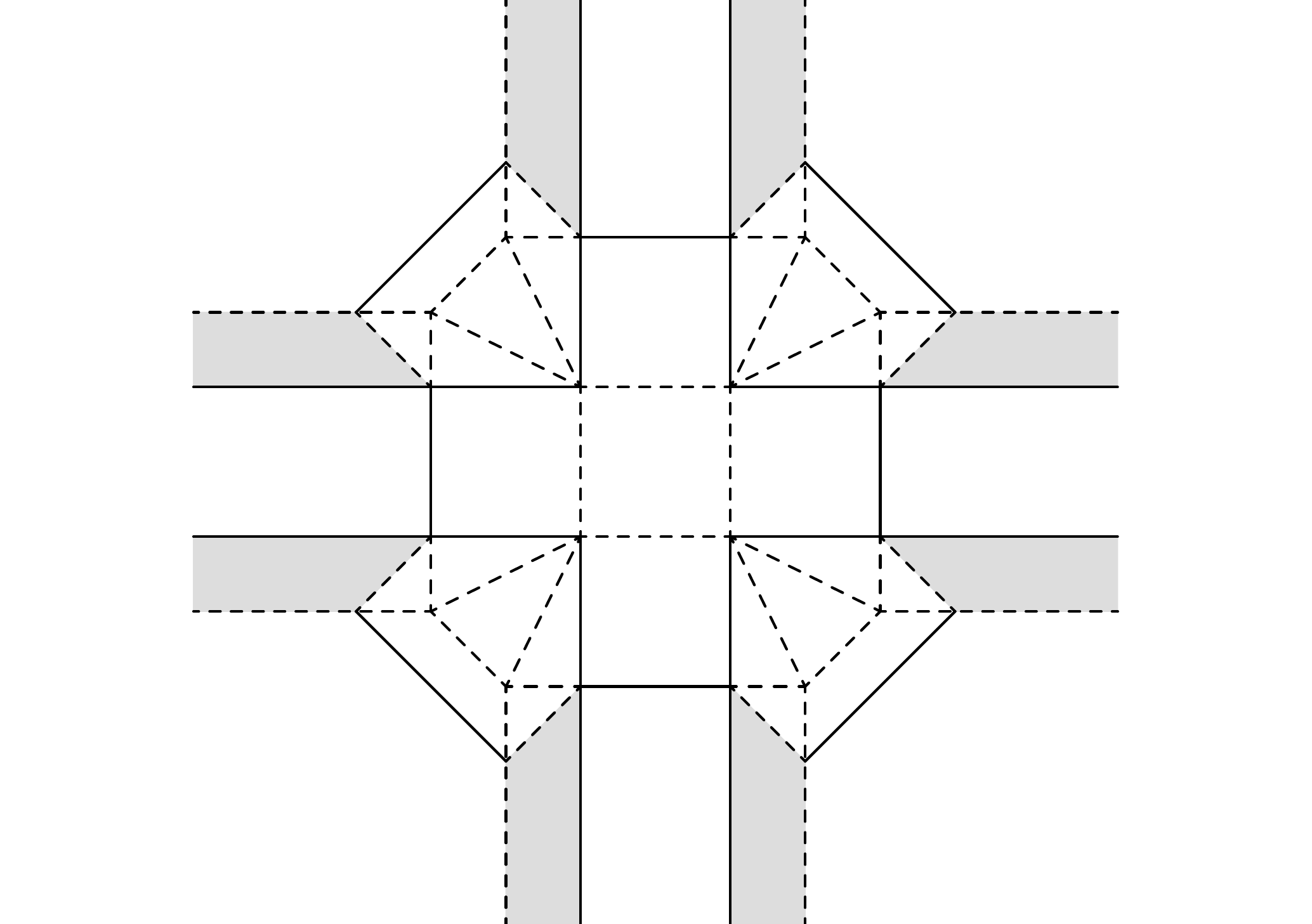}
\includegraphics[width=0.48\hsize]{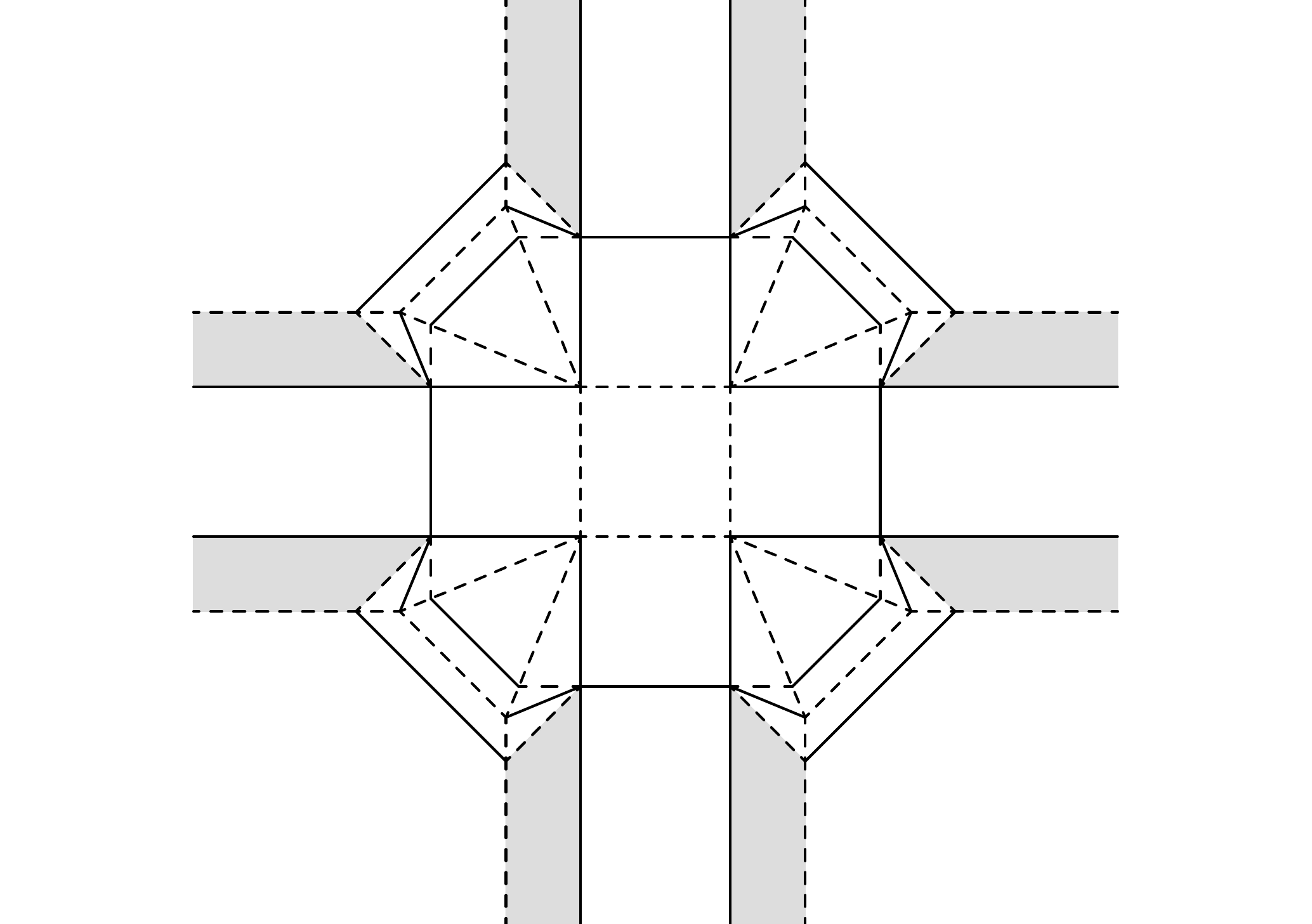}
    \caption{CPs of a negative cube extruded with the conventional and our negative cube gadgets}
    \label{fig:CP_cubes_negative_2}
\end{figure}
\begin{figure}[htbp]
\addtocounter{theorem}{1}
\centering\includegraphics[width=0.48\hsize]{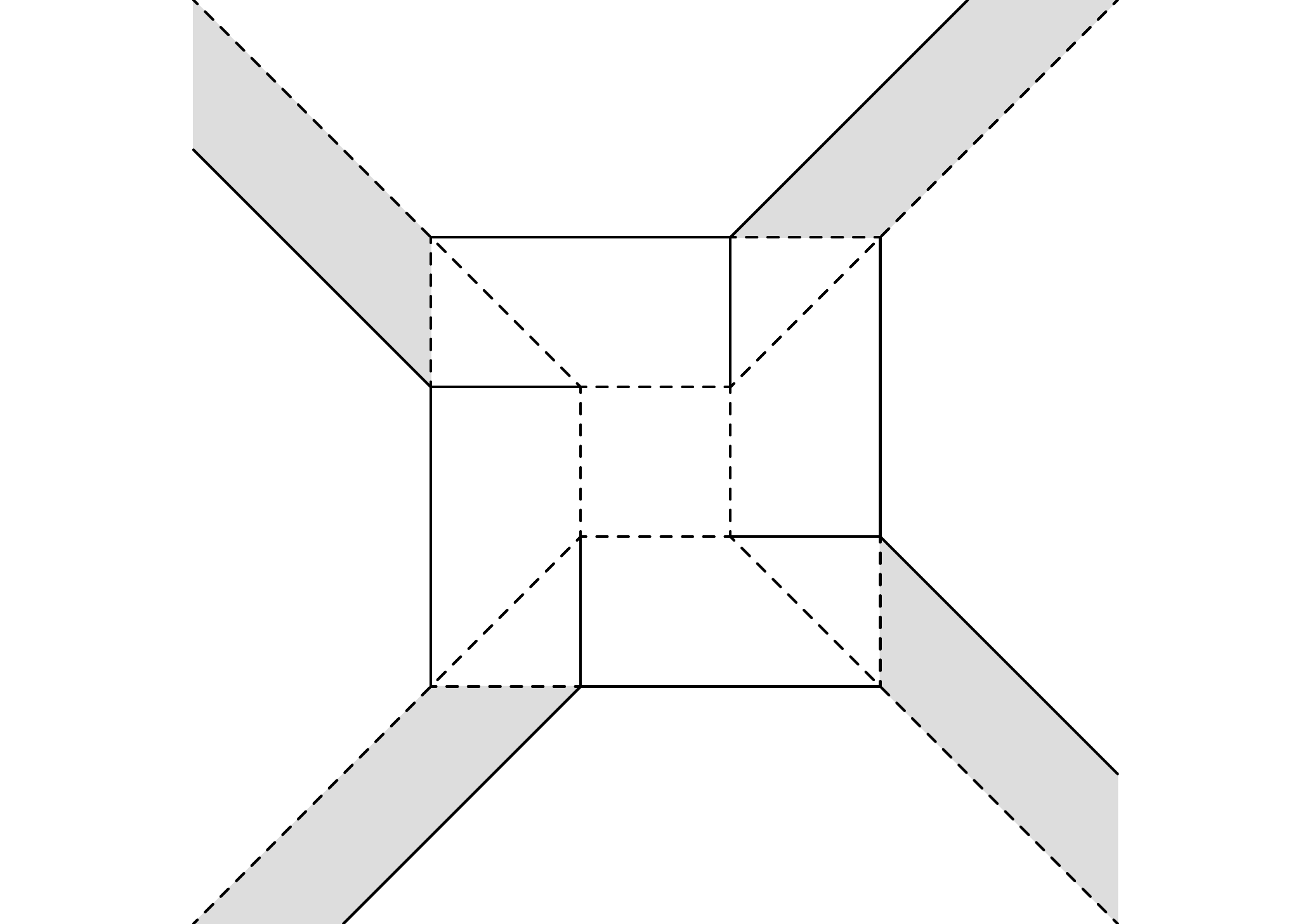}
\includegraphics[width=0.48\hsize]{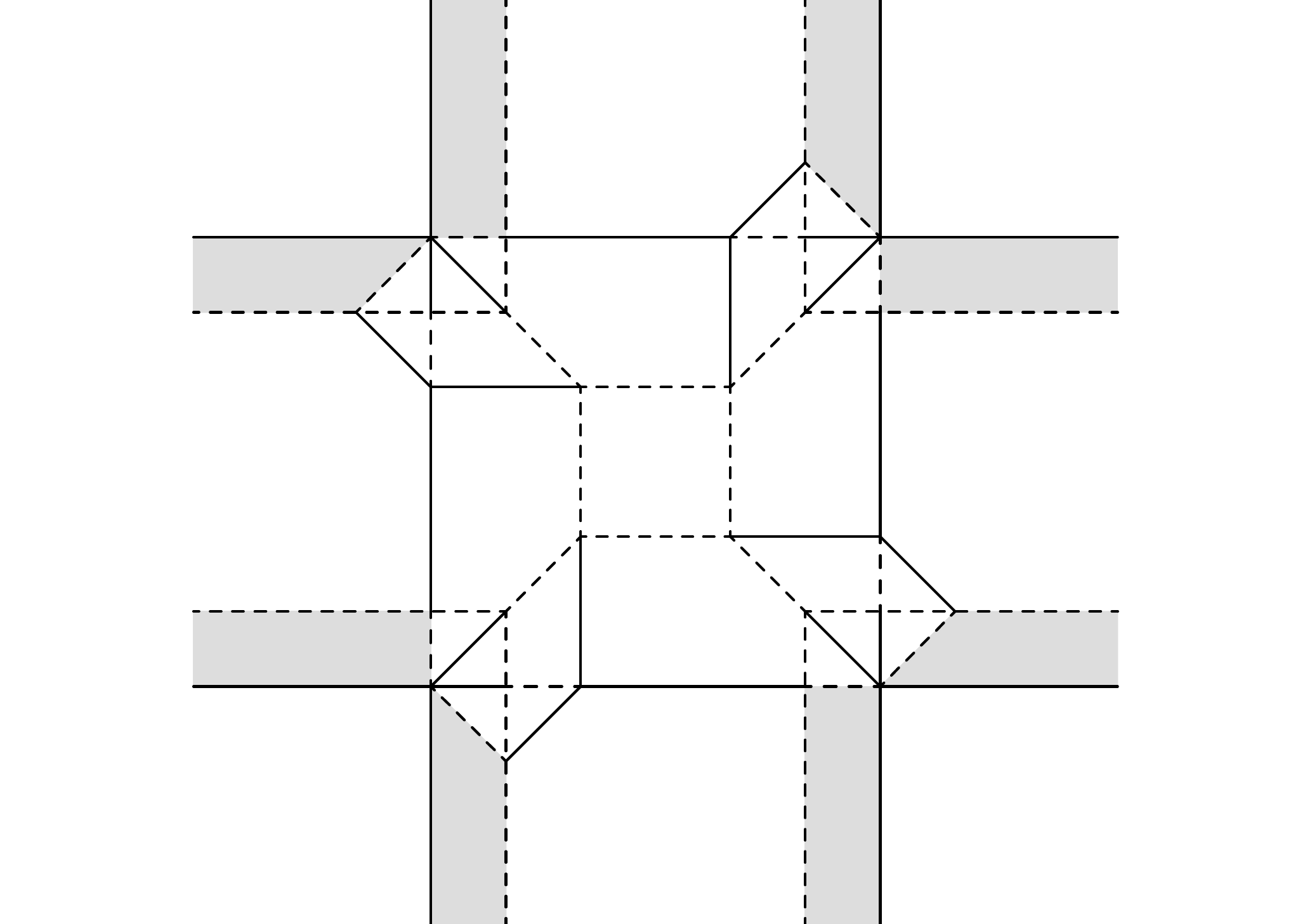}
    \caption{CPs of a negative cube extruded with negative cube gadgets by two known constructions before ours}
    \label{fig:CP_cubes_negative_1}
\end{figure}

In the present paper we give affirmative solutions to Problem $\ref{prob:existence_negative}$.
For this purpose, we present three constructions of negative $3$D gadgets with a supporting triangle on the back side.
These constructions are based on those presented in \cite{Doi19} and \cite{Doi20}, with a major difference that
two supporting flaps on the back side in the positive cases, which are called `ears' and overlap with the side faces,
form a single triangle and do not generally overlap with the side faces in the negative cases.
Among our three constructions, the first two are an extension of those presented in \cite{Doi19}, Section $9$, and the third is new.
We will see that our first construction produces infinitely many negative $3$D gadgets for any positive one in Problem $\ref{prob:existence_negative}$,
while our third one gives a unique negative $3$D gadget.

Our solutions to Problem $\ref{prob:existence_negative}$ enables us to deal with positive and negative $3$D gadgets on the same basis.
Thus from an origami extrusion constructed with $3$D gadgets, may they be a mixture of positive and negative ones,
we can construct its negative if there are no interferences among the $3$D gadgets.

We will also review two known constructions of negative $3$D gadgets before ours because they are useful in their own right.
The first is well-known, and the second is presented by Cheng in \cite{Cheng},
for which we shall give an extension corresponding to the changes of the angles of the outgoing pleats.

\begin{convention}\rm
For the later convenience, we will hereafter adopt the following convention on the standard orientation of the development
for both positive and negative $3$D gadgets:
\begin{itemize}
\item[] For a given $3$D gadget, we determine the standard side from which we see the development so that the extruded faces lie \emph{above} the ambient paper.
In other words, for a negative (resp. positive) $3$D gadget we consider its development seen from the \emph{back} (resp. front) side.
We add subscripts `$\Lt$' and `$\Rt$', which stand for `left' and `right' respectively,
to the points, lines, angles, etc. of the $3$D gadget according to this orientation.
\end{itemize}
\end{convention}
One advantage of the above convention is that a positive and a negative $3$D gadget constructed from the same development engage with each other.
We have another advantage when we consider both positive and negative $3$D gadgets.
If a positive and a negative $3$D gadget share a side face, then the top edge of one gadget is the bottom edge of the other.
However, if we see either gadget from its standard side, then the subscripts `$\Lt$' and `$\Rt$' of the other gadget are consistent with the orientation.
Also note that to construct the negative $3$D gadget corresponding to a given positive one in Problem $\ref{prob:existence_negative}$,
we have to begin with the horizontally flipped development of the positive one.\\

\noindent {\bfseries Notation and terminology.} 
To keep consistency with the previous papers \cite{Doi19} and \cite{Doi20}, we will use the same notation and terminology as possible as we can.
For example, we use subscript $\sigma$ for $\Lt$ and $\Rt$ which stands for `left' and `right' respectively,
and $\sigma'$ to indicate the other side of $\sigma$, that is,
\begin{equation*}
\sigma' =
\begin{dcases}
\Rt&\text{if }\sigma =\Lt ,\\
\Lt&\text{if }\sigma =\Rt .
\end{dcases}
\end{equation*}
Also, we will denote a simple pleat as $(\ell ,m)$, where $\ell$ is either a mountain or a valley fold, and $m$ is the reverse fold.\\

To fix the conditions and the construction of typical points, lines and angles common to most of the later constructions of negative $3$D gadgets,
we prepare the following.
\begin{figure}[htbp]
\addtocounter{theorem}{1}
\centering\includegraphics[width=0.75\hsize]{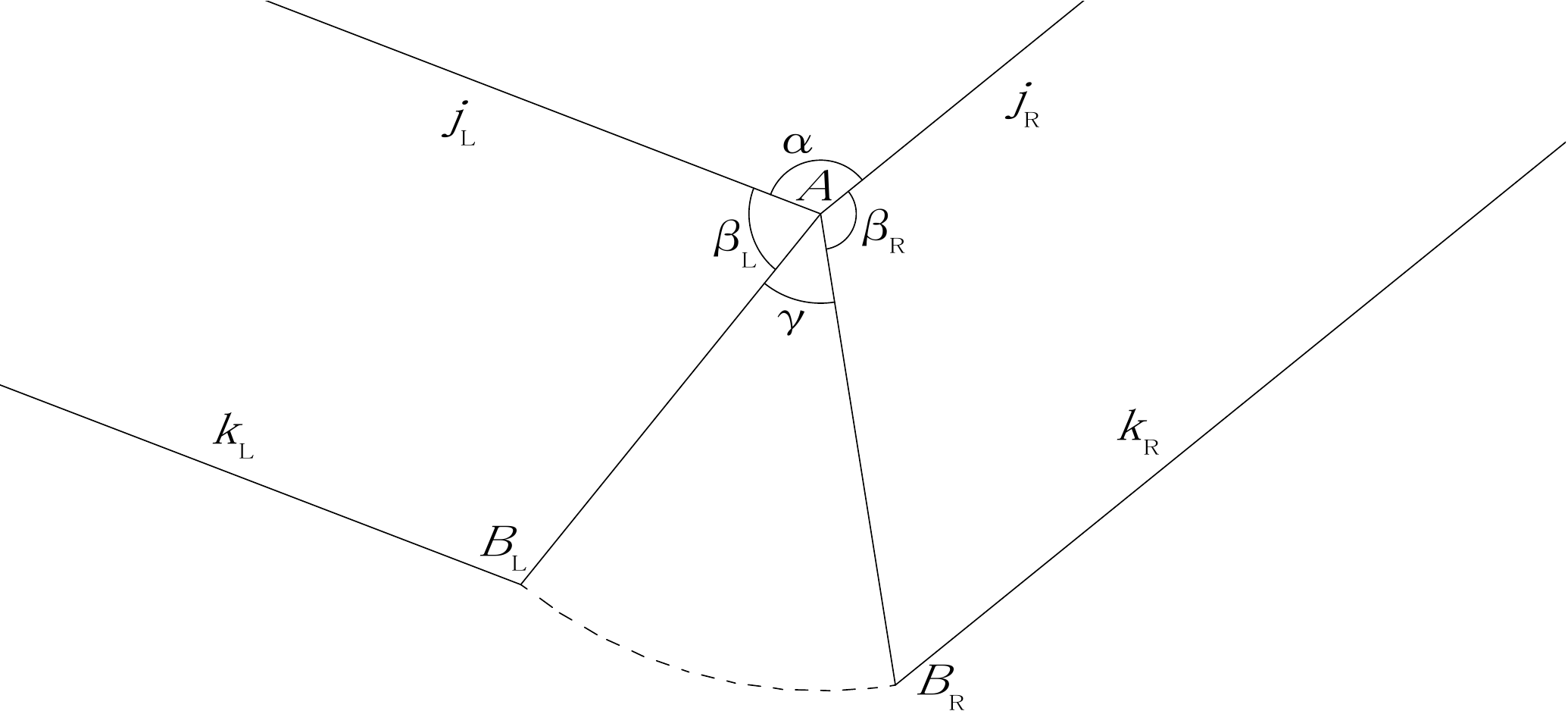}
    \caption{Net of a $3$D object which we want to extrude with a $3$D gadget}
    \label{fig:development_0}
\end{figure}
\begin{construction}\label{const:condition}\rm
We consider a net on a piece of paper as in Figure $\ref{fig:development_0}$.
We require the following conditions.
\begin{enumerate}[(i)]
\item $\alpha <\beta_\Lt + \beta_\Rt$, $\beta_\Lt <\alpha +\beta_\Rt$ and $\beta_\Rt <\alpha+ \beta_\Lt$.
\item $\alpha +\beta_\Lt +\beta_\Rt <2\pi$.
\end{enumerate}
To construct a negative gadget from the above net, we prescribe its simple outgoing pleats by introducing parameters $\delta_\sigma$
for their changes from the direction of $\ora{AB_\sigma}$ for $\sigma =\Lt ,\Rt$, for which we further require the following conditions:
\begin{enumerate}
\item[(iii.a)] $\delta_\Lt ,\delta_\Rt\geqslant 0$, where we take clockwise (resp. counterclockwise) direction as positive for $\sigma =\Lt$ (resp. $\sigma=\Rt$).
\item[(iii.b)] $\delta_\sigma <\beta_\sigma$ and $\delta_\sigma <\pi /2$ for $\sigma =\Lt ,\Rt$. 
\item[(iii.c)] $\alpha +\beta_\Lt +\beta_\Rt -\delta_\Lt -\delta_\Rt >\pi$, or equivalently, $\gamma +\delta_\Lt +\delta_\Rt <\pi$.
\end{enumerate}
In particular, if $\delta_\Lt =\delta_\Rt =0$, then conditions (iii.a)--(iii.c) are simplified as
\begin{enumerate}[(i)]
\setcounter{enumi}{2}
\item $\alpha +\beta_\Lt +\beta_\Rt >\pi$, or equivalently, $\gamma <\pi$.
\end{enumerate}
Then we construct creases $(\ell_\Lt ,m_\Lt )$ and $(\ell_\Rt ,m_\Rt )$ (resp. $(\ell'_\Lt ,m_\Lt )$ and $(\ell'_\Rt ,m_\Rt )$)
which we prescribe as the outgoing pleats of a negative $3$D gadget in Section $\ref{sec:negative_new}$ (resp. in Section $\ref{subsec:negative_Cheng}$)
as follows, where we regard $\sigma$ as taking both $\Lt$ and $\Rt$.
\begin{enumerate}[(1)]
\item Draw a ray $\ell_\Lt$ starting from $B_\Lt$ and going to the direction of $\ora{AB_\Lt}$ followed by a clockwise rotation by $\delta_\Lt$.
Also, draw a ray $\ell_\Rt$ starting from $B_\Rt$ ang going to the direction of $\ora{AB_\Rt}$ followed by a counterclockwise rotation by $\delta_\Rt$.
\item Draw a perpendicular to $\ell_\sigma$ through $B_\sigma$ for both $\sigma =\Lt ,\Rt$, letting $C$ be the intersection point. 
\item Draw a perpendicular bisector $m_\sigma$ to segment $B_\sigma C$ for both $\sigma =\Lt ,\Rt$, letting $P$ be the intersection point. 
\item Draw a ray $\ell'_\sigma$ parallel to and going in the same direction as $\ell_\sigma$ and $m_\sigma$, starting from $C$. 
\end{enumerate}
The resulting creases are shown as solid lines in Figure $\ref{fig:development_1}$.
\end{construction}
\begin{figure}[htbp]
\addtocounter{theorem}{1}
\centering\includegraphics[width=0.75\hsize]{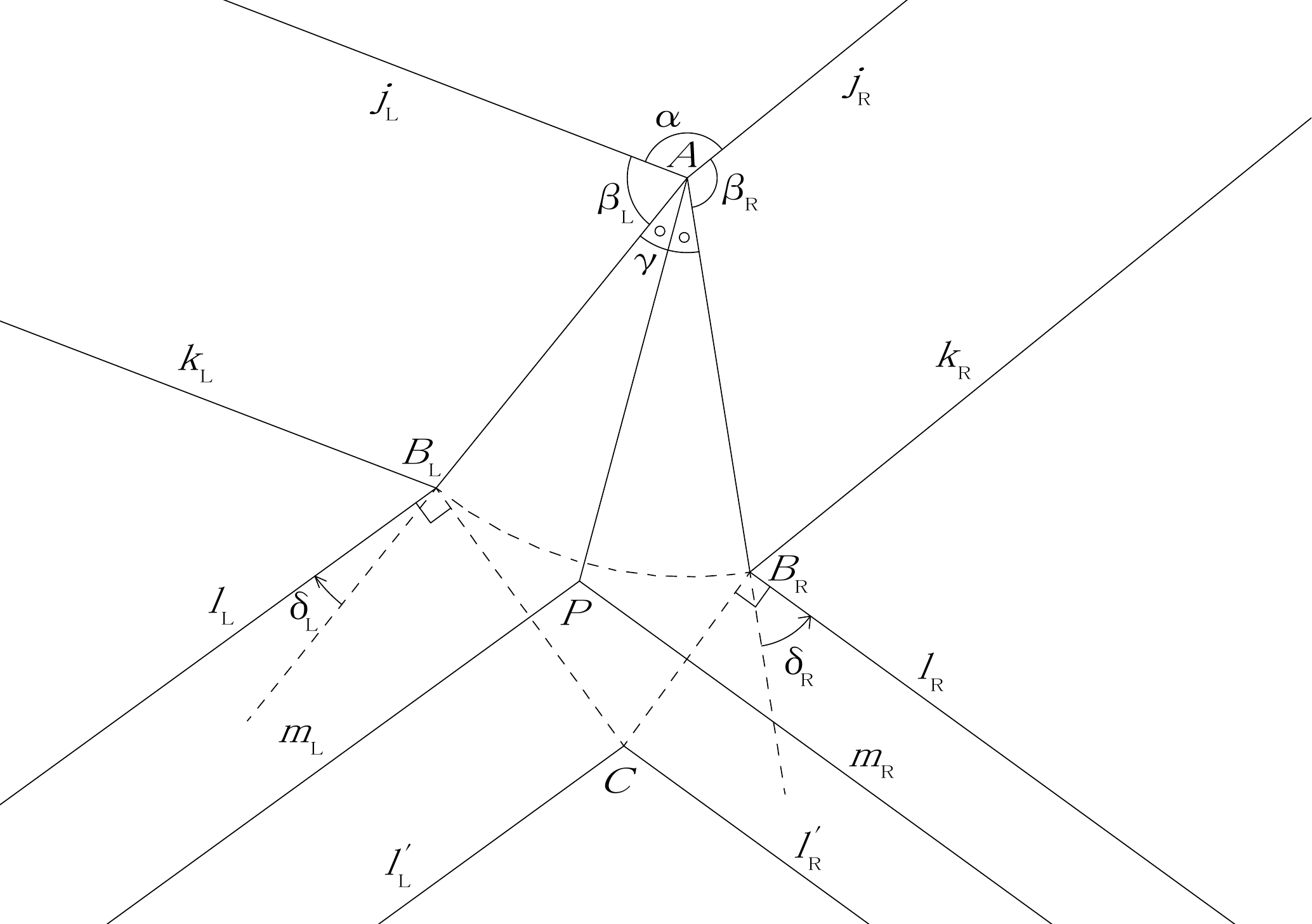}
    \caption{Construction of creases of prescribed pleats for a $3$D gadget}
    \label{fig:development_1}
\end{figure}
The following result will be crucial in Cheng's construction of negative $3$D gadgets given in Section $\ref{subsec:negative_Cheng}$.
\begin{proposition}\label{prop:AP}
Segment $AP$ bisects $\gamma =\angle B_\Lt AB_\Rt$.
Also, we have
\begin{equation*}
\angle AB_\Lt P=\angle AB_\Rt P=\gamma /2 +\delta_\Lt +\delta_\Rt .
\end{equation*}
\end{proposition}
\begin{proof}
Since $P$ is the excenter of $\triangle CB_\Lt B_\Rt$, segment $AP$ is a perpendicular bisector of side $B_\Lt B_\Rt$ of $\triangle AB_\Lt B_\Rt$
with $\norm{AB_\Lt}=\norm{AB_\Rt}$, and thus $AP$ bisects $\gamma$.
This proves the first assersion.
For the second assersion, See \cite{Doi20}, Lemma $3.5$.
\end{proof}

We end this section with the organization of this paper.
In Section $\ref{sec:negative_known}$, we review two known constructions of negative $3$D gadgets before ours.
Section $\ref{subsec:negative_onepleat}$ gives a well-known construction of negative $3$D gadgets with one simple outgoing pleat,
and Section $\ref{subsec:negative_Cheng}$ gives an extension of Cheng's construction.
In Section $\ref{sec:negative_new}$, we present our three constructions of negative $3$D gadgets,
where the first two are an extension of those presented in \cite{Doi19}, Section $9$, and the third is new.
In particular, we give a solution to Problem $\ref{prob:existence_negative}$ in Sections $\ref{subsec:negative_new_1}$ and $\ref{subsec:negative_new_3}$.
Section $\ref{subsec:interferences}$ discusses the condition that there are no interferences between two $3$D gadgets
which share a side face and include at least a negative one.
Section $\ref{sec:repetition}$ deals with repetition/division of negative $3$D gadgets under certain conditions.
Finally, we give our conclusion in Section $\ref{sec:conclusion}$.

\section{Known constructions of negative $3$D gadgets}\label{sec:negative_known}
In this section we review two known constructions of negative $3$D gadgets.
A basic idea common to both constructions is to extend one side face by swinging the flap resulting from bisecting $\gamma$ to the other side face.
It turns out that we can swing the flap to either side face and obtain two negative $3$D gadgets for each choice of the angles.
The resulting $3$D gadgets are never symmetric even when both $\beta_\Lt =\beta_\Rt$ and $\delta_\Lt =\delta_\Rt$ hold.

\subsection{Construction with one simple pleat}\label{subsec:negative_onepleat}
This construction is simple, but the resulting pleat is easy to loosen.
Thus to stabilize the resulting extrusion, we may need to combine it with additional folds such as twists.
Also, the resulting $3$D gadget has a large interference distance to other $3$D gadgets.
For a development as in Figure $\ref{fig:development_1}$, precreasing as in Figure $\ref{fig:precreasing_negative_onepleat}$ and swinging the resulting flap
formed by $B_\Lt Ab$ and $B_\Rt Ab$ to either side face until it extends the other side face and the connected pleat overlaps with the ambient paper,
we obtain the following construction.
\begin{figure}[htbp]
\addtocounter{theorem}{1}
\centering\includegraphics[width=0.75\hsize]{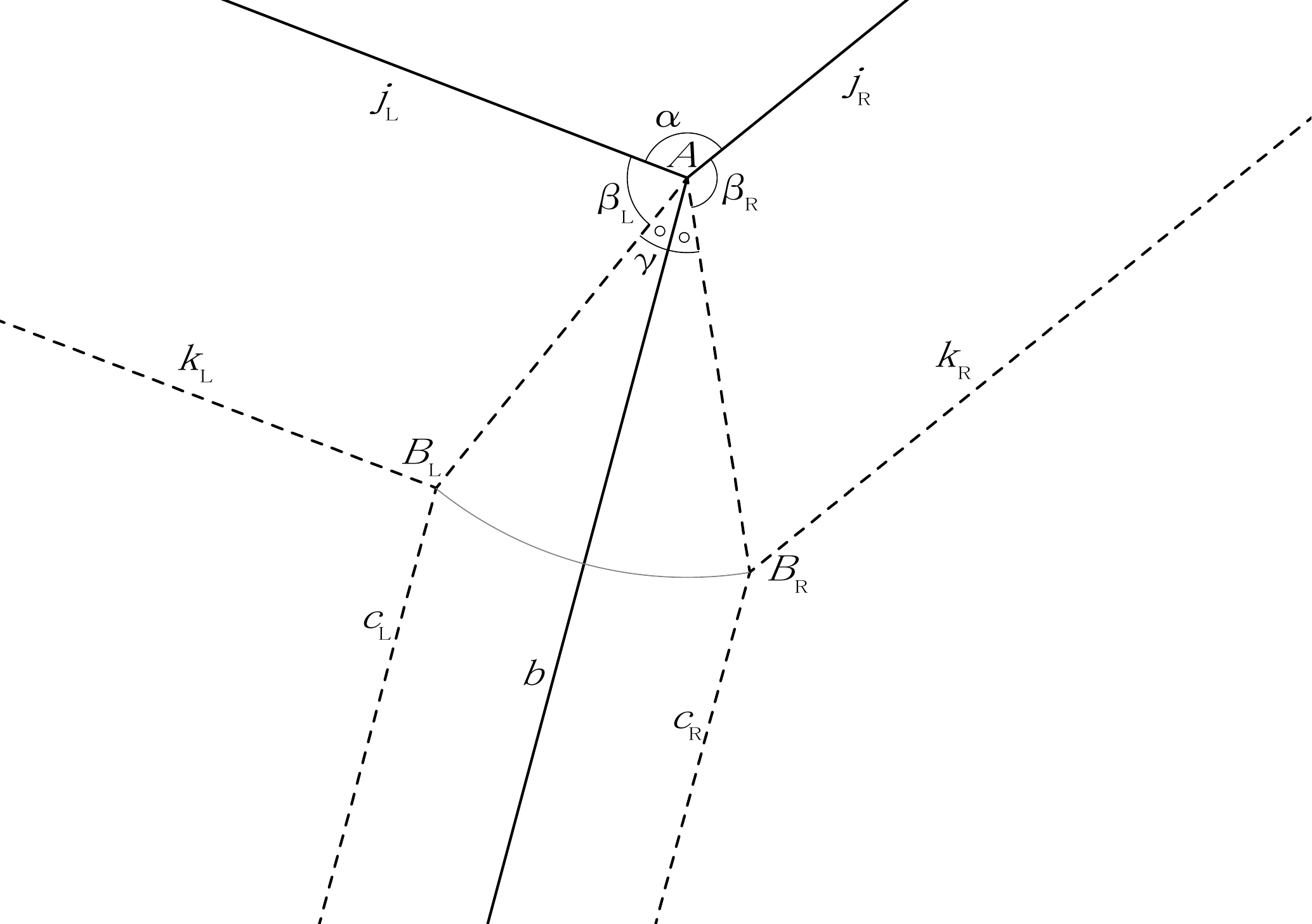}
    \caption{Precreasing for a negative $3$D gadget with one pleat}
    \label{fig:precreasing_negative_onepleat}
\end{figure}
\begin{construction}\label{const:negative_onepleat}\rm
Consider a net as in Figure $\ref{fig:development_1}$, for which we conditions (i), (ii) and (iii) of Construction $\ref{const:condition}$.
Then the crease pattern of the negative $3$D gadget with one simple outgoing pleat is constructed as follows.
\begin{enumerate}[(1)]
\item Draw a ray $b$ starting from $A$ which bisects $\gamma$.
\item Choose $\tau =\Lt$ or $\tau =\Rt$ and extend $k_\tau$ so that $k_\tau$ starts from a point $B'_\tau$ in $b$.
\item Draw a ray $c_{\tau'}$ starting from $B_{\tau'}$, parallel to $b$ and going in the same direction as $b$.
\item The desired crease pattern is shown as the solid lines in Figure $\ref{fig:CP_negative_onepleat}$,
and the assignment of mountain folds and valley folds is given in Table $\ref{tbl:assignment_negative_onepleat}$.
\end{enumerate}
\end{construction}
\begin{figure}[htbp]
\addtocounter{theorem}{1}
\centering
\includegraphics[width=0.48\hsize]{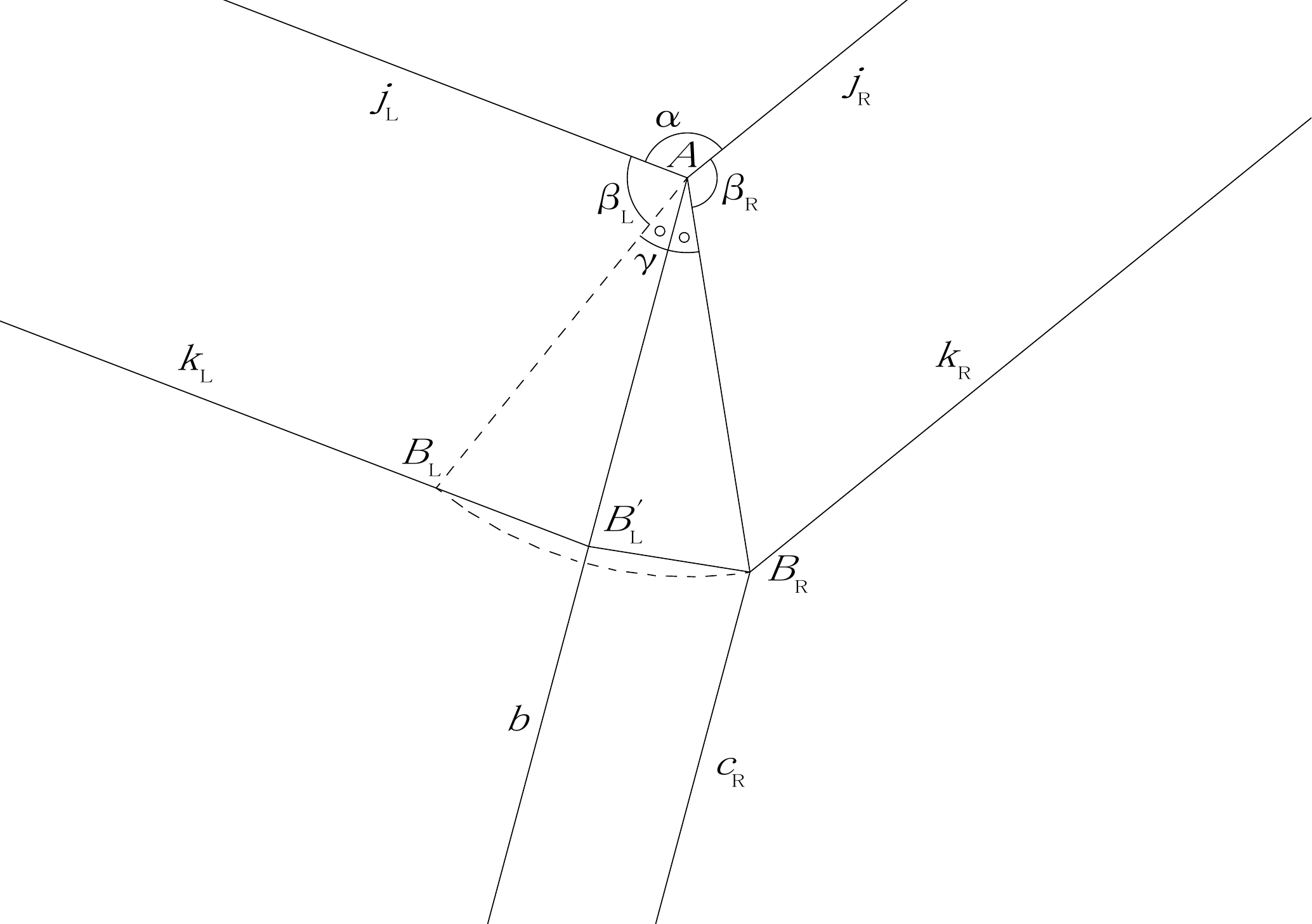}
\includegraphics[width=0.48\hsize]{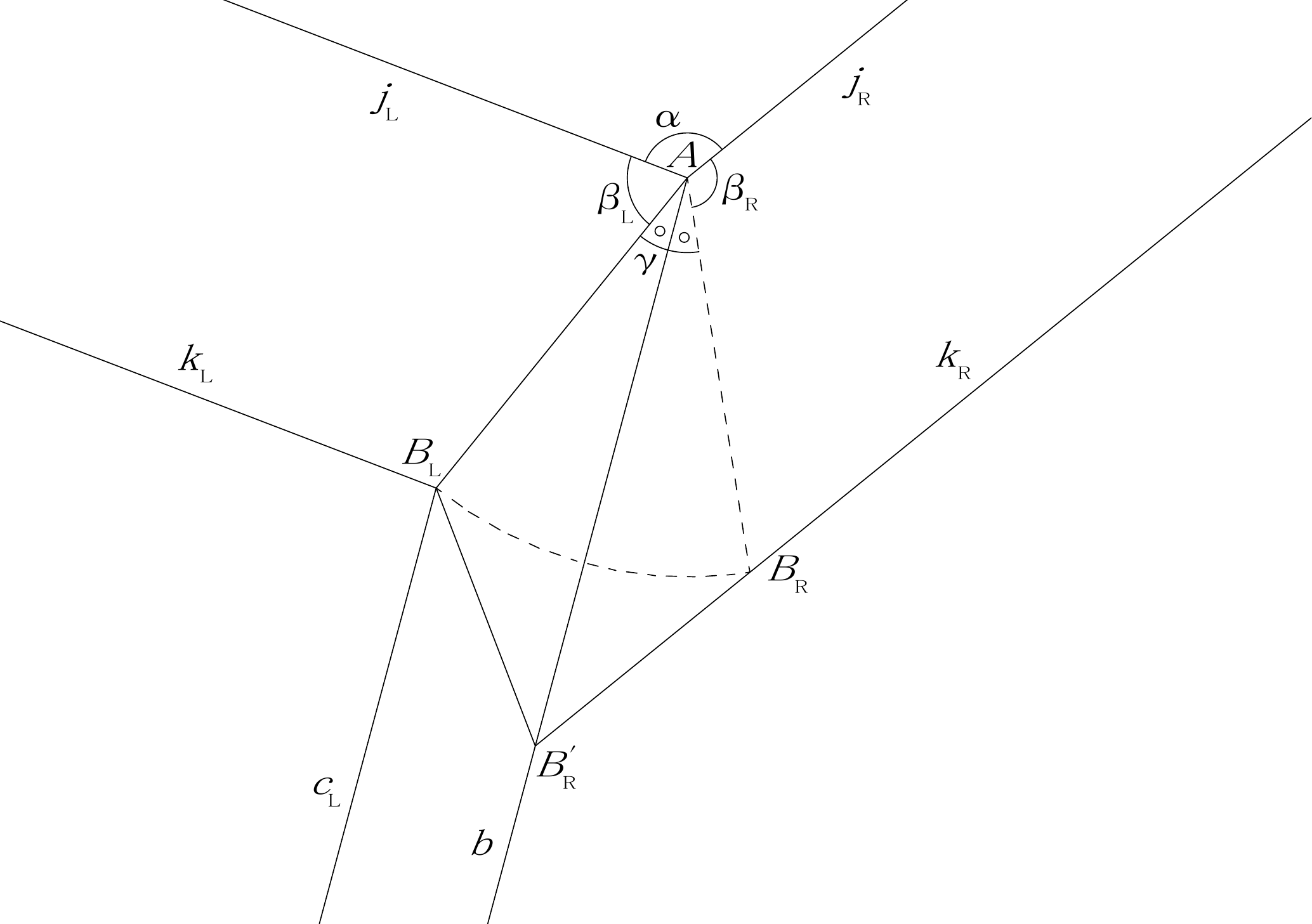}
    \caption{CPs of two possible negative $3$D gadgets with one simple outgoing pleat}
    \label{fig:CP_negative_onepleat}
\end{figure}
\renewcommand{\arraystretch}{1.5}
\addtocounter{theorem}{1}
\begin{table}[h]
\begin{tabular}{c|c}
mountain folds&$j_\Lt ,j_\Rt ,b,B'_\tau B_{\tau'}$\\ \hline
valley folds&$k_\Lt ,k_\Rt ,c_{\tau'},AB_{\tau'}$
\end{tabular}\vspace{0.5cm}
\caption{Assignment of mountain folds and valley folds to a negative $3$D gadget with one simple outgoing pleat}
\label{tbl:assignment_negative_onepleat}
\end{table}
Therefore we obtain two crease patterns according to the choice of $\tau$ in procedure $(2)$.
Note that for the choice of $\tau$ in $(2)$, the extended length of $k_\tau$, which may cause an interference on the side of $\tau'$, is calculated as
\begin{equation*}
\norm{B_\Lt B'_\tau}=\norm{B_\Rt B'_\tau}=\frac{\sin (\gamma /2)}{\sin (\beta_\tau +\gamma /2)}\cdot\norm{AB}.
\end{equation*}
Note also that there is no interference on the side of $\tau'$.

\subsection{Cheng's construction with prescribed outgoing pleats $(\ell'_\Lt ,m_\Lt )$ and $(\ell'_\Rt ,m_\Rt )$}\label{subsec:negative_Cheng}
This construction is originally given in \cite{Cheng}, but only for $\delta_\Lt =\delta_\Rt =0$. 
Here we extend the result for general $\delta_\Lt ,\delta_\Rt\geqslant 0$ satisfying conditions (iii.a)--(iii.c) of Construction $\ref{const:condition}$.
\begin{figure}[htbp]
\addtocounter{theorem}{1}
\centering\includegraphics[width=0.75\hsize]{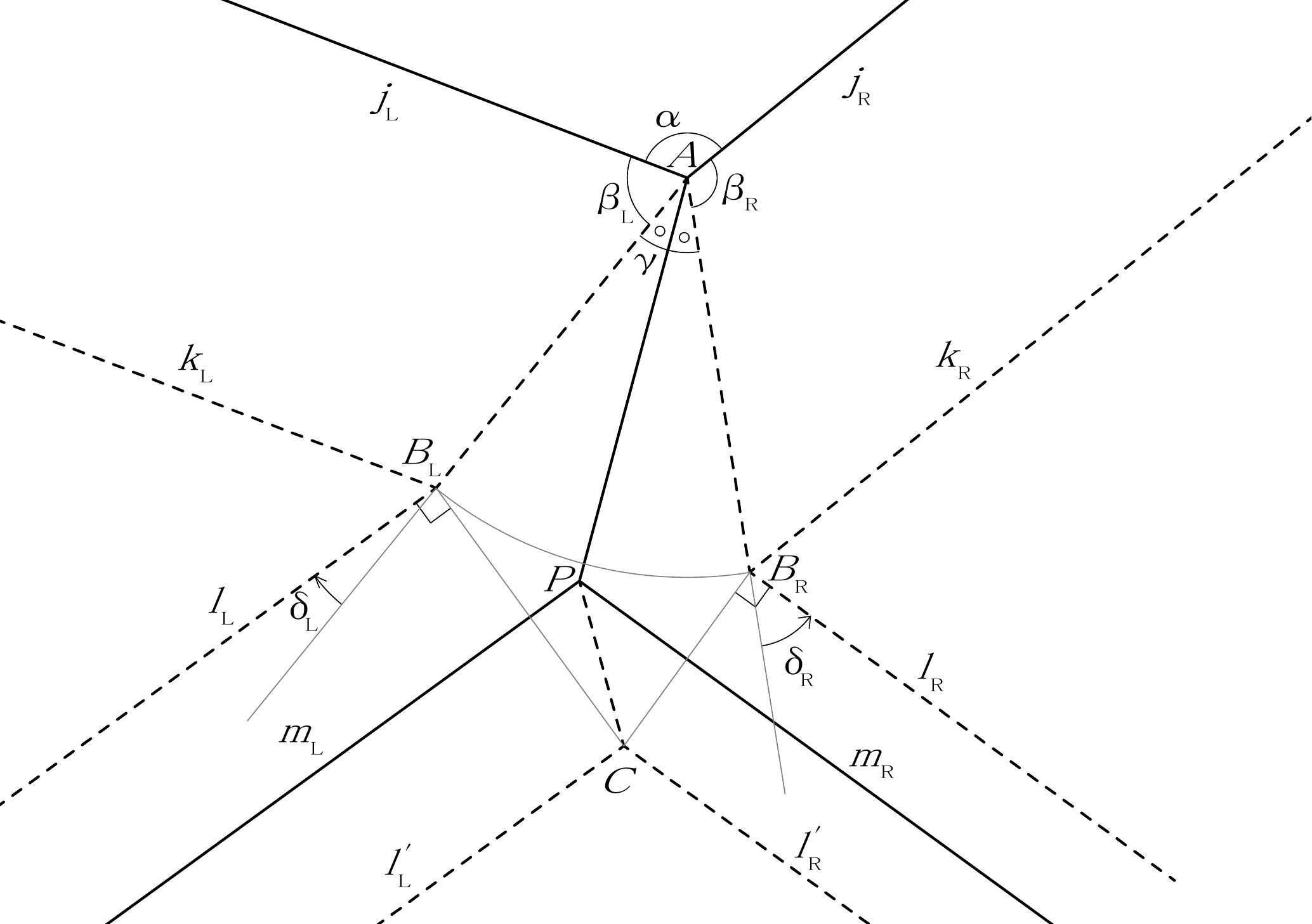}
    \caption{Precreasing for a negative $3$D gadget by Cheng's construction}
    \label{fig:precreasing_negative_Cheng}
\end{figure}
For a development as in Figure $\ref{fig:development_1}$, precreasing as in Figure $\ref{fig:precreasing_negative_Cheng}$ and swinging the resulting the flap
formed by $B_\Lt AP$ and $B_\Rt AP$ to either side face and at the same time the two simple pleats to the side of $\ell'_\Lt$ and $\ell'_\Rt$
until the flap extends the other side face and the two pleats overlap with the ambient paper, we obtain the following construction.
\begin{construction}\label{const:negative_Cheng}\rm
Consider a development as in Figure $\ref{fig:development_1}$, 
for which we require condition (i),(ii) and (iii.a)--(iii.c) of Construction $\ref{const:condition}$.
Then the crease pattern of the negative $3$D gadget with prescribed simple outgoing pleats $(\ell'_\Lt ,m_\Lt )$ and $(\ell'_\Rt ,m_\Rt )$
is constructed as follows.
\begin{enumerate}[(1)]
\item Choose $\tau =\Lt$ or $\tau =\Rt$ and extend $k_\tau$ so that $k_\tau$ starts from a point $B'_\tau$ in polygonal chain $APm_\tau$.
\item If $B'_\tau$ lies in $m_\tau\setminus\{ P\}$, or equivalently, if $\beta_\tau >\gamma /2+\delta_\Lt +\delta_\Rt$, then
draw a line which is a reflection of $B'_\tau C$ across segment $PC$, letting $B'_{\tau'}$ be the intersection point of the line and $m_{\tau'}$.
If $B'_\tau$ lies in $AP$, or equivalently, if $\beta_\tau\leqslant\gamma /2+\delta_\Lt +\delta_\Rt$, then we set $B'_{\tau'}=B'_\tau$.
\item Draw a ray parallel to $AP$ from $B_{\tau'}$ and let $Q_{\tau'}$ be the intersection point of the ray and $m_{\tau'}$.
\item The desired crease pattern is shown as the solid lines in Figure $\ref{fig:CP_negative_Cheng_L}$ for $\beta_\tau\leqslant\gamma /2+\delta_\Lt +\delta_\Rt$
and Figure $\ref{fig:CP_negative_Cheng_R}$ for $\beta_\tau >\gamma /2+\delta_\Lt +\delta_\Rt$,
and the assignment of mountain and valley folds is given in Table $\ref{tbl:assignment_negative_Cheng}$.
\end{enumerate}
\end{construction}
\begin{figure}[htbp]
\addtocounter{theorem}{1}
\centering\includegraphics[width=0.75\hsize]{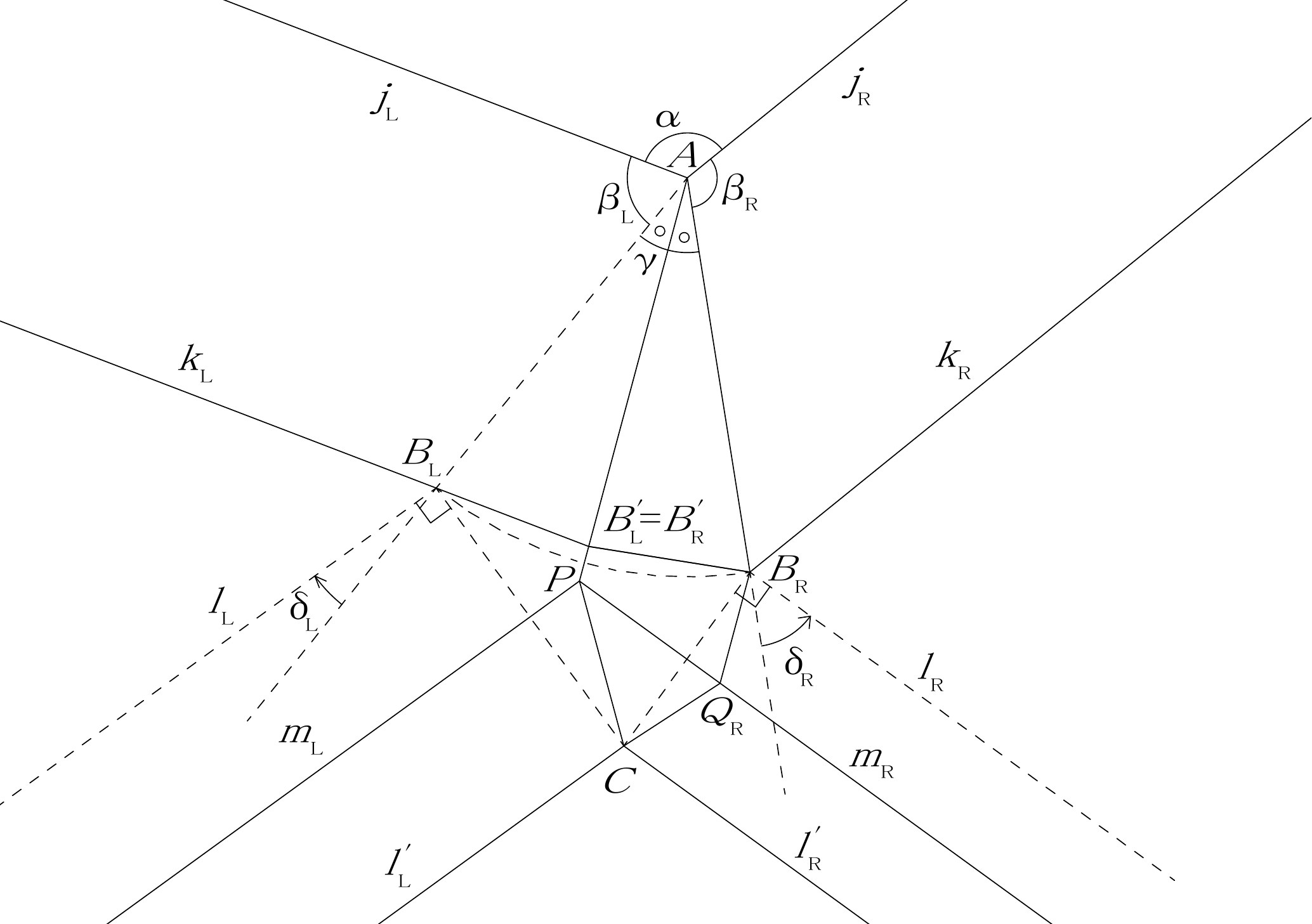}
    \caption{CP of a negative $3$D gadget by Cheng's construction for $\tau =\Lt$ with $\beta_\tau\leqslant\gamma /2+\delta_\Lt +\delta_\Rt$}
    \label{fig:CP_negative_Cheng_L}
\end{figure}
\begin{figure}[htbp]
\addtocounter{theorem}{1}
\centering\includegraphics[width=0.75\hsize]{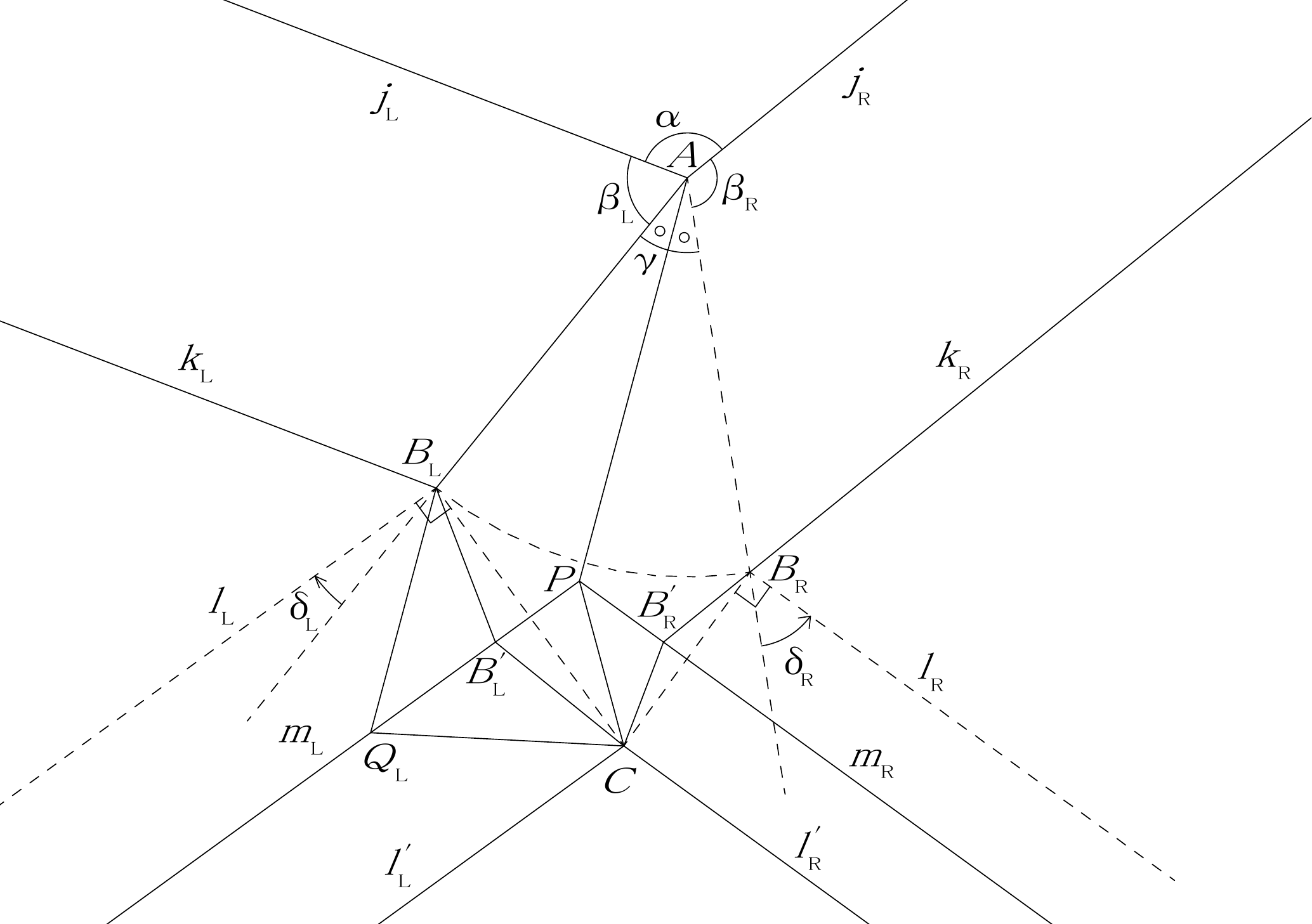}
    \caption{CP of a negative $3$D gadget by Cheng's construction for $\tau =\Rt$ with $\beta_\tau >\gamma /2+\delta_\Lt +\delta_\Rt$}
    \label{fig:CP_negative_Cheng_R}
\end{figure}
\addtocounter{theorem}{1}
\begin{table}[h]
\begin{tabular}{c|c|c}
&common creases&$\beta_\tau >\gamma /2+\delta_\Lt +\delta_\Rt$\\ \hline
mountain folds&$j_\Lt ,j_\Rt ,m_\Lt ,m_\Rt ,AP,B_{\tau'}B'_{\tau'},CQ_{\tau'}$&$B'_\tau C$\\ \hline
valley folds&$k_\Lt ,k_\Rt ,\ell'_\Lt ,\ell'_\Rt ,CP,B_{\tau'}Q_{\tau'}$&$B'_{\tau'}C$
\end{tabular}\vspace{0.5cm}
\caption{Assignment of mountain and valley folds to a negative $3$D gadget with prescribed outgoing pleats $(\ell'_\Lt ,m_\Lt )$ and $(\ell'_\Rt ,m_\Rt )$}
\label{tbl:assignment_negative_Cheng}
\end{table}
\begin{remark}\rm
In Cheng's original construction \cite{Cheng}, where $\delta_\Lt =\delta_\Rt =0$, 
$Q_{\tau'}\in m_{\tau'}$ in $(3)$ was determined so that $\norm{PC}=\norm{CQ_{\tau'}}(=\norm{B_{\tau'}Q_{\tau'}})$.
This is justified for $\delta_\Lt =\delta_\Rt =0$ because $AP$ and $PC$ form a straight line,
so that $PC$ and $B_{\tau'}Q_{\tau'}$ are parallel to each other with the same length.
Also note that if $\delta_\Lt =\delta_\Rt =0$, then $\beta_\tau +\gamma /2<\pi$ always holds, and thus the case $B'_\tau\in AP$ in $(2)$ does not occur.
\end{remark}
\begin{proposition}
For our choice of $\tau$ in $(2)$, we have $\angle AB_{\tau'}B'_{\tau'}=\beta_\tau$.
\end{proposition}
\begin{proof}
If $B'_\tau\in AP$, then $\triangle AB_\tau B'_\tau$ and $\triangle AB_{\tau'}B'_{\tau'}$ are congruent by Proposition $\ref{prop:AP}$, 
where $B'_\tau =B'_{\tau'}$.
Thus we have $\angle AB_{\tau'}B'_{\tau'}=\angle AB_\tau B'_\tau =\beta_\tau$.
Also, if $B'_\tau\in m_\tau\setminus\{ P\}$, then we have
\begin{equation*}
\angle PB_\tau B'_\tau =\angle PCB'_\tau =\angle PCB'_{\tau'}=\angle PB_{\tau'}B'_{\tau'},
\end{equation*}
so that by Proposition $\ref{prop:AP}$ again, we have
\begin{equation*}
\angle AB_{\tau'}B'_{\tau'}=\angle AB_{\tau'}P+\angle PB_{\tau'}B'_{\tau'}=\angle AB_\tau P+\angle PB_\tau B'_\tau =\beta_\tau .
\end{equation*}
This completes the proof.
\end{proof}
For the choice of $\tau$ in $(2)$, the extended length of the base $k_\tau$ of the side face, which may cause an interference on the side of $\tau'$, is given by
$\max\{\norm{B_\Lt B'_\Lt},\norm{B_\Rt ,B'_\Rt}\}$, where $\norm{B_\sigma B'_\sigma}$ for $\sigma =\Lt ,\Rt$ are calculated as
\begin{equation*}
\norm{B_\sigma B'_\sigma}=
\begin{dcases}
\frac{\sin (\gamma /2)}{\sin (\beta_\tau +\gamma /2)}\cdot\norm{AB}
&\text{if }\beta_\tau\leqslant\frac{\gamma}{2}+\delta_\Lt +\delta_\Rt ,\\
\frac{\cos\delta_{\sigma'}-\cos (\gamma +\delta_{\sigma'})}{2\sin (\gamma +\delta_\Lt +\delta_\Rt )\sin (\beta_\tau -\delta_\sigma )}\cdot\norm{AB}
&\text{if }\beta_\tau\geqslant\frac{\gamma}{2}+\delta_\Lt +\delta_\Rt .
\end{dcases}
\end{equation*}
Note that the extended length in the first case is the same as Construction $\ref{const:negative_onepleat}$.
Note also that there is no interference on the side of $\tau$.

Now let us check the flat-foldability around $P$ and $C$.
Around $P$, noting that
\begin{equation*}
\angle APm_\sigma -\angle CPm_\sigma=\angle APB_\sigma\quad\text{for }\sigma =\Lt ,\Rt ,
\end{equation*}
we have the flat-foldability condition
\begin{equation*}
\angle APm_\Rt -\angle CPm_\Rt +\angle CPm_\Lt -\angle APm_\Lt =\angle APB_\Rt -\angle APB_\Lt =0
\end{equation*}
as desired.

To prove the flat-foldability around $C$, first suppose $B'_\tau\in m_\tau\setminus\{P \}$.
Then we calculate as 
\begin{align*}
\angle PCB'_\tau &=\angle PB_\tau B'_\tau =\angle AB_\tau B'_\tau -\angle AB_\tau P=\beta_\tau -\gamma /2-\delta_\Lt -\delta_\Rt ,\\
\angle B'_{\tau'}CQ_{\tau'}&=\angle B'_{\tau'}B_{\tau'}Q_{\tau'}=\angle AB_{\tau'}Q_{\tau'}-\angle AB_{\tau'}B'_{\tau'}
=\pi -\gamma /2-\beta_\tau ,\\
\angle \ell'_\Lt C\ell'_\Rt&=\gamma +\delta_\Lt +\delta_\Rt ,
\end{align*}
so that we have
\begin{equation}\label{eq:flat-foldability_C_1}
\angle PCB'_\tau +\angle B'_{\tau'}CQ_{\tau'}+\angle \ell'_\Lt C\ell'_\Rt =\pi =\angle PCB'_{\tau'}+\angle Q_{\tau'}C\ell'_{\tau'}+\angle B'_\tau C\ell'_\tau .
\end{equation}

Next suppose $B'_\tau\in AP$. 
Then we have similarly
\begin{align*}
\angle PCQ_{\tau'}&=\angle PB_{\tau'}Q_{\tau'}=\angle AB_{\tau'}Q_{\tau'}-\angle AB_{\tau'}P=\pi -\gamma -\delta_\Lt -\delta_\Rt ,\\
\angle\ell'_\Lt C\ell'_\Rt&=\gamma +\delta_\Lt +\delta_\Rt ,
\end{align*}
so that
\begin{equation}\label{eq:flat-foldability_C_2}
\angle PCQ_{\tau'}+\angle\ell'_\Lt C\ell'_\Rt =\pi =\angle Q_{\tau'}C\ell'_{\tau'}+\angle PC\ell'_\tau .
\end{equation}
Hence $\eqref{eq:flat-foldability_C_1}$ and $\eqref{eq:flat-foldability_C_2}$ yield the flat-foldability around $C$.

\section{Constructions of negative $3$D gadgets with prescribed outgoing pleats $(\ell_\Lt ,m_\Lt )$ and $(\ell_\Rt ,m_\Rt )$}\label{sec:negative_new}
In this section we shall give three constructions of negative $3$D gadgets,
which all use the prescribed two simple pleats $(\ell_\Lt ,m_\Lt )$ and $(\ell_\Rt ,m_\Rt )$.
In particular, we will prove that our first and third constructions solve Problem $\ref{prob:existence_negative}$.
We will also discuss in Section $\ref{subsec:interferences}$ the interferences between two adjacent $3$D gadgets at least one of which is negative.

\subsection{First Construction}\label{subsec:negative_new_1}
This is an extension of our first construction of negative $3$D gadgets given in \cite{Doi19}, Section $9$, and
similar to Constructions $\ref{const:negative_onepleat}$ and $\ref{const:negative_Cheng}$ in that
we swing the supporting triangle to either side face until it overlaps with the other.
We can construct an infinite number of negative $3$D gadgets with the same outgoing pleats for each $\tau\in\{\Lt ,\Rt\}$.
\begin{construction}\label{const:negative_1}\rm
Consider a development as in Figure $\ref{fig:development_1}$, for which we require condition (i), (ii), (iii.a)--(iii.c) of Construction $\ref{const:condition}$.
Then the crease pattern of our first negative $3$D gadget with prescribed simple outgoing pleats $(\ell_\Lt ,m_\Lt )$ and $(\ell_\Rt ,m_\Rt )$
is constructed as follows.
\begin{enumerate}[(1)]
\item Choose $\tau =\Lt$ or $\tau =\Rt$.
Also choose a point $E_\tau$ in $m_\tau$ so that $\angle AB_\tau E_\tau$ satisfies
\begin{equation}\label{ineq:range_ABE_1}
\angle AB_\tau E_\tau\in\left[\pi -\beta_\tau ,\pi -\beta_\tau +\frac{\gamma}{2}+\delta_\Lt +\delta_\Rt\right]\cap (\gamma +\delta_\Lt +\delta_\Rt ,\pi ).
\end{equation}
(For a practical choice of $\angle AB_\tau E_\tau$, see Remark $\ref{rem:ABE_practical}$.)
Also redefine $m_\tau$ to be a ray starting from $E_\tau$ and going in the same direction as ray $\ell_\tau$.
\item Let $D$ be the point in the circular arc $B_\Lt B_\Rt$ with center $A$ such that $\phi_\tau =\angle B_\tau AD=2\angle B_\tau AE_\tau$.
\item Let $G_\tau$ be the point in segment $AE_\tau$ such that $\angle AB_\tau G_\tau =\pi -\beta_\tau$.
If we choose $\angle AB_\tau E_\tau =\pi -\beta_\tau$ in $\eqref{ineq:range_ABE_1}$, then we have $G_\tau =E_\tau$.
\item Let $E_{\tau'}$ be the intersection point of $m_{\tau'}$ and a bisector of $\angle DAB_{\tau'}$.
Also redefine $m_{\tau'}$ to be a ray starting from $E_{\tau'}$ and going in the same direction as ray $\ell_{\tau'}$.
\item Let $G_{\tau'}$ be the intersection point of a ray starting from $G_\tau$ through $D$ and segment $AE_{\tau'}$.
If we choose $\angle AB_\tau E_\tau =\pi -\beta_\tau +\gamma /2+\delta_\Lt +\delta_\Rt$ in $\eqref{ineq:range_ABE_1}$, then we have $G_{\tau'}=E_{\tau'}$.
\item Let $P_\tau$ be a point in $m_\tau$ such that $\angle E_\tau B_\tau P_\tau =\delta_\tau$.
If $\delta_\tau =0$, then we have $P_\tau =E_\tau$.
\item Let $P_{\tau'}$ be the intersection point of a ray starting from $P_\tau$ through $C$ and ray $m_{\tau'}$.
\item If $\angle AB_\tau E_\tau <\pi -\beta_\tau +\delta_\tau$,
then determine a point $Q_\tau$ in $m_\tau$ such that $\angle Q_\tau B_\tau\ell_\tau =\beta_\tau -\delta_\tau$.
Also, determine a point $R_\tau$ in segment $E_\tau P_\tau$ such that $\angle R_\tau B_\tau E_\tau =\angle Q_\tau B_\tau E_\tau$.
(In particular, if $\angle AB_\tau E_\tau =\pi -\beta_\tau$, then we have $E_\tau =G_\tau =R_\tau$.
For the existence of $Q_\tau$ and $R_\tau$, see Proposition $\ref{prop:existence_QR}$.)
Then draw segments $B_\tau Q_\tau ,B_\tau R_\tau ,CQ_\tau$ and $CR_\tau$.
\item If $\angle AB_\tau E_\tau >\beta_{\tau'}+\gamma +\delta_\tau$,
then determine a point $Q_{\tau'}$ in $\ell'_{\tau'}$ such that $\angle Q_{\tau'}B_{\tau'}\ell'_{\tau'}=\beta_{\tau'}-\delta_{\tau'}$.
Also, determine a point $R_{\tau'} (\neq E_{\tau'})$ in segment $E_{\tau'}P_{\tau'}$ such that
$\angle R_{\tau'}B_{\tau'}E_{\tau'}=\angle Q_{\tau'}B_{\tau'}E_{\tau'}$.
(For the existence of $Q_{\tau'}$ and $R_{\tau'}$, see Proposition $\ref{prop:existence_QR}$.)
Then draw segments $B_{\tau'}Q_{\tau'},B_{\tau'}R_{\tau'},CQ_{\tau'}$ and $CR_{\tau'}$.
\item If $\angle AB_\tau E_\tau\in [\pi -\beta_\tau +\delta_\tau ,\beta_{\tau'}+\gamma +\delta_\tau ]$, 
then the desired crease pattern is shown as the solid lines in Figure $\ref{fig:CP_negative_1}$.
If not, then the desired crease pattern is shown as the solid lines in Figure $\ref{fig:CP_negative_1_additional}$.
The assignment of mountain and valley folds to the additional creases is given in Table $\ref{tbl:assignment_negative_1}$,
where $\sigma$ takes both $\Lt$ and $\Rt$.
\end{enumerate}
\end{construction}
\begin{figure}[htbp]
\addtocounter{theorem}{1}
\centering\includegraphics[width=0.75\hsize]{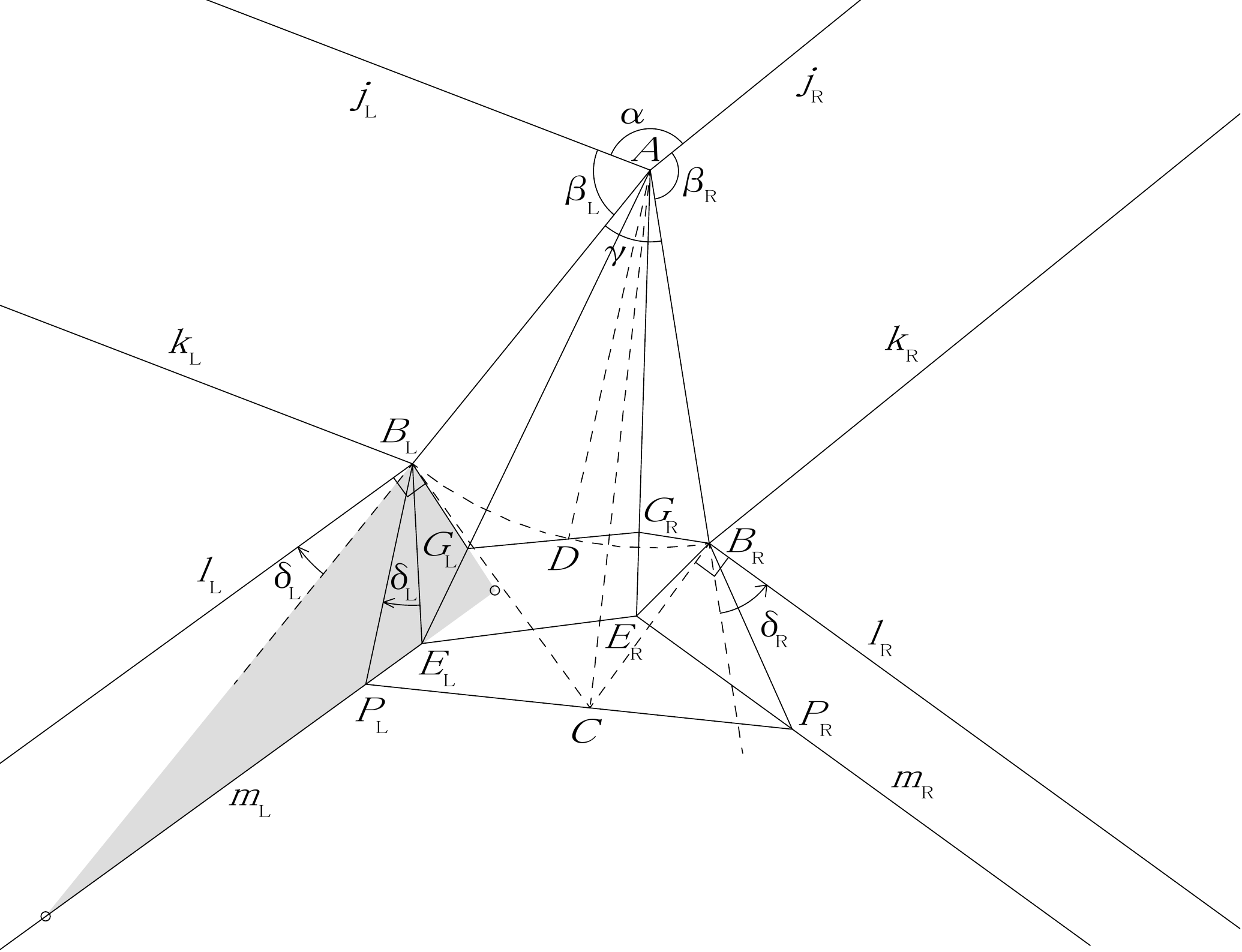}
    \caption{CP of a negative $3$D gadget by our first construction for $\tau =\Lt$ and $\angle AB_\tau E_\tau\geqslant\pi -\beta_\tau +\delta_\tau$, 
    where the possible range of $B_\tau E_\tau$ is shaded}
    \label{fig:CP_negative_1}
\end{figure}
\begin{figure}[htbp]
\addtocounter{theorem}{1}
\centering\includegraphics[width=0.75\hsize]{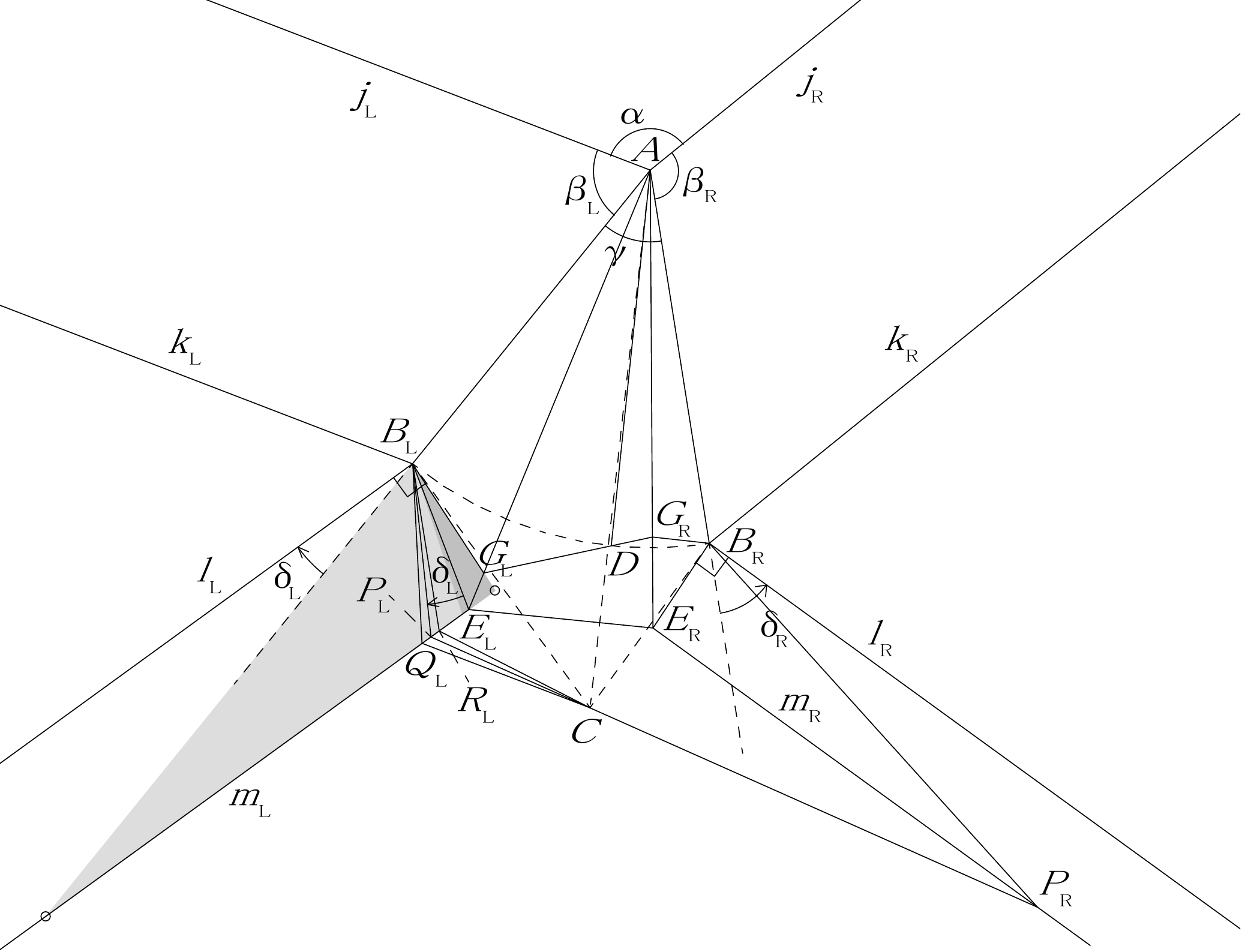}
    \caption{CP of a negative $3$D gadget by our first construction for $\tau =\Lt$ and $\angle AB_\tau E_\tau <\pi -\beta_\tau +\delta_\tau$,
    where the possible range of $B_\tau E_\tau$ with $\angle AB_\tau E_\tau <\pi -\beta_\tau +\delta_\tau$ is dark shaded}
    \label{fig:CP_negative_1_additional}
\end{figure}
\addtocounter{theorem}{1}
\begin{table}[h]
\begin{tabular}{c|c}
&common creases\\ \hline
mountain folds&$j_\sigma ,m_\sigma ,AE_\sigma ,E_\Lt E_\Rt$\\ \hline
valley folds&$k_\sigma ,\ell_\sigma ,AB_\sigma ,G_\Lt G_\Rt ,P_\Lt ,P_\Rt$
\end{tabular}\\
\vspace{0.5cm}
\begin{tabular}{c|c|c|c|c}
&\multicolumn{4}{c}{creases on the side of $\tau ;\quad m_1=\pi -\beta_\tau ,\; m_2=\pi -\beta_\tau +\delta_\tau$}\\
\cline{2-5}
&$\delta_\tau =0$&\multicolumn{3}{c}{$\delta_\tau >0\; (E_\tau\neq P_\tau )$}\\
\cline{3-5}
&$(E_\tau =P_\tau)$&$\angle AB_\tau E_\tau =m_1$&$m_1<\angle AB_\tau E_\tau <m_2$&$\angle AB_\tau E_\tau\geqslant m_2$\\
&&$(E_\tau =G_\tau =R_\tau )$&\multicolumn{2}{c}{$(E_\tau\neq G_\tau )$}\\
\hline
\multirow{3}{*}{\bettershortstack{mountain\\ folds}}&\multirow{3}{*}{$B_\tau G_\tau$}&---&\multicolumn{2}{c}{$B_\tau G_\tau$}\\
\cline{3-5}
&&\multicolumn{3}{c}{$B_\tau P_\tau$}\\
\cline{3-5}
&&\multicolumn{2}{c|}{$B_\tau Q_\tau ,CR_\tau$}&---\\
\hline
\multirow{3}{*}{\bettershortstack{valley\\ folds}}&&\multicolumn{3}{c}{$B_\tau E_\tau$}\\
\cline{3-5}
&\multicolumn{2}{c|}{---}&$B_\tau R_\tau$&\multirow{2}{*}{---}\\
\cline{3-4}
&&\multicolumn{2}{c|}{$CQ_\tau$}&
\end{tabular}\\
\vspace{0.5cm}
\begin{tabular}{c|rc|c|c|c}
&\multicolumn{5}{c}{\bettershortstack{creases on the side of $\tau'; \quad M_1=\beta_{\tau'}+\gamma +\delta_\tau$,\\
\hfill$M_2=\pi -\beta_\tau +\gamma /2+\delta_\Lt +\delta_\Rt$}}\\
\cline{2-6}
&\multicolumn{3}{c|}{$\qquad\angle AB_\tau E_\tau\leqslant M_1\qquad$}&\multicolumn{2}{c}{$\angle AB_\tau E_\tau >M_1$}\\
\cline{2-6}
&\multicolumn{2}{c|}{$\angle AB_\tau E_\tau =M_2$}&\multicolumn{2}{c|}{$\angle AB_\tau E_\tau \neq M_2$}&$\angle AB_\tau E_\tau =M_2$\\
&\multicolumn{2}{c|}{$(E_{\tau'}=G_{\tau'})$}&\multicolumn{2}{c|}{$(E_{\tau'}\neq G_{\tau'})$}&$(E_{\tau'}=G_{\tau'})$\\
\hline
\multirow{3}{*}{\bettershortstack{mountain\\ folds}}&\multicolumn{5}{c}{$B_{\tau'}P_{\tau'}$}\\
\cline{2-6}
&\multirow{2}{*}{$\qquad\qquad$---}&$\quad$
&\multicolumn{2}{c|}{$B_{\tau'}G_{\tau'}$}&\multicolumn{1}{l}{\multirow{2}{*}{\bettershortstack{$B_{\tau'}Q_{\tau'},$\\ $CR_{\tau'}$}}}\\
\cline{4-5}
&\multicolumn{3}{c|}{}&\multicolumn{2}{c}{}\\
\hline
valley&\multicolumn{5}{c}{$B_{\tau'}E_{\tau'}$}\\
\cline{2-6}
folds&\multicolumn{3}{c|}{---}&\multicolumn{2}{c}{$B_{\tau'}R_{\tau'},CQ_{\tau'}$}
\end{tabular}\\
\vspace{0.5cm}
\caption{Assignment of mountain and valley folds to a negative $3$D gadget by our first construction}
\label{tbl:assignment_negative_1}
\end{table}
The following result is immediate.
\begin{lemma}\label{lem:ADG'}
For the choice of $\tau$ in $(1)$, we have $\angle ADG_{\tau'}=\angle AB_{\tau'}G_{\tau'}=\beta_\tau$.
\end{lemma}
Thus for the choice of $\tau$ in $(1)$ and also $\phi_\tau$ in $(2)$, the extended length of the base $k_\tau$ of the side face,
which may cause an interference on the side of $\tau'$, is given by $\norm{DG_{\tau'}}$, which is calculated as
\begin{equation*}
\norm{DG_{\tau'}}=\norm{B_{\tau'}G_{\tau'}}=\frac{\norm{AB}}{\sin\beta_\tau /\tan (\phi_{\tau'}/2)+\cos\beta_\tau}.
\end{equation*}
On the other hand, an interference occurs on the side of $\tau$ within the length $\norm{D G_\tau}$ given by
\begin{equation*}
\norm{DG_\tau}=\norm{B_\tau G_\tau}=\frac{\norm{AB}}{\sin\beta_\tau /\tan (\phi_\tau /2)-\cos\beta_\tau}. 
\end{equation*}
The above constuction of negative $3$D gadgets leads to the following two results,
where the validity of the construction and the foldability of the crease patterns are verified below.
\begin{theorem}\label{thm:negative_1_engaging}
For any positive $3$D gadget with $(\alpha ,\beta_\Lt ,\beta_\Rt ,\delta_\Lt ,\delta_\Rt )$
satisfying conditions $(\mathrm{i})$, $(\mathrm{ii})$ and $(\mathrm{iii.a})$--$(\mathrm{iii.c})$ of Construction $\ref{const:condition}$, 
we can apply Construction $\ref{const:negative_1}$ to obtain an infinite number of negative $3$D gadgets for each $\tau\in\{\Lt ,\Rt\}$
with the same net and prescribed outgoing pleats as the positive one in the development.
All of these negative $3$D gadgets engage with the original positive one.
\end{theorem}
Also, if we consider the horizontally flipped crease pattern of each negative $3$D gadget
resulting in Theorem $\ref{thm:negative_1_engaging}$, we have the following result, which solves Problem $\ref{prob:existence_negative}$.
\begin{corollary}
For any positive $3$D gadget with $(\alpha ,\beta_\Lt ,\beta_\Rt ,\delta_\Lt ,\delta_\Rt )$
satisfying conditions $(\mathrm{i})$, $(\mathrm{ii})$ and $(\mathrm{iii.a})$--$(\mathrm{iii.c})$ of Construction $\ref{const:condition}$,
we can apply Construction $\ref{const:negative_3}$ to obtain an infinite number of $3$D gadgets for each $\tau\in\{\Lt ,\Rt\}$
with the same outgoing pleats such that the top and two side faces are negatively extruded, or in other words, sunk to the opposite direction.
\end{corollary}

Now let us see where condition $\eqref{ineq:range_ABE_1}$ in $(1)$ of Construction $\ref{const:negative_1}$ comes from.
We will follow the discussion and terminology in \cite{Doi20}, Section $4$.
In particular, we set $\psi_\sigma =\angle B_\sigma AC-\angle B_\sigma AD$ and $\rho_\sigma =\angle ACE_\sigma -\angle DCE_\sigma$ for $\sigma =\Lt ,\Rt$,
so that we have $\gamma_\Lt +\gamma_\Rt =\gamma$ and $\psi_\Lt +\psi_\Rt =\rho_\Lt +\rho_\Rt =0$.
For simplicity, we denote below $B_\tau ,k_\tau ,\beta_\tau ,\psi_\tau$, etc. by $B,k,\beta ,\psi$, etc.,
and $B_{\tau'},k_{\tau'},\beta_{\tau'},\psi_{\tau'}$, etc. by $B',k',\beta',\psi'$, etc. 
with an exception that we denote $\gamma =\gamma_\Lt +\gamma_\Rt$ by $\gamma_+$ to avoid confusion.
Also, we denote $\delta_\Lt +\delta_\Rt$ by $\delta_+$.

Firstly, for the existence of point $G$, we need $\angle ABG\leqslant\angle ABE$, which is equivalent to
\begin{equation}\label{ineq:existence_G}
\frac{\pi}{2}-\beta\leqslant\frac{\gamma}{2}+\frac{\psi}{2}+\rho +\delta
\end{equation}
by \cite{Doi20}, Lemma $4.4$.
Secondly, it follows from $-\gamma'<\psi <\gamma$ and $\gamma'+\delta'-\pi /2<\rho <\pi /2 -\gamma -\delta$ that
\begin{equation}\label{ineq:existence_E}
\frac{\gamma_+}{2}+\delta'-\frac{\pi}{2}<\frac{\gamma}{2}+\frac{\psi}{2}+\rho <\frac{\pi}{2}-\delta .
\end{equation}
Thirdly, for the existence of point $G'$, we need
\begin{equation*}
\angle ADG'\leqslant\angle ADE'.
\end{equation*}
Then from Lemma $\ref{lem:ADG'}$ and
\begin{equation*}
\angle ADE'=\angle AB'E'=\frac{\pi}{2}+\frac{\gamma'}{2}+\frac{\psi'}{2}+\rho'+\delta'
=\frac{\pi}{2}+\frac{\gamma_+}{2}+\delta'-\left(\frac{\gamma}{2}+\frac{\psi}{2}+\rho\right) ,
\end{equation*}
we see that
\begin{equation}\label{ineq:existence_G'}
\beta +\frac{\gamma}{2}+\frac{\psi}{2}+\rho\leqslant\frac{\pi}{2}+\frac{\gamma_+}{2}+\delta'.
\end{equation}
Fourthly, for the existence of point $P'$, we need
\begin{equation}\label{ineq:existence_P_1}
\begin{dcases}
\angle CEE'+\angle EE'P'<\pi&\text{if }\delta =0,\text{ so that }P=E,\\
\angle CPE +\angle PEE'+\angle EE'P'<2\pi&\text{if }\delta >0.
\end{dcases}
\end{equation}
The left-hand side of the second inequality of $\eqref{ineq:existence_P_1}$ is calculated as
\begin{equation}\label{eq:E'P'P}
\begin{aligned}
\angle CPE&+\angle PEE'+\angle EE'P'\\
&=(\angle CPE+\angle CEP)+\angle CEE'+\angle CE'E +\angle CE'P'\\
&=(\pi -\delta )+\left(\frac{\gamma}{2}+\delta -\frac{\psi}{2}\right) +\left(\frac{\gamma'}{2}+\delta'-\frac{\psi'}{2}\right)
+\left(\frac{\pi}{2}+\frac{\gamma'}{2}+\frac{\psi'}{2}+\rho'\right)\\
&=\frac{3}{2}\pi +\gamma_++\delta'-\left(\frac{\gamma}{2}+\frac{\psi}{2}+\rho\right) .
\end{aligned}
\end{equation}
Thus the second inequality of $\eqref{ineq:existence_P_1}$ is rewritten as
\begin{equation}\label{ineq:existence_P}
\gamma_++\delta'-\frac{\pi}{2}<\frac{\gamma}{2}+\frac{\psi}{2}+\rho .
\end{equation}
In a similar way, we see that the first inequlity of $\eqref{ineq:existence_P_1}$ is included in $\eqref{ineq:existence_P}$.

Therefore we obtained four inequalities $\eqref{ineq:existence_G}$, $\eqref{ineq:existence_E}$, $\eqref{ineq:existence_G'}$ and $\eqref{ineq:existence_P}$,
which are arranged into the following inequalities for $\angle ABE =\pi /2+\gamma /2+\delta +\psi /2+\rho$:
\begin{equation}\label{ineq:range_ABE_2}
\begin{aligned}
\pi -\beta&\leqslant\angle ABE\leqslant\pi -\beta +\frac{\gamma_+}{2}+\delta_+,\\
\gamma_++\delta_+&<\angle ABE<\pi ,
\end{aligned}
\end{equation}
which is the same as condition $\eqref{ineq:range_ABE_1}$,
where the left (resp. the right) equality of the first inequality of corresponds to $G =E$ (resp. $G'=E'$).
This explains how we find point $E_\tau$ in Construction $\ref{const:negative_1}$, $(1)$.
Note that the range of $\angle ABE$ determined by $\eqref{ineq:range_ABE_2}$ is nonempty because of
\begin{equation}\label{ineq:nonempty}
\beta +\frac{\gamma_+}{2}<\pi\quad\text{and}\quad\gamma_++\delta_+<\pi ,
\end{equation}
which come from conditions (i) and (iii.c) of Construction $\ref{const:condition}$ respectively.

Now we have the following results.
\begin{proposition}\label{prop:angle_relations}
For the choice of $\tau\in\{\Lt ,\Rt\}$ in Construction $\ref{const:negative_1}$, $(1)$, we have the following relations:
\begin{equation}\label{eq:angle_relations_1}
\begin{aligned}
\angle E_\tau B_\tau G_\tau +\angle E_{\tau'}B_{\tau'}G_{\tau'}&=\frac{\gamma}{2}+\delta_\Lt +\delta_\Rt ,\quad
\angle E_{\tau'}B_{\tau'}P_{\tau'}=\frac{\gamma}{2}+\delta_{\tau'},\\
\angle E_\tau P_\tau P_{\tau'}=\pi -\angle AB_\tau E_\tau&,\quad\angle E_{\tau'}P_{\tau'}P_\tau =\angle AB_\tau E_\tau -(\gamma +\delta_\Lt +\delta_\Rt ),
\quad\text{so that}\\
\angle E_\tau P_\tau P_{\tau'}+\angle E_{\tau'}P_{\tau'}P_\tau&=\pi -(\gamma +\delta_\Lt +\delta_\Rt ).
\end{aligned}
\end{equation}
In particular, we have $E_{\tau'}\neq P_{\tau'}$. Also, we have
\begin{equation}\label{eq:angle_relations_2}
\begin{aligned}
\angle E_{\tau'}B_{\tau'}G_{\tau'}&=\pi -\beta_\tau +\frac{\gamma}{2}+\delta_\Lt +\delta_\Rt -\angle AB_\tau E_\tau ,\\
\angle AB_{\tau'}E_{\tau'}&=\pi +\frac{\gamma}{2}+\delta_\Lt +\delta_\Rt -\angle AB_\tau E_\tau ,\\
\angle P_{\tau'}B_{\tau'}\ell_{\tau'}&=\angle AB_\tau E_\tau -(\gamma +\delta_\Lt +\delta_\Rt ).
\end{aligned}
\end{equation}
\end{proposition}
\begin{proof}
First we will prove $\eqref{eq:angle_relations_1}$.
Using
\begin{align*}
\angle EBG+\angle E'B'G'&=\angle EDG+\angle E'DG'\\
&=\pi -\angle EDE'=\pi -\angle ECE'\\
&=\angle ECP+\angle E'CP'=\angle EBP+\angle E'B'P'
\end{align*}
and
\begin{equation*}
\pi -\angle ECE'=\angle CEE' +\angle CE'E
=\left(\frac{\gamma}{2}+\delta -\frac{\psi}{2}\right) +\left(\frac{\gamma'}{2}+\delta'-\frac{\psi'}{2}\right) =\frac{\gamma_+}{2}+\delta_+
\end{equation*}
which follows from the calculation of $\eqref{eq:E'P'P}$, we obtain the first equation in the first line, and also the second equation by $\angle EBP=\delta$.
For the second line, the first equation is easy and the second follows from the calculation of $\eqref{eq:E'P'P}$, which leads to the third line.

Next we will prove $\eqref{eq:angle_relations_2}$.
Using the first line of $\eqref{eq:angle_relations_1}$, we calculate as
\begin{align*}
\angle E'B'G'&=\gamma_+/2+\delta_+-\angle EBG\\
&=\gamma_+/2+\delta_+-(\angle ABE -\angle ABG)=\pi -\beta +\gamma_+/2+\delta_+-\angle ABE,
\end{align*}
so that using Lemma $\ref{lem:ADG'}$ we have
\begin{align*}
\angle AB'E'&=\angle AB'G'+\angle E'B'G'=\pi +\gamma_+/2+\delta_+-\angle ABE.
\end{align*}
Thus using the above equations and the first line of $\eqref{eq:angle_relations_1}$, we have
\begin{align*}
\angle P'B'\ell'&=(\pi +\delta')-\angle AB'E'-\angle E'B'P'=\angle ABE-(\gamma_++\delta_+)
\end{align*}
as desired.
This completes the proof of Propositon $\ref{prop:angle_relations}$.
\end{proof}
\begin{proposition}\label{prop:existence_QR}
If $AB_\tau E_\tau\in [\pi -\beta_\tau ,\pi -\beta_\tau +\delta_\tau )$, then we have
\begin{equation*}
\angle E_\tau B_\tau P_\tau\geqslant\angle P_\tau B_\tau\ell_\tau -\angle k_\tau B_\tau\ell_\tau >0,
\end{equation*}
where equality holds for $\angle AB_\tau E_\tau =\pi -\beta_\tau$, or equivalently, $E_\tau =G_\tau$.
Thus there exist points $Q_\tau$ and $R_\tau$ in $(8)$ of Construction $\ref{const:negative_1}$.

If $\angle AB_\tau E_\tau\in (\beta_{\tau'}+\gamma +\delta_\tau ,\pi -\beta_\tau +\gamma /2+\delta_\Lt +\delta_\Rt ]$, then we have
\begin{equation*}
\angle E_{\tau'}B_{\tau'}P_{\tau'}>\angle k_{\tau'}B_{\tau'}\ell_{\tau'}-\angle P_{\tau'}B_{\tau'}\ell_{\tau'}>0.
\end{equation*}
Thus there exist points $Q_{\tau'}$ and $R_{\tau'}$ in $(9)$ of Construction $\ref{const:negative_1}$ such that $R_{\tau'}\neq E_{\tau'}$.
\end{proposition}
\begin{proof}
First suppose $ABE\in [\pi -\beta ,\pi -\beta +\delta )$.
Then we have
\begin{equation*}
\angle PB\ell -\angle kB\ell =(\pi +\delta -\angle ABE-\delta )-(\beta -\delta )=\pi -\beta +\delta -\angle ABE>0,
\end{equation*}
so that
\begin{equation*}
\angle EBP -(\angle kB\ell -\angle PB\ell )=\delta -(\pi -\beta+\delta -\angle ABE)=\angle ABE -(\pi -\beta )\geqslant 0,
\end{equation*}
where equality holds for $\angle ABE=\pi -\beta$.
This proves the first assertion.

Next suppose $\angle ABE\in (\beta'+\gamma_++\delta ,\pi -\beta +\gamma_+/2+\delta_+]$.
Then by Proposition $\ref{prop:angle_relations}$ we have
\begin{equation*}
\angle P'B'\ell' -\angle k'B'\ell'=(\angle ABE-\gamma_+-\delta_+)-(\beta'-\delta')=\angle ABE-(\beta'+\gamma_++\delta )>0,
\end{equation*}
so that by Proposition $\ref{prop:angle_relations}$ again we have
\begin{align*}
\angle E'B'P' -(\angle P'B'\ell' -\angle k'B'\ell')&=(\gamma_+/2+\delta')-(\angle ABE-\beta'-\gamma_+-\delta )\\
=\beta'+3\gamma_+/2+\delta_+-\angle ABE&\geqslant\beta'+3\gamma_+/2+\delta_+-(\pi -\beta +\gamma_+/2+\delta_+)\\
&=\beta +\beta'+\gamma_+-\pi >0
\end{align*}
as desired.
This proves the second assertion.
\end{proof}
\begin{remark}\label{rem:ABE_practical}\rm
Here we discuss a practical choice of $\angle ABE$, not to say the best.
Although we can construct a negative $3$D gadget under $\eqref{ineq:range_ABE_2}$,
we need additional creases if we have $\angle kB\ell <\angle PB\ell$ or $\angle k'B'\ell'<\angle P'B'\ell'$,
which are equivalent to $\angle ABE<\pi -\beta +\delta$ and $\angle ABE>\beta'+\gamma_++\delta$ respectively.
This is because if so the face bounded by polygonal chain $\ell_\sigma B_\sigma P_\sigma m_\sigma$ overlaps
with the side face for $\sigma$, and particularly with $k_\sigma$ for either $\sigma =\Lt$ or $\sigma =\Rt$.
Thus it is practical to assume
\begin{equation}\label{ineq:range_ABE_3}
\angle AB_\tau E_\tau\in\left[\pi -\beta_\tau +\delta_\tau ,
\min\left\{\pi -\beta_\tau +\frac{\gamma}{2}+\delta_\Lt +\delta_\Rt ,\beta_{\tau'}+\gamma +\delta_\tau\right\}\right]
\cap (\gamma +\delta_\Lt +\delta_\Rt ,\pi )
\end{equation}
instead of $\eqref{ineq:range_ABE_1}$.
The range determined by $\eqref{ineq:range_ABE_3}$ is still nonempty
because of $\beta_\tau +\beta_{\tau'}+\gamma =2\pi -\alpha >\pi$ and $\beta_\tau -\delta_\tau >0$
coming from conditions (i) and (iii.b) respectively of Construction $\ref{const:condition}$, as well as $\eqref{ineq:nonempty}$.

Also, so as not to make $\angle EPP'$ and $\angle E'P'P$ too small, we should take $\angle ABE$ well over $\gamma_++\delta_+$ and well below $\pi$
by the second line of $\eqref{eq:angle_relations_1}$ in Proposition $\ref{prop:angle_relations}$.
In particular, by the fourth line of $\eqref{eq:angle_relations_1}$ 
we have $\angle EPP'=\angle E'PP'$ if we can take $\angle ABE$ satisfying $\eqref{ineq:range_ABE_1}$ as
\begin{equation*}
\angle ABE=(\pi +\gamma_++\delta_+)/2.
\end{equation*}
Even if we can not, it is practical to take $\angle ABE$ as close to $(\pi +\gamma_++\delta_+)/2$ as possible.
\end{remark}

To end this subsection, we check the foldability of the resulting crease patterns.\\
$\bullet$ \emph{flat-foldability around $B$. }
This holds without the creases made with $k$, and thus we will forget them.
The creases starting from $B$ except for those made with $k$ are
\begin{equation*}
\begin{dcases}
AB\text{ and }\ell&\text{if }\delta =0,\\
AB ,BE ,BP\text{ and }\ell&\text{if }\delta >0,
\end{dcases}
\end{equation*}
The flat-foldability in the first case is trivial because $AB$ and $\ell$ form a straight line,
and that in the second case holds because we have
\begin{equation*}
\angle EBP+\angle AB\ell =\delta +(\pi -\delta )=\pi =\angle ABE+\angle PB\ell .
\end{equation*}
$\bullet$ \emph{flat-foldability around $G$ and $P$. }
This is obvious except for $G=E$ or $P=E$, which is included in the flat-foldability around $E$ below.\\
$\bullet$ \emph{flat-foldability around $E$. }
As with $B$, the creases from $E$ except for those made with $k$ are
\begin{equation*}
\begin{dcases}
AE ,EE', EP'\text{ and }m&\text{if }\delta =0,\\
AE ,BE ,EE'\text{ and }m&\text{if }\delta >0.
\end{dcases}
\end{equation*}
Then we can check the flat-foldability in the respective cases as
\begin{align*}
&\angle AEE'-\angle E'EP'+\angle P'Em-\angle AEm\\
&=(\angle AED+\angle DEE')-\angle CEE'+\angle CEm-(\angle AEB+\angle BEm)=0,\\
&\angle AEE'-\angle E'Em+\angle BEm-\angle AE B\\
&=(\angle AED+\angle DEE')-(\angle CEE'+\angle CEm)+\angle BEm-\angle AEB=0
\end{align*}
using 
\begin{equation*}
\angle AEB=\angle AED,\quad\angle CEE'=\angle DEE',\quad\angle BEm=\angle CEm.
\end{equation*}
$\bullet$ \emph{flat-foldability around $P'$. }This is obvious.\\
$\bullet$ \emph{flat-foldability around $E'(\neq G')$. }We have
\begin{align*}
&\angle AE'B'-\angle B'E'P'+\angle EE'P'-\angle AE'E\\
&=\angle AE'B'-\angle B'E'P'+(\angle CE'P'+\angle CE'E)-(\angle DE'E+\angle AE'D)=0
\end{align*}
as desired.\\
$\bullet$ \emph{foldability around $B'$. }We check that $\angle k'B'G'=\pi -\alpha$.
Let us forget the creases made with $k'$.
Then using Proposition $\ref{prop:angle_relations}$, we calculate as
\begin{equation}\label{eq:pi-alpha}
\begin{aligned}
\angle k'B'G'&=\angle k'B'\ell'-\angle P'B'\ell'+\angle E'B'P'-\angle E'B'G'\\
&=(\beta'-\delta')-(\pi +\delta'-\angle AB'E'-\angle E'B'P')+\angle E'B'P'-(\angle AB'E'-\beta )\\
&=\beta +\beta' +\gamma_+-\pi =\pi -\alpha
\end{aligned}
\end{equation}
as desired.
\begin{remark}\rm
We can include the case $\alpha =\beta_\Lt +\beta_\Rt$ in Construction $\ref{const:negative_1}$, where the resulting extrusion is flat.
The assignment of mountain and valley folds to the creases on the side of $\tau'$ for $E_{\tau'}=G_{\tau'}$
is given in Table $\ref{tbl:assignment_negative_1_flat}$,
while the assignment for $E_{\tau'}\neq G_{\tau'}$ is the same as Tables $\ref{tbl:assignment_negative_1}$.

To verify this, suppose $\alpha =\beta +\beta'$.
Note that we have
\begin{equation*}
\angle AB'G'=\beta ,\quad\angle AB'k'=\pi -\beta',
\end{equation*}
so that
\begin{equation}\label{eq:alpha-pi}
\angle AB'G'-\angle AB'k'=\alpha -\pi .
\end{equation}
Thus if $G'\neq E'$, then summing up $\eqref{eq:pi-alpha}$ and $\eqref{eq:alpha-pi}$ gives that
\begin{equation*}
\angle k'B'\ell'-\angle P'B'\ell'+\angle E'B'P'-\angle E'B'G' +\angle AB'G' -\angle AB'k'=0,
\end{equation*}
which implies the flat-foldability around $B'$.
If $G' =E'$, where $\angle ABE=\pi -\beta +\gamma_+/2+\delta_+$, then we have
\begin{align*}
\angle AB'P'&=\angle AB'G'+\angle E'B'P'=\beta +\gamma_+/2+\delta',\quad k'B'\ell' =\beta'-\delta',\\
\angle P'B'\ell'&=\angle ABE-(\gamma_++\delta_+)=\pi -\beta -\gamma_+/2,\quad\angle AB'k'=\pi -\beta',
\end{align*}
where we used Proposition $\ref{prop:angle_relations}$, so that
\begin{equation*}
\angle AB'P'+\angle k'B'\ell' =\beta+\beta'+\gamma_+/2 =\pi =2\pi -(\beta +\beta'+\gamma_+/2)=\angle AB'k'+\angle P'B'\ell',
\end{equation*}
which implies the flat-foldability around $B'$.
Also, we have
\begin{align*}
\angle AE'P'&-\angle EE'P'+\angle EE'G-\angle AE'G\\
&=(\angle AG'B'+\angle B'E'P')-(\angle CE'P'+\angle CE'E)+\angle DE'E-\angle AG'D\\
&=(\angle AG'B'-\angle AG'D)+(\angle B'E'P'-\angle CE'P')-(\angle CE'E-\angle DE'E)=0,
\end{align*}
which implies the flat-foldability around $E'$.
\end{remark}
\begin{table}[h]
\begin{tabular}{c|c|c}
&\multicolumn{2}{c}{\bettershortstack{creases on the side of $\tau'$ for\\
$\alpha =\beta_\Lt +\beta_\Rt$ and $E_{\tau'}=G_{\tau'};$\\
\hfill$M_1=\beta_{\tau'}+\gamma +\delta_\tau$}}\\ 
\cline{2-3}
&$\angle AB_\tau E_\tau\leqslant M_1$&$\angle AB_\tau E_\tau >M_1$\\
\hline
\multirow{2}{*}{mountain folds}&\multicolumn{2}{c}{$B_{\tau'}P_{\tau'}$}\\
\cline{2-3}
&---&$B_{\tau'}Q_{\tau'},CR_{\tau'}$\\
\hline
valley folds&---&$B_{\tau'}R_{\tau'},CQ_{\tau'}$
\end{tabular}\vspace{0.5cm}
\caption{Assignment of mountain and valley folds to the creases of a negative $3$D gadget by our first construction
on the side of $\tau'$ for $\alpha =\beta_\Lt +\beta_\Rt$ and $E_{\tau'}=G_{\tau'}$}
\label{tbl:assignment_negative_1_flat}
\end{table}

\subsection{Second construction}\label{subsec:negative_new_2}
This is an extension to $\delta_\sigma\geqslant 0$ of our second construction of negative $3$D gadgets given in \cite{Doi19}, Section $9$.
Suppose either $\beta_\Lt ,\beta_\Rt\leqslant\pi /2$ or $\beta_\Lt ,\beta_\Rt\geqslant\pi /2$ holds. 
Here we consider a development as in Figure $\ref{fig:development_negative_2}$ with prescribed simple outgoing pleats
$(\ell_\Lt ,m_\Lt )$ and $(\ell_\Rt ,m_\Rt )$ instead of Figure $\ref{fig:development_1}$.
Then we consider the crease pattern as in Figure $\ref{fig:CP_negative_2_initial}$,
where we take $G_\Lt ,G_\Rt ,P'_\Lt ,P'_\Rt$ so that segments $G_\Lt G_\Rt$ and $P'_\Lt P'_\Rt$ are parallel to segment $E_\Lt E_\Rt$.
Since $\triangle AG_\Lt G_\Rt$ does not face the correct direction,
this is \emph{not} the correct crease pattern, and thus $\triangle AG_\Lt G_\Rt$ has to be rotated with a rotation of crease $P'_\Lt P'_\Rt$.
\begin{figure}[htbp]
\addtocounter{theorem}{1}
\centering\includegraphics[width=0.75\hsize]{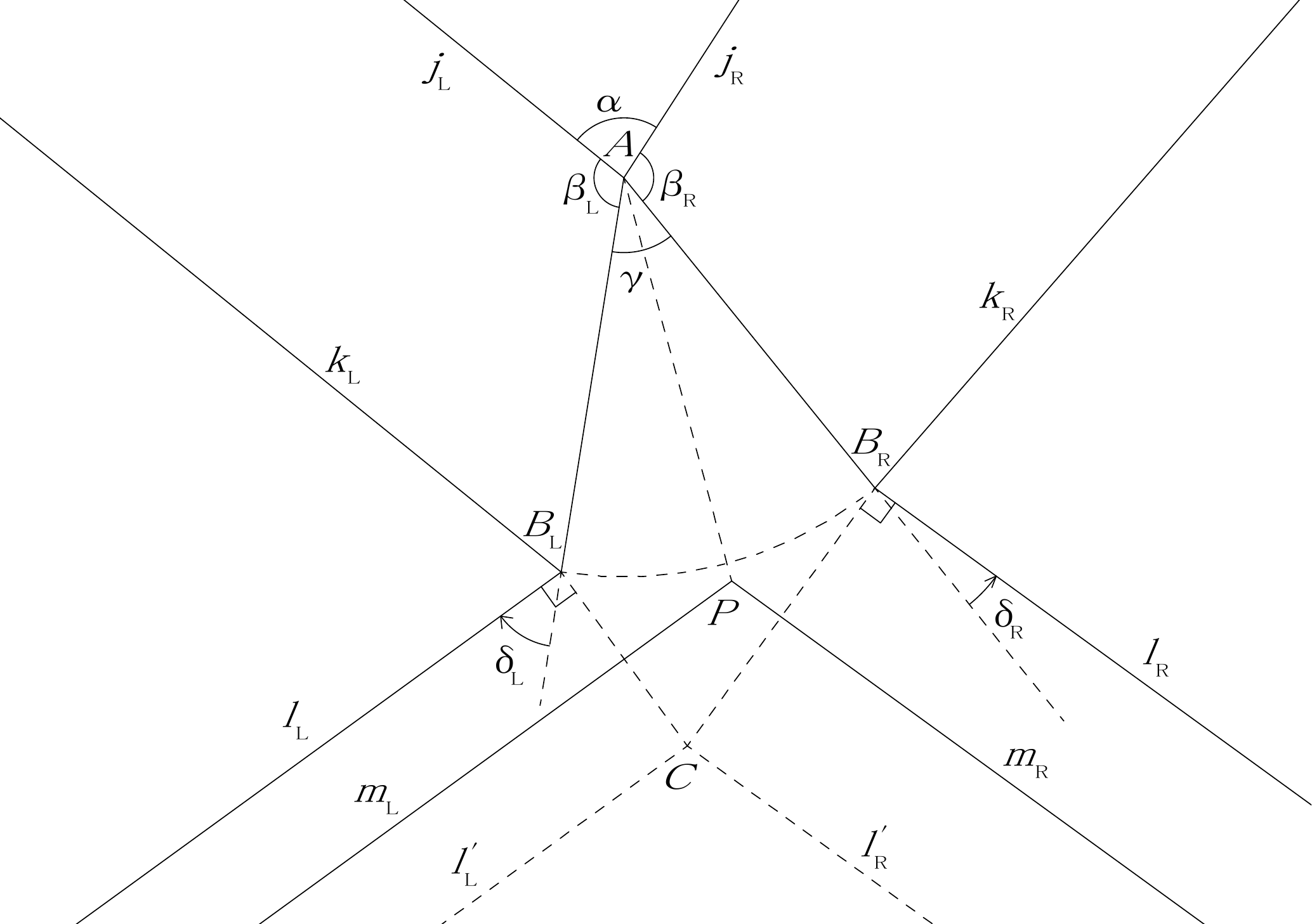}
    \caption{Development of a negative $3$D gadget to which we apply our second construction}
    \label{fig:development_negative_2}
\end{figure}
\begin{figure}[htbp]
\addtocounter{theorem}{1}
\centering\includegraphics[width=0.75\hsize]{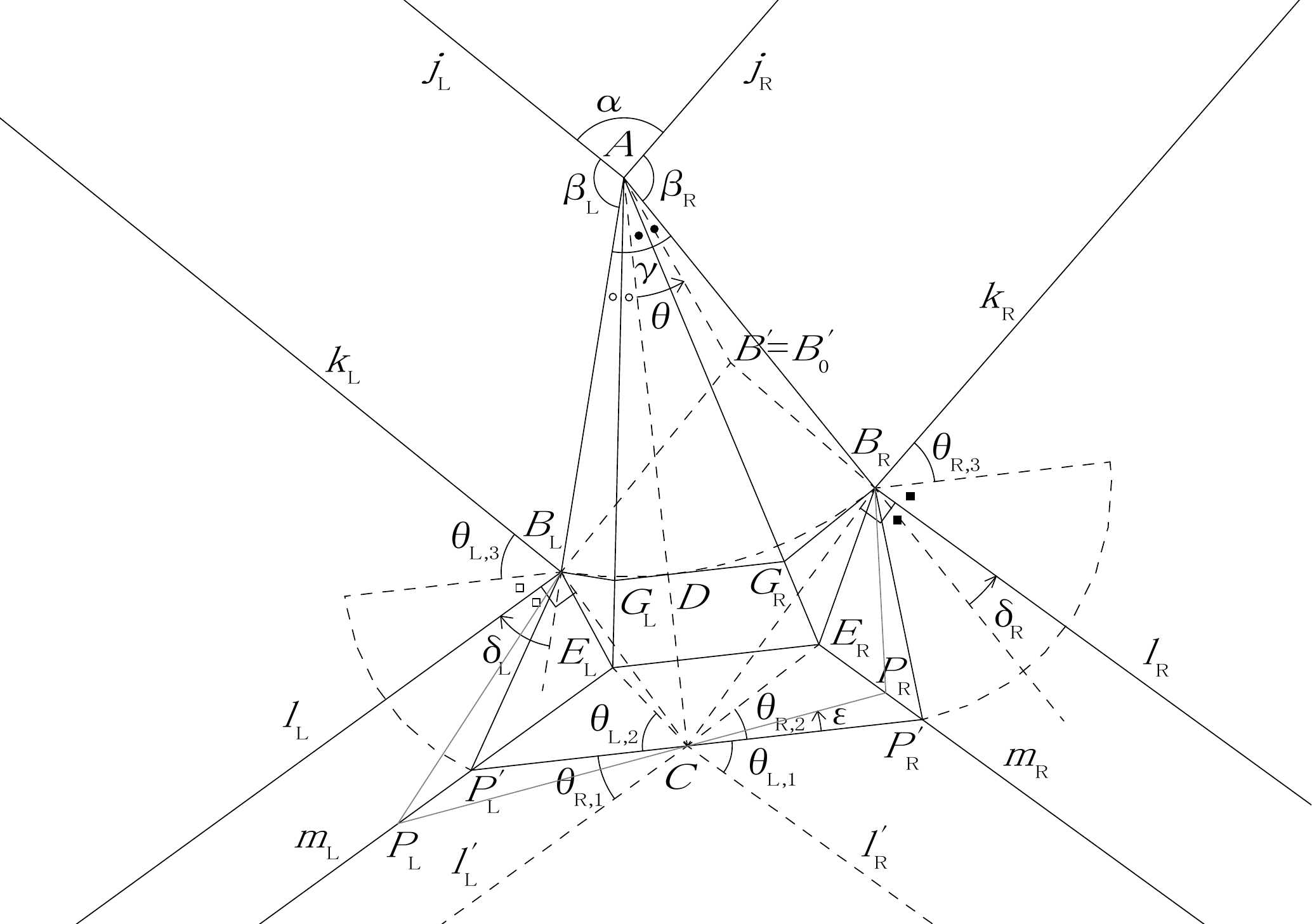}
    \caption{CP of a negative $3$D gadget where crease $P'_\Lt P'_\Rt$ is to be rotated}
    \label{fig:CP_negative_2_initial}
\end{figure}
Let $B'$ be the intersection point of a perpendicular through $B_\Lt$ to $k_\Lt$ and a perpendicular through $B_\Rt$ to $k_\Rt$ in the crease pattern.
Also, let $\theta$ be the angle such that if $AC$ rotates counterclockwise around $A$ by $\theta\in (-\pi/2 ,\pi /2]$,
then $AC$ overlaps with an extension of $AB'$.
To be more precise, we define $B'_0$ as $B'$ if $\beta_\Lt ,\beta_\Rt\geqslant\pi$
and the antipodal point of $B'$ in the circle with center $A$ and radius $\norm{AB'}$ if $\beta_\Lt ,\beta_\Rt\leqslant\pi$.
We see easily that $B'_0$ lies strictly inside $\angle B_\Lt AB_\Rt$ if $\beta_\Lt ,\beta_\Rt >\pi /2$ or $\beta_\Lt ,\beta_\Rt <\pi /2$,
and in $AB_\sigma$ if $\beta_\sigma =\pi /2$ and $\beta_{\sigma'}\neq\pi /2$, while $B'=B'_0=A$ if $\beta_\Lt =\beta_\Rt =\pi /2$.
Then $\theta$ is defined for $B'_0\neq A$ by
\begin{equation*}
\theta =\angle B_\Lt AB'_0-\angle B_\Lt AC
\end{equation*}
For $B'_0\neq A$, we calculate $\theta$ as
\begin{equation}\label{eq:theta}
\theta =\tan^{-1}\left\{\frac{\tan\beta_\Rt -\tan\beta_\Lt}{2+(\tan\beta_\Lt +\tan\beta_\Rt )/\tan (\gamma /2)}\right\}
-\tan^{-1}\left\{\frac{\tan\delta_\Rt -\tan\delta_\Lt}{2+(\tan\delta_\Lt +\tan\delta_\Rt )/\tan (\gamma /2)}\right\} ,
\end{equation}
and we particularly have
\begin{equation*}
\theta=\begin{dcases}
-\gamma_\Lt&\text{if }\beta_\Lt =\pi /2\text{ and }\beta_\Rt\neq\pi /2,\\
\phantom{-}\gamma_\Rt&\text{if }\beta_\Rt =\pi /2\text{ and }\beta_\Lt\neq\pi /2.
\end{dcases}
\end{equation*}
Let us define $\proj_{A,0}$ to be the projection of the resulting extrusion to the crease pattern in the bottom plane $\{ z=0\}$ with $\proj_{A,0}(A)=A$.
Then we have $\proj_{A,0}(B)=B'$.
Since $G_\Lt G_\Rt$ in the bottom plane $\{ z=0\}$ must be perpendicular to $AB$ in the resulting extrusion, and also perpendicular to $AB'$ in the crease pattern,
we have to rotate $G_\Lt G_\Rt$ counterclockwise around $D$ by $\theta$ with a rotation of crease $P'_\Lt P'_\Rt$ in the crease pattern.

Let us observe how a rotation of crease $P'_\Lt P'_\Rt$ in the crease pattern affects the resulting extrusion.
As in Figure $\ref{fig:CP_negative_2_initial}$, suppose we use instead of $P'_\Lt P'_\Rt$
a new crease $P_\Lt P_\Rt$ made with a counterclockwise rotation of $P'_\Lt P'_\Rt$ by $\epsilon$,
and accordingly $B_\sigma P_\sigma$ instead of $B_\sigma P'_\sigma$ for $\sigma =\Lt ,\Rt$.
Then in the resulting extrusion we have
\begin{equation}\label{eq:kDG_L}
\begin{aligned}
\angle k_\Lt DG_\Lt =\angle k_\Lt BG_\Lt&=\angle k_\Lt B_\Lt \ell_\Lt -\angle P_\Lt B_\Lt\ell_\Lt +\angle E_\Lt B_\Lt P_\Lt -\angle E_\Lt B_\Lt G_\Lt\\
&=\angle k_\Lt B_\Lt \ell_\Lt -(\angle P'_\Lt B_\Lt\ell_\Lt -\epsilon )+(\angle E_\Lt B_\Lt P'_\Lt +\epsilon )-\angle E_\Lt B_\Lt G_\Lt\\
&=(\beta_\Lt -\delta_\Lt )-(\pi +\delta_\Lt -\angle AB_\Lt E_\Lt -\angle E_\Lt B_\Lt P'_\Lt )\\
&\phantom{{}=(\beta_\Lt -\delta_\Lt )}+\angle E_\Lt B_\Lt P'_\Lt -(\angle AB_\Lt E_\Lt -\pi /2)+2\epsilon\\
&=\beta_\Lt -2\delta_\Lt -\pi /2+2\angle E_\Lt B_\Lt P'_\Lt +2\epsilon\\
&=\beta_\Lt +\gamma_\Lt -\pi /2+2\epsilon ,
\end{aligned}
\end{equation}
and similarly
\begin{equation}\label{eq:kDG_R}
\angle k_\Rt DG_\Rt =\beta_\Rt +\gamma_\Rt -\pi /2-2\epsilon ,
\end{equation}
where we used 
\begin{equation*}
\angle E_\sigma B_\sigma P'_\sigma =\angle E_\sigma CP'_\sigma =\angle CE_\sigma E_{\sigma'}=\gamma_\sigma /2+\delta_\sigma\quad\text{for }\sigma =\Lt ,\Rt
\end{equation*}
by \cite{Doi20}, Lemma $4.4$ for $\psi_\sigma =\rho_\sigma =0$.
On the other hand, to make $G_\Lt G_\Rt$ orthogonal $AB$ in the resulting extrusion we need
\begin{equation}\label{eq:kDG_orthogonal}
\begin{aligned}
\angle k_\Lt DG_\Lt&=\angle j_\Lt AB'_0-\pi /2=\beta_\Lt +\gamma_\Lt -\pi /2+\theta ,\\
\angle k_\Lt DG_\Rt&=\angle j_\Rt AB'_0-\pi /2=\beta_\Rt +\gamma_\Rt -\pi /2-\theta .
\end{aligned}
\end{equation}
Hence comparing $\eqref{eq:kDG_L}$ and $\eqref{eq:kDG_R}$ with $\eqref{eq:kDG_orthogonal}$, we conclude that $\epsilon =\theta /2$.

However, we have some constraints on the range of $P_\sigma$, and thus on that of $\epsilon$.

Firstly, let $\theta_{\sigma ,1}$ and $\theta_{\sigma ,2}$ for $\sigma =\Lt ,\Rt$ be the angles shown in Figure $\ref{fig:CP_negative_2_initial}$,
which are calculated as
\begin{equation*}\label{eq:theta_sigma_12}
\begin{aligned}
\theta_{\sigma ,1}&=\angle P'_{\sigma'}C\ell'_{\sigma'}=\angle P'_{\sigma'}B_{\sigma'}\ell_{\sigma'}
=\pi +\delta_{\sigma'}-\angle AB_{\sigma'}E_{\sigma'}-\angle E_{\sigma'}B_{\sigma'}P'_{\sigma'}=\frac{\pi}{2}-\gamma_{\sigma'}-\delta_{\sigma'},\\
\theta_{\sigma ,2}&=\angle E_\sigma CP'_\sigma =\angle CE_\sigma E_{\sigma'}=\frac{\gamma_\sigma}{2}+\delta_\sigma ,
\end{aligned}
\end{equation*}
where we used \cite{Doi20}, Lemma $4.4$ for $\psi_\sigma =\rho_\sigma =0$.
Then $P_\Lt\in m_\Lt$ and $P_\Rt\in m_\Rt$ imply
\begin{equation*}
-\theta_{\Lt ,2}\leqslant\epsilon <\theta_{\Rt ,1}\quad\text{and}\quad-\theta_{\Lt ,1}<\epsilon\leqslant\theta_{\Rt ,2}
\end{equation*}
respectively, which we arrange into
\begin{align}
-\theta_{\Lt ,1}&<\epsilon <\theta_{\Rt ,1}\quad\text{and}\label{ineq:range_epsilon_1}\\
-\theta_{\Lt ,2}&\leqslant\epsilon\leqslant\theta_{\Rt ,2}.\label{ineq:range_epsilon_2}
\end{align}

Secondly, we have $\angle k_\sigma DG_\sigma\geqslant 0$ for $\sigma =\Lt ,\Rt$.
It follows from $\eqref{eq:kDG_L}$ and $\eqref{eq:kDG_R}$ that
\begin{equation}\label{ineq:range_epsilon_3_1}
\pi /2-\beta_\Lt -\gamma_\Lt\leqslant 2\epsilon\leqslant\beta_\Rt +\gamma_\Rt -\pi /2.
\end{equation}
Thus letting $\theta_{\sigma ,3}$ be the angles shown in Figure $\ref{fig:CP_negative_2_initial}$, we have
\begin{equation*}
\theta_{\sigma ,3}=\angle k_\sigma B_\sigma\ell_\sigma -\angle P'_\sigma B_\sigma\ell_\sigma
=\beta_\sigma -\delta_\sigma -\theta_{\sigma',2}=\beta_\sigma +\gamma_\sigma -\frac{\pi}{2},
\end{equation*}
so that $\eqref{ineq:range_epsilon_3_1}$ becomes
\begin{equation}\label{ineq:range_epsilon_3_2}
-\theta_{\Lt ,3}\leqslant 2\epsilon\leqslant\theta_{\Rt ,3}.
\end{equation}

\begin{lemma}\label{lem:conditions_theta_23}
Suppose $\beta_\Lt ,\beta_\Rt\geqslant\pi /2$ or $\beta_\Lt ,\beta_\Rt\leqslant\pi /2$, and either $\beta_\Lt\neq\pi /2$ or $\beta_\Rt\neq\pi /2$,
so that $\theta$ is well-defined.
For $\epsilon =\theta /2$, conditions $\eqref{ineq:range_epsilon_2}$ and $\eqref{ineq:range_epsilon_3_2}$ hold automatically.
\end{lemma}
\begin{proof}
Condition $\eqref{ineq:range_epsilon_2}$, and condition $\eqref{ineq:range_epsilon_3_2}$ for $\beta_\Lt ,\beta_\Rt\geqslant\pi /2$
are obvious because we have $-\gamma_\Lt\leqslant\theta\leqslant\gamma_\Rt$.

Suppose $\beta_\Lt ,\beta_\Rt\leqslant\pi /2$ in condition $\eqref{ineq:range_epsilon_3_2}$.
For $\sigma =\Lt ,\Rt$, let $p_\sigma$ be a ray starting from $A$ which is parallel to $B'B_\sigma$ and going to the side of $B_\sigma$
as shown in Figure $\ref{fig:construction_p}$.
Then it follows from
\begin{equation}\label{ineq:B'Ap_L}
\gamma_\Lt +\theta =\angle B_\Lt AB'_0=\angle B_\Lt Ap_\Lt +\angle B'_0Ap_\Lt =\pi /2-\beta_\Lt +\angle B'_0Ap_\Lt\geqslant\pi /2-\beta_\Lt ,
\end{equation}
we have the left inequality of $\eqref{ineq:range_epsilon_3_2}$.
Also, it follows from
\begin{equation}\label{ineq:B'Ap_R}
\gamma_\Lt +\theta =\angle B_\Lt AB_\Rt -\angle B'_0 Ap_\Rt -\angle B_\Rt Ap_\Rt =\beta_\Rt +\gamma -\pi /2-\angle B'_0Ap_\Rt
\leqslant\beta_\Rt +\gamma -\pi /2,
\end{equation}
we have the right inequality of $\eqref{ineq:range_epsilon_3_2}$.
Equality in the left (resp. the right) inequality of $\eqref{ineq:range_epsilon_3_2}$ holds if $\angle B'_0Ap_\Lt =0$ (resp. $\angle B'_0Ap_\Rt =0$),
or equivalently $\beta_\Lt =\pi /2$ and $\beta_\Rt <\pi /2$ (resp. $\beta_\Rt =\pi /2$ and $\beta_\Lt <\pi /2$),
because $\angle B'_0Ap_\sigma\in [0,(\pi -\alpha )/2]$ can be determined by
\begin{equation}\label{eq:tan_B'Ap}
\tan\angle B'_0Ap_\sigma =\frac{\sin\alpha\cos\beta_\sigma}{\cos\beta_{\sigma'}-\cos\alpha\cos\beta_\sigma}.
\end{equation}
This completes the proof of Lemma $\ref{lem:conditions_theta_23}$.
\end{proof}
\begin{figure}[htbp]
\addtocounter{theorem}{1}
\centering\includegraphics[width=0.75\hsize]{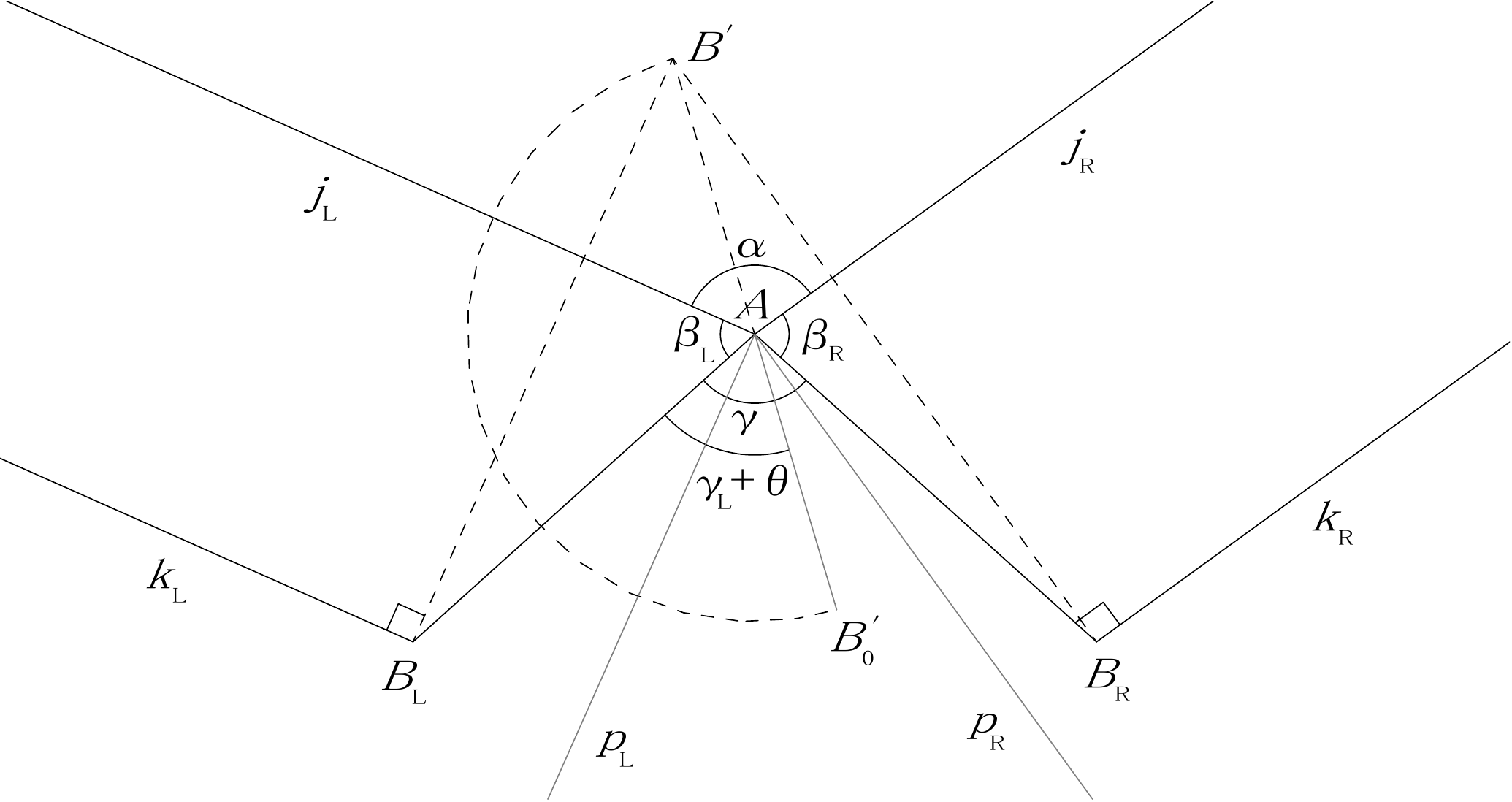}
    \caption{Construction of rays $p_\Lt$ and $p_\Rt$}
    \label{fig:construction_p}
\end{figure}
\begin{lemma}\label{lem:pi-beta_B'Ap}
Suppose $\beta_\Lt ,\beta_\Rt\leqslant\pi /2$, and either $\beta_\Lt\neq\pi /2$ or $\beta_\Rt\neq\pi /2$, so that $B'_0\neq A$.
For $\sigma =\Lt ,\Rt$, let $p_\sigma$ be a ray starting from $A$ which is parallel to $B'B_\sigma$ and going to the side of $B_\sigma$
as shown in Figure $\ref{fig:construction_p}$.
Then we have
\begin{equation}\label{ineq:pi-beta_B'Ap}
\pi /2-\beta_\sigma\leqslant\angle B'_0Ap_\sigma\quad\text{for }\sigma =\Lt ,\Rt ,
\end{equation}
where equality holds if $\beta_\sigma =\pi /2$.
\end{lemma}
\begin{proof}
If $\cos\beta_{\sigma'}-\cos\alpha\cos\beta_\sigma\leqslant 0$ in $\eqref{eq:tan_B'Ap}$, then we have $\angle B'_0Ap_\sigma\geqslant\pi /2$, so that
$\eqref{ineq:pi-beta_B'Ap}$ holds.
Thus we may assume $\cos\beta_{\sigma'}-\cos\alpha\cos\beta_\sigma >0$,
and it is sufficient to prove 
\begin{equation*}
\tan (\pi /2-\beta_\sigma )\leqslant\tan\angle B'_0Ap_\sigma ,
\end{equation*}
which is rewritten as
\begin{equation}\label{ineq:frac_alpha_beta}
\frac{\cos\beta_\sigma}{\sin\beta_\sigma}<\frac{\sin\alpha\cos\beta_\sigma}{\cos\beta_{\sigma'}-\cos\alpha\cos\beta_\sigma}
\end{equation}
by $\eqref{eq:tan_B'Ap}$.
By the assumption, we see that $\eqref{ineq:frac_alpha_beta}$ holds because
\begin{align*}
&\sin\alpha\cos\beta_\sigma\sin\beta_\sigma-(\cos\beta_{\sigma'}-\cos\alpha\cos\beta_\sigma )\cos\beta_\sigma\\
&=(\sin\alpha\sin\beta_\sigma -\cos\beta_{\sigma'}+\cos\alpha\cos\beta_\sigma )\cos\beta_\sigma\\
&=\{\cos (\alpha -\beta_\sigma )-\cos\beta_{\sigma'}\}\cos\beta_\sigma\\
&=2\sin\frac{\alpha +\beta_{\sigma'}-\beta_\sigma}{2}\sin\frac{\beta_\sigma +\beta_{\sigma'}-\alpha}{2}\cos\beta_\sigma\geqslant 0,
\end{align*}
where equality holds if $\beta_\sigma =\pi /2$.
This completes the proof.
\end{proof}
\begin{proposition}\label{prop:kBl>PBl}
Suppose $\beta_\Lt ,\beta_\Rt\geqslant\pi /2$ or $\beta_\Lt ,\beta_\Rt\leqslant\pi /2$, and either $\beta_\Lt\neq\pi /2$ or $\beta_\Rt\neq\pi /2$,
so that the angle $\theta$ is well-defined.
Also suppose $\theta$ satisfy conditions $\eqref{ineq:range_epsilon_1}$, $\eqref{ineq:range_epsilon_2}$
and $\eqref{ineq:range_epsilon_3_2}$ for $\epsilon =\theta /2$.
Then for this $\theta$ we have
\begin{equation*}
\angle k_\sigma B_\sigma\ell_\sigma >\angle P_\sigma B_\sigma\ell_\sigma\quad\text{for }\sigma =\Lt ,\Rt .
\end{equation*}
Thus in our second construction,
there appear no creases like $B_\sigma Q_\sigma$ or $B_\sigma R_\sigma$ in Construction $\ref{const:negative_1}$ which are made with $k_\sigma$.
\end{proposition}
\begin{proof}
Since we have
\begin{equation}\label{ineq:kBl-PBl}
\begin{aligned}
\angle k_\Lt B_\Lt\ell_\Lt -\angle P_\Lt B_\Lt\ell_\Lt&=\beta_\Lt +\gamma_\Lt -\pi /2+\theta /2,\\
\angle k_\Rt B_\Rt\ell_\Rt -\angle P_\Rt B_\Rt\ell_\Rt&=\beta_\Rt +\gamma_\Rt -\pi /2-\theta /2
\end{aligned}
\end{equation}
by a calculation similar to $\eqref{eq:kDG_L}$, it is sufficient to prove that the right-hand sides are positive.

First suppose $\beta_\Lt ,\beta_\Rt\geqslant\pi /2$.
Then using $\eqref{ineq:range_epsilon_2}$ for $\epsilon =\theta /2$, we have
\begin{align*}
\gamma_\Lt +\theta /2\geqslant\gamma_\Lt /2+\delta_\Lt >0,\quad\gamma_\Rt -\theta /2\geqslant\gamma_\Rt /2+\delta_\Rt >0,
\end{align*}
from which follows the positivity of the right-hand sides of $\eqref{ineq:kBl-PBl}$.

Next suppose $\beta_\Lt ,\beta_\Rt\leqslant\pi /2$.
Then we calculate the double of the first line of $\eqref{ineq:kBl-PBl}$ as
\begin{align*}
2\beta_\Lt +2\gamma_\Lt -\pi +\theta&=(\beta_\Lt +\gamma_\Lt +\theta -\pi /2)+(\beta_\Lt +\gamma_\Lt -\pi /2)\\
&=\angle B'_0Ap_\Lt +\beta_\Lt +\gamma_\Lt -\pi /2\geqslant\gamma_\Lt >0,
\end{align*}
where we used $\eqref{ineq:B'Ap_L}$ and Lemma $\ref{lem:pi-beta_B'Ap}$ in the second line.
Also, we can use $\eqref{ineq:B'Ap_R}$ and Lemma $\ref{lem:pi-beta_B'Ap}$ to prove that the second line of $\eqref{ineq:kBl-PBl}$ is positive.
This completes the proof of Proposition $\ref{prop:kBl>PBl}$.
\end{proof}
Hence setting $\theta_\sigma =2\theta_{\sigma ,1}$, we see from $\eqref{ineq:range_epsilon_1}$, $\eqref{ineq:range_epsilon_2}$, $\eqref{ineq:range_epsilon_3_2}$
and Lemma $\ref{lem:conditions_theta_23}$ that our second construction is possible if $\theta =2\epsilon$ satisfies $\theta\in (-\theta_\Lt ,\theta_\Rt )$.
Note that the flat-foldability around $P_\sigma$ is obvious and that around $E_\sigma$ is verified in the same way as in Section $\ref{subsec:negative_new_1}$.

Now our second construction is summarized as follows.
\begin{construction}\label{const:negative_2}\rm
Consider a development as in Figure $\ref{fig:development_negative_2}$, 
for which we require condition (i), (ii) and (iii.a)--(iii.c) of Construction $\ref{const:condition}$.
Suppose either $\beta_\Lt ,\beta_\Rt\leqslant\pi /2$ or $\beta_\Lt ,\beta_\Rt\geqslant\pi /2$ holds.
Let $B'$ be the intersection point of a perpendicular through $B_\Lt$ to $k_\Lt$ and a perpendicular through $B_\Rt$ to $k_\Rt$.
Also, let $B'_0$ be $B'$ if $\beta_\Lt ,\beta_\Rt\geqslant\pi$
and the antipodal point of $B'$ in the circle with center $A$ and radius $\norm{AB'}$ if $\beta_\Lt ,\beta_\Rt\leqslant\pi$,
so that $B'_0$ lies inside $\angle B_\Lt AB_\Rt$.
If either $\beta_\Lt\neq\pi /2$ or $\beta_\Lt\neq\pi /2$, 
then define an angle $\theta$ numerically by $\eqref{eq:theta}$ or geometrically in Figure $\ref{fig:CP_negative_2_initial}$ by
\begin{equation*}
\theta =\angle B_\Lt AB'_0-\angle B_\Lt AC.
\end{equation*}
Also define $\theta_\sigma$ for $\sigma =\Lt ,\Rt$ by
\begin{equation*}
\theta_\sigma =\pi -2\gamma_{\sigma'}-2\delta_{\sigma'}
\end{equation*}
(We can also define $\theta_\sigma$ geometrically in Figure $\ref{fig:CP_negative_2_initial}$ by
\begin{equation*}
\theta_\sigma /2=\theta_{\sigma ,1}=\angle P'_{\sigma'}C\ell'_{\sigma'}=\angle P'_{\sigma'}B_{\sigma'}\ell_{\sigma'}
\end{equation*}
after procedure $(2)$ below.)
Suppose $\theta\in (-\theta_\Lt ,\theta_\Rt )$ if either $\beta_\Lt\neq\pi /2$ or $\beta_\Rt\ne	\pi /2$,
or choose any $\theta\in (-\theta_\Lt ,\theta_\Rt )$ if $\beta_\Lt =\beta_\Rt =\pi /2$.
Then the crease pattern of our second negative $3$D gadget with prescribed simple outgoing pleats $(\ell_\Lt ,m_\Lt )$ and $(\ell_\Rt ,m_\Rt )$
is constructed as follows, where we regard $\sigma$ as taking both $\Lt$ and $\Rt$.
\begin{enumerate}[(1)]
\item Let $E_\sigma$ be the intersection point of $m_\sigma$ and a bisector of $\angle B_\sigma AC$.
Also redefine $m_\sigma$ to be a ray starting from $E_\sigma$ and going in the same direction as $\ell_\sigma$.
\item Let $D$ be the intersection point of $AC$ and the circular arc $B_\Lt B_\Rt$ with center $A$.
Draw a perpendicular through $D$ to $AD$, letting $G_\sigma$ be the intersection point of the perpendicular and $AE_\sigma$.
Also draw a perpendicular through $C$ to $AC$, letting $P'_\sigma$ be the intersection point of the perpendicular and $m_\sigma$.
\item Rotate the line through $P'_\Lt$ and $P'_\Rt$ around $C$ by $\theta /2$ in the counterclockwise direction, 
letting $P_\sigma$ be the intersection point of the rotated line and $m_\sigma$.
\item The desired crease pattern is shown in Figure $\ref{fig:CP_negative_2}$,
and the assignment of mountain and valley folds is given in Table $\ref{tbl:assignment_negative_2}$.
\end{enumerate}
\end{construction}
\begin{figure}[htbp]
\addtocounter{theorem}{1}
\centering\includegraphics[width=0.75\hsize]{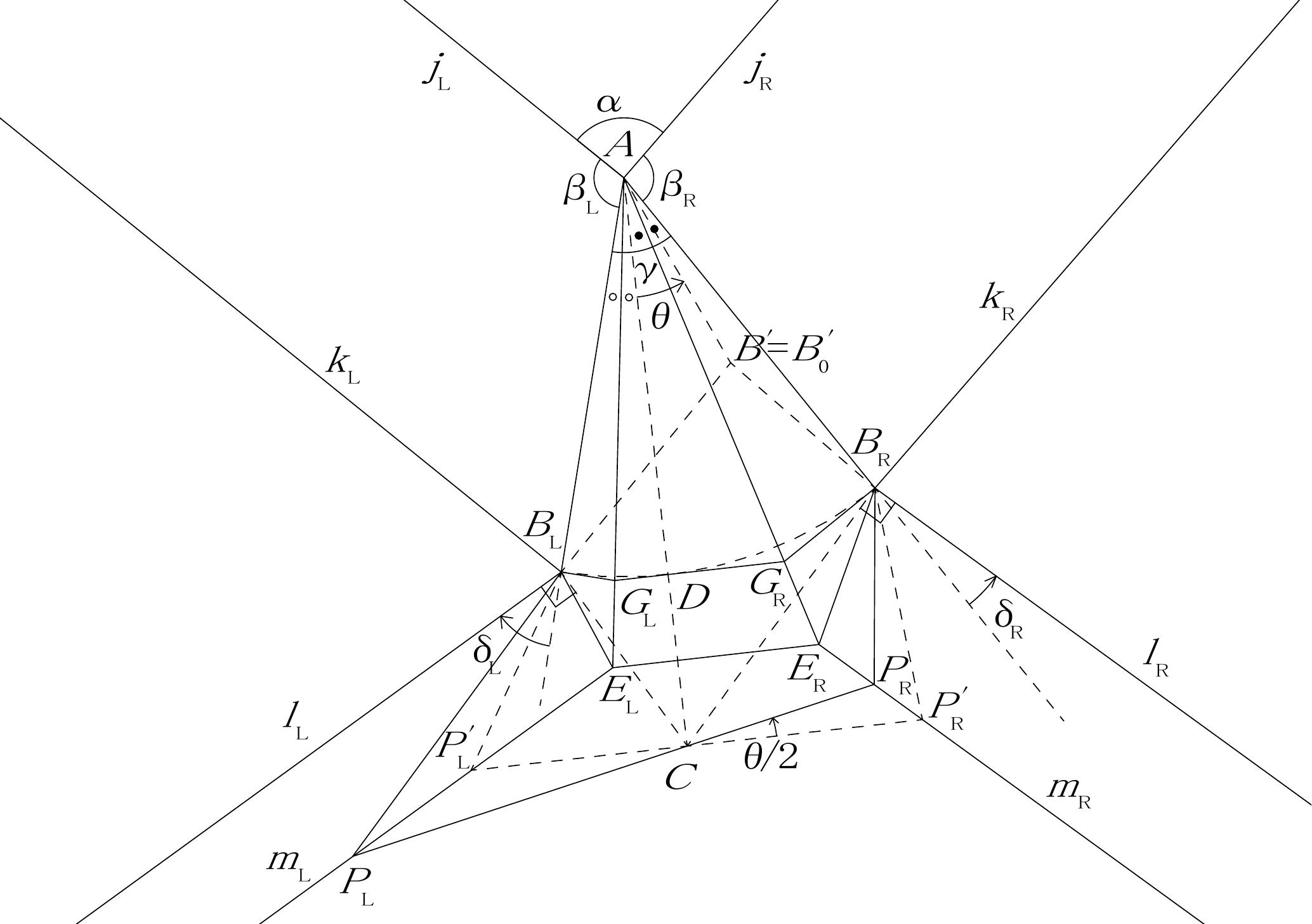}
    \caption{CP of a negative $3$D gadget by our second construction}
    \label{fig:CP_negative_2}
\end{figure}
\addtocounter{theorem}{1}
\begin{table}[h]
\begin{tabular}{c|c}
mountain folds&$j_\sigma ,m_\sigma ,AE_\sigma ,B_\sigma G_\sigma ,B_\sigma P_\sigma ,E_\Lt E_\Rt$\\ \hline
valley folds&$k_\sigma ,\ell_\sigma ,AB_\sigma ,B_\sigma E_\sigma ,G_\Lt G_\Rt ,P_\Lt P_\Rt$
\end{tabular}\vspace{0.5cm}
\caption{Assignment of mountain and valley folds to a negative $3D$ gadget by our second construction}
\label{tbl:assignment_negative_2}
\end{table}
The triangular flap $\triangle AG_\Lt G_\Rt$ may interfere with adjacent $3$D gadgets within the length
$\norm{DG_\sigma}$ on the side of $\sigma$, which is calculated as
\begin{equation*}
\norm{DG_\sigma}=\norm{B_\sigma G_\sigma}=\norm{AB}\tan\frac{\gamma_\sigma}{2}.
\end{equation*}

\subsection{Third construction}\label{subsec:negative_new_3}
We begin with the following construction.
\begin{construction}\label{const:negative_3}\rm
Consider a development as in Figure $\ref{fig:development_1}$, 
for which we require condition (i), (ii) and (iii.a)--(iii.c) of Construction $\ref{const:condition}$.
Then the crease pattern of our third negative $3$D gadget with prescribed simple outgoing pleats $(\ell_\Lt ,m_\Lt )$ and $(\ell_\Rt ,m_\Rt )$
is constructed as follows, where we regard $\sigma$ as taking both $\Lt$ and $\Rt$.
\begin{enumerate}
\item Determine  a point $D$ in the circular arc $B_\Lt B_\Rt$ with center $A$ 
such that $\psi_\Lt =\angle B_\Lt AC -\angle B_\Lt AD$ or $\rho_\Lt =\angle B_\Lt CA -\angle B_\Lt CD$ is calculated using Theorem $\ref{thm:existence_rho}$.
\item Let $E_\sigma$ be the intersection point of $m_\sigma$ and the bisector of $\angle B_\sigma AD$,
and redefine $m_\sigma$ to be a ray starting from $E_\sigma$ and going in the same direction as $\ell_\sigma$.
\item Draw a parallel to segment $E_\Lt E_\Rt$ through $D$, letting $G_\sigma$ be the intersection point of the parallel and segment $AE_\sigma$.
\item Draw a parallel to segment $E_\Lt E_\Rt$ through $C$, letting $P_\sigma$ be the intersection point of the parallel and ray $m_\sigma$.
\item The desired crease pattern is shown as the solid lines in Figure $\ref{fig:CP_negative_3}$, 
and the assignment of mountain and valley folds is given in Table $\ref{tbl:assignment_negative_3}$.
\end{enumerate}
\end{construction}
\begin{figure}[htbp]
\addtocounter{theorem}{1}
\centering\includegraphics[width=0.75\hsize]{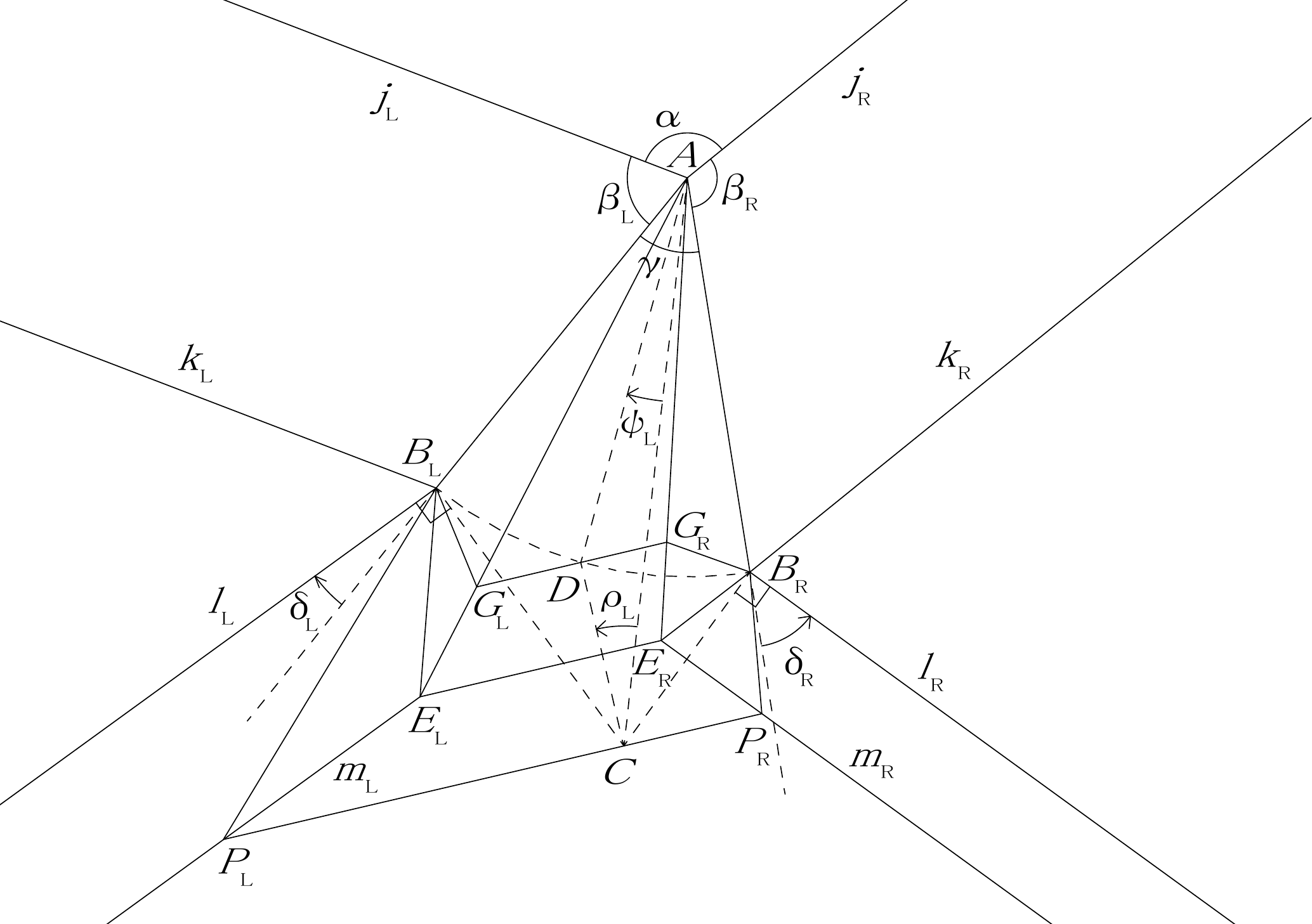}
    \caption{CP of a negative $3$D gadget by our third construction}
    \label{fig:CP_negative_3}
\end{figure}
\addtocounter{theorem}{1}
\begin{table}[h]
\begin{tabular}{c|c}
mountain folds&$j_\sigma ,m_\sigma ,AE_\sigma ,B_\sigma G_\sigma ,B_\sigma P_\sigma, E_\Lt E_\Rt$\\ \hline
valley folds&$k_\sigma ,\ell_\sigma ,AB_\sigma ,B_\sigma E_\sigma ,G_\Lt G_\Rt ,P_\Lt P_\Rt$
\end{tabular}\vspace{0.5cm}
\caption{Assignment of mountain and valley folds to a negative $3$D gadget by our third construction}
\label{tbl:assignment_negative_3}
\end{table}
Before finding $\psi_\Lt$ or $\rho_\Lt$ in $(1)$ of the above construction, we check the foldability of the gadget.
\begin{proposition}\label{prop:EBG=EBP}
Regardless of the position of $D$ in the circular arc $B_\Lt B_\Rt$ with center $A$,
we have
\begin{equation*}
\angle E_\sigma B_\sigma G_\sigma =\angle E_\sigma B_\sigma P_\sigma\quad\text{for }\sigma =\Lt ,\Rt .
\end{equation*}
\end{proposition}
\begin{proof}
For any $\psi_\Lt =\angle B_\Lt AD -\angle B_\Lt AC$, we have
\begin{align*}
\angle E_\sigma B_\sigma G_\sigma&=\angle E_\sigma DG_\sigma =\angle DE_\sigma E_{\sigma'},\quad\text{and}\\
\angle E_\sigma B_\sigma P_\sigma&=\angle E_\sigma DP_\sigma =\angle CE_\sigma E_{\sigma'},
\end{align*}
which are equal because $E_\Lt E_\Rt$ is a perpendicular bisector of $CD$.
\end{proof}
Thus in the resulting extrusion, using \cite{Doi20}, Lemma $4.4$ we have for both $\sigma =\Lt ,\Rt$ that
\begin{equation}\label{eq:kBP_1}
\begin{aligned}
\angle k_\sigma BP_\sigma&=\angle k_\sigma B_\sigma \ell_\sigma -\angle P_\sigma B_\sigma \ell_\sigma\\
&=(\beta_\sigma -\delta_\sigma )-(\pi +\delta_\sigma -\angle AB_\sigma E_\sigma -\angle E_\sigma B_\sigma P_\sigma )\\
&=(\beta_\sigma -\delta_\sigma )-\left\{\pi +\delta_\sigma -\left(\frac{\pi}{2}+\frac{\gamma_\sigma}{2}+\delta_\sigma+\frac{\psi_\sigma}{2}+\rho_\sigma\right)
-\left(\frac{\gamma_\sigma}{2}+\delta_\sigma-\frac{\psi_\sigma}{2}\right)\right\}\\
&=\beta_\sigma +\gamma_\sigma +\rho_\sigma -\frac{\pi}{2}.
\end{aligned}
\end{equation}
On the other hand, since $P_\sigma P_{\sigma'}$ overlaps with $G_\sigma G_{\sigma'}$,
$\angle k_\sigma BP_\sigma$ must be equal to the angle formed by the extensions of $j_\sigma$ and $G_\sigma G_{\sigma'}$, which is calculated as
\begin{equation}\label{eq:kBP_2}
\begin{aligned}
\angle ADG_\sigma -(\pi -\angle j_\sigma AD)
&=\left(\psi_\sigma +\rho_\sigma +\frac{\pi}{2}\right) -\{\pi -(\beta_\sigma +\gamma_\sigma -\psi_\sigma)\}\\
&=\beta_\sigma +\gamma_\sigma +\rho_\sigma -\frac{\pi}{2}.
\end{aligned}
\end{equation}
Hence we see from $\eqref{eq:kBP_1}$ and $\eqref{eq:kBP_2}$ together with the flat-foldability around $E_\sigma$ and $P_\sigma$ for $\sigma =\Lt ,\Rt$ that
we can construct successfully the supporting flap $\triangle AG_\Lt G_\Rt$ on the back side of the negative gadget
\emph{if there exists} an appropriate choice of $\psi_\Lt$ or $\rho_\Lt$.

Now let us find an equation which $\psi_\Lt$ or $\rho_\Lt$ satisfies.
In the resulting extrusion, noting that ridge $AB$ overlaps with $AD$, we need
\begin{equation}\label{eq:ADG=ABG}
\angle ADG_\Rt =\angle ABG_\Rt .
\end{equation}
This is rewritten as
\begin{equation}\label{eq:DA_EE=BA_EE}
e(\ora{DA})\cdot e(\ora{E_\Lt E_\Rt})=e(\ora{BA})\cdot e(\ora{E_\Lt E_\Rt}),
\end{equation}
where where $e(\ora{v})$ denotes a unit direction vector of $\ora{v}$, i.e., $e(\ora{v})=\ora{v}/\norm{v}$, and
$\ora{BA}$ is a $3$D vector in the resulting extrusion while $\ora{DA}$ is a $2$D vector in the development.
Suppose $\norm{AB}=1$ and the unit direction vectors of $k_\Rt$ and $k_\Lt$ are $(1,0,0)$ and $(\cos\alpha ,\sin\alpha ,0)$ respectively.
Using
\begin{equation*}
e(\ora{CA})=(\cos\alpha_\Rt ,\sin\alpha_\Rt ,0),
\end{equation*}
where $\alpha_\sigma =\pi -\beta_\sigma -\gamma_\sigma$, so that $\alpha_\Lt +\alpha_\Rt =\alpha$, we have that
\begin{align*}
e(\ora{DA})&=(\cos (\alpha_\Rt -\psi_\Lt ),\sin (\alpha_\Rt -\psi_\Lt ),0),\\
e(\ora{E_\Lt E_\Rt})&=\left(\cos\left(\alpha_\Rt +\rho_\Lt -\frac{\pi}{2}\right) ,\sin\left(\alpha_\Rt +\rho_\Lt -\frac{\pi}{2}\right),0\right)\\
&=(\sin (\beta_\Rt +\gamma_\Rt -\rho_\Lt ) ,\cos (\beta_\Rt +\gamma_\Rt -\rho_\Lt ) ,0) ,\\
e(\ora{BA})&=\left(-\cos\beta_\Rt ,\frac{\cos\alpha\cos\beta_\Rt -\cos\beta_\Lt}{\sin\alpha},\lambda (\alpha ,\beta_\Lt ,\beta_\Rt )\right) ,
\end{align*}
where $\lambda (\alpha ,\beta_\Lt ,\beta_\Rt )$ is the height of the extrusion for $\norm{AB}=1$
which is given in \cite{Doi19}, Lemma $2.2$, but we will not need its explicit value.
Also, $\tan\gamma_\sigma$ and $\tan\rho_\sigma$ are given in \cite{Doi19}, Lemma $4.1$ and \cite{Doi20}, equation $(4.21)$ respectively as
\begin{equation}
\begin{aligned}\label{eq:gamma_rho}
\tan\gamma_\sigma&=\frac{1-\cos\gamma +\sin\gamma\tan\delta_{\sigma'}}{\sin\gamma +\cos\gamma\tan\delta_{\sigma'}+\tan\delta_\sigma},\\
\tan\rho_\sigma&=\frac{\sin\psi_\sigma}{r-\cos\psi_\sigma},
\end{aligned}
\end{equation}
where $r$ is given in \cite{Doi19}, Proposition $5.2$ as
\begin{equation}\label{eq:r}
r=\frac{\norm{AC}}{\norm{AB}}=\frac{1}{\cos\gamma_\sigma -\sin\gamma_\sigma\tan\delta_\sigma}.
\end{equation}
Thus we calculate as
\begin{align}
\label{eq:DA_EE_1}
e(\ora{DA})\cdot e(\ora{E_\Lt E_\Rt})\sin\alpha&=-\sin(\beta_\Lt +\beta_\Rt +\gamma )\sin (\psi_\Lt +\rho_\Lt ),\\
\label{eq:BA_EE_1}
e(\ora{BA})\cdot e(\ora{E_\Lt E_\Rt})\sin\alpha&=\cos\beta_\Rt\cos (\beta_\Lt +\gamma_\Lt +\rho_\Lt ) -\cos\beta_\Lt\cos (\beta_\Rt +\gamma_\Rt -\rho_\Lt ).
\end{align}

Now define a function $\Gamma =\Gamma (\delta_\sigma ;\delta_{\sigma'})$ by
\begin{equation*}
\Gamma (\delta_\sigma ;\delta_{\sigma'})
=\tan^{-1}\left(\frac{1-\cos\gamma +\sin\gamma\tan\delta_{\sigma'}}{\sin\gamma +\cos\gamma\tan\delta_{\sigma'}+\tan\delta_\sigma}\right)
\end{equation*}
Then we can write $\gamma_\Lt$ and $\gamma_\Rt$ as
\begin{equation}\label{eq:Gamma_LR}
\gamma_\Lt =\Gamma (\delta_\Lt ;\delta_\Rt ),\quad\gamma_\Rt =\Gamma (\delta_\Rt ;\delta_\Lt ).
\end{equation}
Therefore, $r$ is considered as a function $r(\delta_\sigma ,\delta_{\sigma'})$ defined by
\begin{equation*}
r(\delta_\sigma ;\delta_{\sigma'})=\frac{1}{\cos\Gamma (\delta_\sigma ;\delta_{\sigma'})-\sin\Gamma (\delta_\sigma ;\delta_{\sigma'})\tan\delta_\sigma},
\end{equation*}
which satisfies 
\begin{equation}\label{eq:symmetry_r}
r(\delta_\Lt ;\delta_\Rt )=r(\delta_\Rt ;\delta_\Lt ).
\end{equation}

Let us return to equations $\eqref{eq:DA_EE_1}$ and $\eqref{eq:BA_EE_1}$.
Using the second expression of $\eqref{eq:gamma_rho}$, we have that
\begin{align*}
\cos (\gamma_\sigma\pm\rho_\Lt )&=(\cos\gamma_\sigma\mp\sin\gamma_\sigma \tan\rho_\Lt )\cos\rho_\Lt ,\\
\sin (\gamma_\sigma\pm\rho_\Lt )&=(\sin\gamma_\sigma\pm\cos\gamma_\sigma \tan\rho_\Lt )\cos\rho_\Lt .
\end{align*}
Consequently, we calculate $\eqref{eq:BA_EE_1}$ as
\begin{equation}\label{eq:BA_EE_2}
\begin{aligned}
e&(\ora{BA})\cdot e(\ora{E_\Lt E_\Rt})\times\frac{\sin\alpha}{\cos\rho_\Lt}\\
=&\{\cos\beta_\Lt (\cos\gamma_\Lt -\sin\gamma_\Lt\tan\rho_\Lt )-\sin\beta_\Lt (\sin\gamma_\Lt +\cos\gamma_\Lt\tan\rho_\Lt )\}\cos\beta_\Rt\\
&-\{\cos\beta_\Rt (\cos\gamma_\Rt +\sin\gamma_\Rt\tan\rho_\Lt )-\sin\beta_\Rt (\sin\gamma_\Rt +\cos\gamma_\Rt\tan\rho_\Lt )\}\cos\beta_\Lt\\
=&-\{(\cos\beta_\Lt\sin\gamma_\Lt +\sin\beta_\Lt\cos\gamma_\Lt )\cos\beta_\Rt
+(\cos\beta_\Rt\sin\gamma_\Rt +\sin\beta_\Rt\cos\gamma_\Rt )\cos\beta_\Lt\}\cdot\tan\rho_\Lt\\
&+\{(\cos\beta_\Lt\cos\gamma_\Lt -\sin\beta_\Lt\sin\gamma_\Lt )\cos\beta_\Rt -(\cos\beta_\Rt\cos\gamma_\Rt -\sin\beta_\Rt\sin\gamma_\Rt )\cos\beta_\Lt\}.
\end{aligned}
\end{equation}
Also, since it follows from the second expression of $\eqref{eq:gamma_rho}$ that
\begin{align*}
\frac{\sin (\psi_\Lt +\rho_\Lt )}{\cos\rho_\Lt}&=\sin\psi_\Lt +\cos\psi_\Lt\tan\rho_\Lt\\
&=(r-\cos\psi_\Lt )\tan\rho_\Lt +\cos\psi_\Lt \tan\rho_\Lt =r\tan\rho_\Lt ,
\end{align*}
we calculate $\eqref{eq:DA_EE_1}$ as
\begin{equation}\label{eq:DA_EE_2}
e(\ora{DA})\cdot e(\ora{E_\Lt E_\Rt})\times\frac{\sin\alpha}{\cos\rho_\Lt}=-r\sin(\beta_\Lt +\beta_\Rt +\gamma )\tan\rho_\Lt
\end{equation}
Equating $\eqref{eq:BA_EE_2}$ and $\eqref{eq:DA_EE_2}$, we finally obtain an equation
\begin{equation}\label{eq:DA_EE=BA_EE_alt_1}
(V_1(\beta_\Lt ,\delta_\Lt ;\beta_\Rt ,\delta_\Rt )-r(\delta_\Lt ;\delta_\Rt )\cdot V_2(\beta_\Lt ;\beta_\Rt ))\cdot\tan\rho_\Lt
=W(\beta_\Lt ,\delta_\Lt ;\beta_\Rt ,\delta_\Rt )
\end{equation}
which is equivalent to $\eqref{eq:DA_EE=BA_EE}$, where
\begin{equation}\label{eq:V_W_r}
\begin{aligned}
V_1=&(\cos\beta_\Lt\sin\gamma_\Lt +\sin\beta_\Lt\cos\gamma_\Lt )\cos\beta_\Rt +(\cos\beta_\Rt\sin\gamma_\Rt +\sin\beta_\Rt\cos\gamma_\Rt )\cos\beta_\Lt\\
=&\sin (\beta_\Lt +\gamma_\Lt )\cos\beta_\Rt +\sin (\beta_\Rt +\gamma_\Rt )\cos\beta_\Lt ,\\
V_2=&\sin(\beta_\Lt +\beta_\Rt +\gamma ),\\
W=&(\cos\beta_\Lt\cos\gamma_\Lt -\sin\beta_\Lt\sin\gamma_\Lt )\cos\beta_\Rt -(\cos\beta_\Rt\cos\gamma_\Rt -\sin\beta_\Rt\sin\gamma_\Rt )\cos\beta_\Lt \\
=&\cos (\beta_\Lt +\gamma_\Lt )\cos\beta_\Rt -\cos (\beta_\Rt +\gamma_\Rt )\cos\beta_\Lt ,\\
r=&\frac{1}{\cos\gamma_\Lt -\sin\gamma_\Lt\tan\delta_\Lt}=\frac{1}{\cos\gamma_\Rt -\sin\gamma_\Rt\tan\delta_\Rt}.
\end{aligned}
\end{equation}
We see easily from $\eqref{eq:V_W_r}$ and $\eqref{eq:Gamma_LR}$ that
\begin{equation}\label{eq:symmetry_VW}
\begin{aligned}
V_1(\beta_\Lt ,\delta_\Lt ;\beta_\Rt ,\delta_\Rt )&=V_1(\beta_\Rt ,\delta_\Rt ;\beta_\Lt ,\delta_\Lt ),\\
V_2(\beta_\Lt ;\beta_\Rt )&=V_2(\beta_\Rt ;\beta_\Lt ),\\
W(\beta_\Lt ,\delta_\Lt ;\beta_\Rt ,\delta_\Rt )&=-W(\beta_\Rt ,\delta_\Rt ;\beta_\Lt ,\delta_\Lt )
\end{aligned}
\end{equation}
along with $\eqref{eq:symmetry_r}$.
As $\psi_\Lt$ moves in the range $(-\gamma_\Rt ,\gamma_\Lt )$, $\rho_\Lt$ moves in the range
\begin{equation*}
\rho_\Lt\in\left(\gamma_\Rt +\delta_\Rt -\frac{\pi}{2}, \frac{\pi}{2}-\gamma_\Lt -\delta_\Lt \right) .
\end{equation*}

\begin{theorem}\label{thm:existence_rho}
Suppose $\alpha ,\beta_\Lt ,\beta_\Rt ,\delta_\Lt$ and $\delta_\Rt$ satisfy
conditions $(\mathrm{i})$, $(\mathrm{ii})$ and $(\mathrm{iii.a})$--$(\mathrm{iii.c})$ of Construction $\ref{const:condition}$.
Then there exists a unique solution $\rho_\Lt\in (\gamma_\Rt +\delta_\Rt -\pi /2, \pi /2-\gamma_\Lt -\delta_\Lt )$
of equation $\eqref{eq:DA_EE=BA_EE_alt_1}$, which is given by
\begin{equation*}
\rho_\Lt =\tan^{-1}\left(\frac{W}{V_1-rV_2}\right) ,
\end{equation*}
where $V_1,V_2,W$ and $r$ are given by $\eqref{eq:V_W_r}$.
Thus there exists a unique solution $\psi_\Lt\in (-\gamma_\Rt ,\gamma_\Lt )$ of equation $\eqref{eq:ADG=ABG}$, which is given by
\begin{equation*}
\psi_\Lt =\sin^{-1}(r\sin\rho_\Lt )-\rho_\Lt .
\end{equation*}
\end{theorem}

Before proving the theorem, we begin with the following two lemmas.
\begin{lemma}\label{lem:Phi>0}
Define $\Phi (\rho_\Lt ;\beta_\Lt ,\delta_\Lt ;\beta_\Rt ,\delta_\Rt )$ by
\begin{equation*}
\Phi (\rho_\Lt ;\beta_\Lt ,\delta_\Lt ;\beta_\Rt ,\delta_\Rt )=(V_1-rV_2)\tan\rho_\Lt -W.
\end{equation*}
Then we have 
\begin{equation}\label{eq:Phi_L}
\begin{aligned}
\Phi&\left(\frac{\pi}{2}-\gamma_\Lt -\delta_\Lt;\beta_\Lt ,\delta_\Lt ;\beta_\Rt ,\delta_\Rt\right)\\
&=\frac{2\sin (\gamma /2)}{\sin\gamma_\Lt +\cos\gamma_\Lt\tan\delta_\Lt}
\cdot (\sin\beta_\Lt -\tan\delta_\Lt\cos\beta_\Lt )\sin\left(\beta_\Rt +\frac{\gamma}{2}\right) >0.
\end{aligned}
\end{equation}
\end{lemma}
\begin{proof}[Proof of Lemma $\ref{lem:Phi>0}$.]
Using
\begin{equation*}
\tan\left(\frac{\pi}{2}-\gamma_\Lt -\delta_\Lt\right) =\frac{1}{\tan (\gamma_\Lt +\delta_\Lt )}
=\frac{1-\tan\gamma_\Lt\tan\delta_\Lt}{\tan\gamma_\Lt +\tan\delta_\Lt}
=\frac{\cos\gamma_\Lt -\sin\gamma_\Lt\tan\delta_\Lt}{\sin\gamma_\Lt +\cos\gamma_\Lt\tan\delta_\Lt},
\end{equation*}
we have 
\begin{equation*}
r\tan\left(\frac{\pi}{2}-\gamma_\Lt -\delta_\Lt\right) =\frac{1}{\sin\gamma_\Lt +\cos\gamma_\Lt\tan\delta_\Lt}.
\end{equation*}
Thus we can calculate $(\sin\gamma_\Lt +\cos\gamma_\Lt\tan\delta_\Lt )\cdot\Phi$ as
\begin{equation}\label{eq:(s+ct)Phi_L}
\begin{aligned}
(\sin\gamma_\Lt &+\cos\gamma_\Lt\tan\delta_\Lt )\cdot\Phi\left(\frac{\pi}{2}-\gamma_\Lt -\delta_\Lt;\beta_\Lt ,\delta_\Lt ;\beta_\Rt ,\delta_\Rt\right)\\
&=V_1(\cos\gamma_\Lt -\sin\gamma_\Lt\tan\delta_\Lt)-V_2-W(\sin\gamma_\Lt +\cos\gamma_\Lt\tan\delta_\Lt )\\
&=(V_1\cos\gamma_\Lt -V_2-W\sin\gamma_\Lt )-(V_1\sin\gamma_\Lt +W\cos\gamma_\Lt )\tan\delta_\Lt .
\end{aligned}
\end{equation}
To simplify the expression, we set
\begin{align*}
CC&=\cos\beta_\Lt\cos\beta_\Rt ,\quad SC=\sin\beta_\Lt\cos\beta_\Rt ,\quad CS=\cos\beta_\Lt\sin\beta_\Rt ,\quad SS=\sin\beta_\Lt\sin\beta_\Rt ,\\
c_\sigma&=\cos\delta_\sigma ,\quad s_\sigma =\sin\delta_\sigma\quad\text{for }\sigma =\Lt ,\Rt .
\end{align*}
Then $V_1,V_2$ and $W$ are written as
\begin{align*}
V_1&=CC(s_\Lt +s_\Rt )+SC s_\Lt +CS c_\Rt ,\\
V_2&=CC\sin\gamma +SC\cos\gamma +CS\cos\gamma -SS\sin\gamma ,\\
W&=CC(c_\Lt -c_\Rt )-SC s_\Lt +CS s_\Rt ,
\end{align*}
so that we have
\begin{align*}
V_1\cos\gamma_\Lt -W\sin\gamma_\Lt&=CC(c_\Lt s_\Rt +s_\Lt c_\Rt )+SC(c_\Lt^2+s_\Lt^2)+CS(c_\Lt c_\Rt -s_\Lt s_\Rt )\\
&=CC\sin\gamma +SC+CS\cos\gamma ,\\
V_1\cos\gamma_\Lt -W\sin\gamma_\Lt -V_2&=SC(1-\cos\gamma )+SS\sin\gamma\\
&=\sin\beta_\Lt\{\cos\beta_\Rt (1-\cos\gamma )+\sin\beta_\Rt\sin\gamma\} ,\\
V_1\sin\gamma_\Lt +W\cos\gamma_\Lt&=CC\{(s_\Lt + s_\Rt )s_\Lt +(c_\Lt -c_\Rt )c_\Lt\} +CS(s_\Lt c_\Rt +c_\Lt s_\Rt )\\
&=CC(1-\cos\gamma )+CS\sin\gamma\\
&=\cos\beta_\Lt\{\cos\beta_\Rt (1-\cos\gamma )+\sin\beta_\Rt\sin\gamma\} .
\end{align*}
Consequently, the right-hand side of $\eqref{eq:(s+ct)Phi_L}$ is calculated as
\begin{align*}
(\sin\beta_\Lt -\cos\beta_\Lt\tan\delta_\Lt )&\{\cos\beta_\Rt (1-\cos\gamma )+\sin\beta_\Rt\sin\gamma\}\\
&=2\sin\frac{\gamma}{2}\cdot (\sin\beta_\Lt -\cos\beta_\Lt\tan\delta_\Lt )\sin\left(\beta_\Rt +\frac{\gamma}{2}\right) ,
\end{align*}
which is positive because of $0<\gamma /2<\pi /2$, $0<\delta_\Lt <\beta_\Lt <\pi$ and $\beta_\Rt +\gamma /2<\pi$.
This completes the proof of Lemma $\ref{lem:Phi>0}$.
\end{proof}
\begin{lemma}\label{lem:Phi<0}
We have 
\begin{equation}\label{eq:Phi_R}
\begin{aligned}
\Phi&\left(\gamma_\Rt +\delta_\Rt -\frac{\pi}{2};\beta_\Lt ,\delta_\Lt ;\beta_\Rt ,\delta_\Rt\right)\\
&=-\frac{2\sin (\gamma /2)}{\sin\gamma_\Rt +\cos\gamma_\Rt\tan\delta_\Rt}
\cdot (\sin\beta_\Rt -\cos\beta_\Rt\tan\delta_\Rt )\sin\left(\beta_\Lt +\frac{\gamma}{2}\right) <0.
\end{aligned}
\end{equation}
\end{lemma}
\begin{proof}[Proof of Lemma $\ref{lem:Phi<0}$.]
We see immediately from the definition of $\Phi$, $\eqref{eq:symmetry_r}$, $\eqref{eq:V_W_r}$ and $\eqref{eq:symmetry_VW}$ that
\begin{equation*}
\Phi\left(\gamma_\Rt +\delta_\Rt -\frac{\pi}{2};\beta_\Lt ,\delta_\Lt ;\beta_\Rt ,\delta_\Rt\right)
=-\Phi\left(\frac{\pi}{2}-\gamma_\Rt +\delta_\Rt ;\beta_\Rt ,\delta_\Rt ;\beta_\Lt ,\delta_\Lt\right) ,
\end{equation*}
and the right-hand side of $\eqref{eq:Phi_R}$ is obtained by interchanging the suffices $\Lt$ and $\Rt$ in $\eqref{eq:Phi_L}$.
\end{proof}
\begin{proof}[Proof of Theorem $\ref{thm:existence_rho}$.]
By the intermediate value theorem, we see from the above two lemmas that
there exists $\rho_\Lt\in (\gamma_\Rt +\delta_\Rt -\pi /2, \pi /2-\gamma_\Lt -\delta_\Lt )$ such that $\Phi =0$, 
which is a desired solution of $\eqref{eq:DA_EE=BA_EE_alt_1}$.
Also, we see that $V_1-rV_2>0$, so that $(V_1-rV_2)\tan\rho_\Lt$ is monotonically increasing with respect to $\rho_\Lt$,
from which follows the uniqueness of the solution $\rho_\Lt$.
\end{proof}
\begin{corollary}\label{cor:negative_3_engaging}
For any positive $3$D gadget with $(\alpha ,\beta_\Lt ,\beta_\Rt ,\delta_\Lt ,\delta_\Rt )$
satisfying conditions $(\mathrm{i})$, $(\mathrm{ii})$ and $(\mathrm{iii.a})$--$(\mathrm{iii.c})$ of Construction $\ref{const:condition}$, 
we can apply Construction $\ref{const:negative_3}$ to obtain a unique negative $3$D gadget
with the same net and prescribed outgoing pleats as the positive one in the development.
This negative $3$D gadget engages with the original positive one.
\end{corollary}
If we consider the horizontally flipped crease pattern of the negative $3$D gadget
resulting in Corollary $\ref{cor:negative_3_engaging}$, we have the following result, which solves Problem $\ref{prob:existence_negative}$.
\begin{corollary}
For any positive $3$D gadget with $(\alpha ,\beta_\Lt ,\beta_\Rt ,\delta_\Lt ,\delta_\Rt )$
satisfying conditions $(\mathrm{i})$, $(\mathrm{ii})$ and $(\mathrm{iii.a})$--$(\mathrm{iii.c})$ of Construction $\ref{const:condition}$,
we can apply Construction $\ref{const:negative_3}$ to obtain a unique negative $3$D gadget with the same outgoing pleats such that
the top and two side faces are negatively extruded, or in other words, sunk to the opposite direction.
\end{corollary}

Now we set
\begin{align*}
\beta_+&=\beta_\Lt +\beta_\Rt ,\quad\beta_-=\beta_\Rt -\beta_\Lt ,\quad\beta'_+=\beta_++\frac{\gamma}{2}\quad\text{and}\\
\quad\omega&=\gamma_\Lt -\frac{\gamma}{2}=\frac{\gamma}{2}-\gamma_\Rt .
\end{align*}
Then a straightforward calculation gives an alternative expression of $V_1,V_2$ and $W$ as
\begin{equation}\label{eq:V_W_alt}
\begin{aligned}
V_1&=\left\{\sin\left(\beta_++\frac{\gamma}{2}\right) +\cos\beta_-\sin\frac{\gamma}{2}\right\}\cos\omega +\sin\beta_-\sin\frac{\gamma}{2}\sin\omega\\
&=\left(\sin\beta'_++\cos\beta_-\sin\frac{\gamma}{2}\right)\cos\omega +\sin\beta_-\sin\frac{\gamma}{2}\sin\omega ,\\
V_2&=\sin (\beta_++\gamma )=\sin\left(\beta'_++\frac{\gamma}{2}\right),\\
W&=\sin\beta_-\sin\frac{\gamma}{2}\cos\omega -\left\{\sin\left(\beta_++\frac{\gamma}{2}\right)+\cos\beta_-\sin\frac{\gamma}{2}\right\}\sin\omega\\
&=\sin\beta_-\sin\frac{\gamma}{2}\cos\omega -\left(\sin\beta'_++\cos\beta_-\sin\frac{\gamma}{2}\right)\sin\omega .
\end{aligned}
\end{equation}
For an alternative expression of $r$, we observe from $\eqref{eq:r}$ that $1/r$ is written in two ways as
\begin{equation}\label{eq:1/r}
\cos\left(\frac{\gamma}{2}+\omega\right)-\sin\left(\frac{\gamma}{2}+\omega\right)\cdot\tan\delta_\Lt
=\frac{1}{r}=\cos\left(\frac{\gamma}{2}-\omega\right)-\sin\left(\frac{\gamma}{2}-\omega\right)\cdot\tan\delta_\Rt .
\end{equation}
Multiplying the left- and the right-hand side of $\eqref{eq:1/r}$ by $\sin (\gamma /2-\omega)$ and $\sin (\gamma /2+\omega )$ respectively,
and adding the resulting terms together, we have
\begin{equation*}
\frac{1}{r}\cdot 2\sin\frac{\gamma}{2}\cos\omega =\sin\gamma +\frac{1}{2}(\cos\gamma -\cos 2\omega )(\tan\delta_\Lt +\tan\delta_\Rt ),
\end{equation*}
so that
\begin{equation}\label{eq:r_alt}
r=\frac{2\sin (\gamma /2)\cos\omega}{\sin\gamma +\{(1+\cos\gamma )/2-\cos^2\omega\} (\tan\delta_\Lt +\tan\delta_\Rt )}.
\end{equation}
Consequently, we can use $\eqref{eq:V_W_alt}$ and $\eqref{eq:r_alt}$ to obtain an expression of $\tan\rho_\Lt$ without arc tangents as follows.
\begin{proposition}
In Theorem $\ref{thm:existence_rho}$, we can express $\tan\rho_\Lt$ in terms of $t=\tan\omega$ as
\begin{align*}
&\tan\rho_\Lt =\frac{W}{V_1-rV_2}\\
&=\frac{\displaystyle\sin\beta_--t\cdot\left(\frac{\sin\beta'_+}{\sin (\gamma /2)}+\cos\beta_-\right)}
{\displaystyle\left(\frac{\sin\beta'_+}{\sin (\gamma /2)}+\cos\beta_-\right)+t\sin\beta_-
-\frac{2\sin (\beta'_++\gamma /2)}{\sin\gamma +\{(1+\cos\gamma )/2-1/(1+t^2)\} (\tan\delta_\Lt +\tan\delta_\Rt )}},
\end{align*}
where $t$ is given by
\begin{equation*}
t=\tan\omega =\frac{\tan\delta_\Rt -\tan\delta_\Lt}{2+(\tan\delta_\Lt +\tan\delta_\Rt )/\tan (\gamma /2)}.
\end{equation*}
\end{proposition}
Then the following result is immediate by a straightforward calculation.
\begin{proposition}\label{prop:rho_alt}
Suppose $\delta_\Lt =\delta_\Rt =0$.
Then we have
\begin{equation*}
\tan\rho_\Lt =\frac{\sin\beta_-}{\cos\beta_--\cos\beta'_+/\cos (\gamma /2)}.
\end{equation*}
In particular, if $\beta_\Lt +\beta_\Rt =\pi$, then we have $\rho_\Lt =(\beta_\Rt -\beta_\Lt )/2$.
\end{proposition}
As in the previous sections, the triangular flap $\triangle AG_\Lt G_\Rt$ may interfere with adjacent $3$D gadgets within the length
$K_{\nega ,\sigma}(B)$ on the side of $\sigma$, which are calculated as
\begin{equation}\label{eq:interference_negative_3}
\begin{aligned}
K_{\nega ,\Lt}(B)&=\norm{DG_\Lt}=\frac{\norm{AB}}{\cos (\psi_\Lt +\rho_\Lt )/\tan (\gamma_\Lt /2)+\sin (\psi_\Lt +\rho_\Lt )},\\
K_{\nega ,\Rt}(B)&=\norm{DG_\Rt}=\frac{\norm{AB}}{\cos (\psi_\Lt +\rho_\Lt )/\tan (\gamma_\Rt /2)-\sin (\psi_\Lt +\rho_\Lt )},
\end{aligned}
\end{equation}
where we used 
\begin{equation*}
\angle ADG_\Lt =\pi /2+\psi_\Lt +\rho_\Lt ,\quad\angle ADG_\Rt =\pi /2-\psi_\Lt -\rho_\Lt .
\end{equation*}
In particular, if $\beta_\Lt =\beta_\Rt$ and $\delta_\Lt =\delta_\Rt$, then $\psi_\Lt =\rho_\Lt =0$, so that we have
\begin{equation*}
K_{\nega ,\Lt}=K_{\nega ,\Rt}=\norm{AB}\tan\frac{\gamma}{4}.
\end{equation*}
which do not depend on either $\beta_\sigma$ or $\delta_\sigma$.
\begin{remark}\rm
As with our first construction, we see from Lemmas $\ref{lem:Phi>0}$ and $\ref{lem:Phi<0}$
that we can also include the case $\alpha =\beta_\Lt +\beta_\Rt$ or equivalently $\beta_\Lt +\beta_\Rt +\gamma /2=\pi$ in Construction $\ref{const:negative_3}$
because $\Phi$ in those lemmas do not vanish and thus $\phi_\Lt$ takes a value strictly between $0$ and $\gamma$.
In this case the resulting negative gadget is flat, and we have the flat-foldability condition around $B_\sigma$
\begin{equation*}
\angle AB_\sigma G_\sigma +\angle E_\sigma B_\sigma P_\sigma +\angle k_\sigma B_\sigma \ell_\sigma =\pi
=\angle E_\sigma B_\sigma G_\sigma +\angle P_\sigma B_\sigma \ell_\sigma +\angle AB_\sigma k_\sigma ,
\end{equation*}
which is rewritten as
\begin{equation}\label{eq:flat-foldability_negative_flat}
\angle AB_\sigma E_\sigma +\angle k_\sigma B_\sigma \ell_\sigma =\pi =\angle AB_\sigma k_\sigma +\angle E_\sigma B_\sigma \ell_\sigma
\end{equation}
by Proposition $\ref{prop:EBG=EBP}$.
On the other hand, there is a unique positive flat gadget with the same $\alpha$, $\beta_\sigma$ and $\delta_\sigma$ for both $\sigma =\Lt ,\Rt$,
where the assignment of mountain and valley folds is given in Table $\ref{tbl:assignment_positive_flat}$.
For this positive gadget, it follows from \cite{Doi20}, Section $3$ that $\angle B_\sigma AE_\sigma =\zeta_\sigma$ with $\zeta_\Lt +\zeta_\Rt =\gamma /2$,
where $\zeta_\sigma$ is given geometrically in Definition $3.2$ and numerically in Proposition $3.6$ of \cite{Doi20}.
This also satisfies the flat-foldability condition $\eqref{eq:flat-foldability_negative_flat}$ around $B_\sigma$.
Thus we must have $\angle B_\sigma AE_\sigma =\zeta_\sigma$ for the \emph{negative} flat gadget,
which does not hold if $\alpha <\beta_\Lt +\beta_\Rt$ where there are an infinite number of positive $3$D gadgets.
\end{remark}
\addtocounter{theorem}{1}
\begin{table}[h]
\begin{tabular}{c|c}
mountain folds&$j_\sigma ,\ell_\sigma ,AB_\sigma ,B_\sigma E_\sigma$\\ \hline
valley folds&$k_\sigma ,m_\sigma ,AE_\sigma ,E_\Lt E_\Rt$
\end{tabular}\vspace{0.5cm}
\caption{Assignment of mountain and valley folds to a unique positive flat gadget for $\alpha =\beta_\Lt +\beta_\Rt$}
\label{tbl:assignment_positive_flat}
\end{table}

\subsection{Interferences caused by negative $3$D gadgets}\label{subsec:interferences}
Here we only consider flat-back positive $3$D gadgets given in \cite{Doi19} and \cite{Doi20}, and negative ones by our third construction,
but other cases are similar.
As discussed in \cite{Doi19}, Section $5$ and \cite{Doi20}, Section $5$, for two positive $3$D gadgets sharing a side face,
an interference may occur in the bottom edge.
However, since our negative $3$D gadgets have a supporting triangle which does not necessarily overlap with their side faces,
an interference between two adjacent negative ones may occur at the intersection point of the extentions of the bases of their supporting triangles.

Suppose two negative $3$D gadgets share a side face $A_1 A_2 B_2 B_1$, where $A_1$ and $A_2$ may be identical, and $B_1$ is located to the left of $B_2$.
We will add subscripts `1' and `2' to the quantities and points belonging to the left and the right gadget respectively.
Set $\theta_{1,\Rt}=\angle B_2 B_1 G_{1,\Rt}$ and $\theta_{2,\Lt}=\angle B_1 B_2 G_{2,\Lt}$ in the resulting extrusion,
which are calculated using $\eqref{eq:kBP_1}$ as
\begin{align*}
\theta_{1,\Rt}=\beta_{1,\Rt}+\gamma_{1,\Rt}-\rho_{1,\Lt}-\pi /2,\quad\theta_{2,\Lt}=\beta_{2,\Lt}+\gamma_{2,\Lt}+\rho_{2,\Lt}-\pi /2.
\end{align*}
Then the minimum length $K_{\nega ,\min}(B_1 B_2)$ of $\norm{B_1 B_2}$ such that no interference occurs is given by
\begin{align*}
K_{\nega ,\min}(B_1 B_2)=
\begin{dcases}
0&\text{if }\theta_{1,\Rt}+\theta_{2,\Lt}\geqslant\pi ,\\
\min\{K_{\nega ,\Rt}(B_1)\sin\theta_{1,\Rt},K_{\nega ,\Lt}(B_2)\sin\theta_{2,\Lt}\}&\\
\times\left(\frac{1}{\tan\theta_{1,\Rt}}+\frac{1}{\tan\theta_{2,\Lt}}\right)&\text{if }\theta_{1,\Rt}+\theta_{2,\Lt}<\pi ,
\end{dcases}
\end{align*}
where $K_{\nega ,\Rt}(B_1)$ and $K_{\nega ,\Lt}(B_2)$ are given by $\eqref{eq:interference_negative_3}$.
Hence there are no interferences between the two negative $3$D gadgets if and only if $\norm{B_1 B_2}\geqslant K_{\nega ,\min}(B_1 B_2)$.
Of course, we have to consider another minimum of $\norm{B_1 B_2}$ which is given by the length of $B_1 B_2$
when we set $A_1=A_2$ for $\beta_{1,\Lt}+\beta_{2,\Rt}\geqslant\pi$.

Interferences caused by a flat-back positive $3$D gadget given in \cite{Doi20} and a negative one which share a side face, may be rather complicated to treat.
\begin{example}\label{ex:prism_VI}\rm
As an example, let us consider the resuting extruded prism of the crease pattern shown in Figure $\ref{fig:CP_prism_VI}$,
where the divided parts of the top face and the side faces are lightly and dark shaded respectively.
In the figure, we set
\begin{align*}
\alpha_1&=\alpha_2=\pi /3,\quad\beta_{1,\Lt}=\beta_{1,\Rt}=\beta_{2,\Lt}=\beta_{2,\Rt}=\pi /2,\quad\text{so that}\quad\gamma_1=\gamma_2=2\pi /3,\\
\delta_{1,\Lt}&=\delta_{1,\Rt}=\delta_{2,\Lt}=\delta_{2,\Rt}=0,\quad\psi_{1,\Lt}=\psi_{2,\Lt}=0,\quad\text{and}\\
\norm{A_1 B_{1,\Rt}}&=\norm{A_1 B_{2,\Lt}},\quad\text{so that}\quad A_1 B_{2,\Lt}A_2 B_{1,\Rt}\text{ is a square}.
\end{align*}
Then we see that there are no interferences.
However, let us see how interferences occur if we make $\norm{A_1 B_{2,\Lt}}$ (and accordingly $\norm{A_2 B_{1,\Rt}}$) shorter.
Suppose $\norm{A_1 B_1}=1$, so that $\norm{A_2 B_2}=1$.
Then we have the following.
\begin{itemize}
\item The flap formed by kite $A_1 B_{1,\Rt}G_{1,\Rt}D_1$ of the positive $3$D gadget collides with $A_2$
if $\norm{A_2 B_{1,\Rt}}<\norm{B_{1,\Rt}G_{1,\Rt}}=1/\sqrt{3}$.
\item The flap formed by kite $A_2 B_{2,\Lt}G_{2,\Lt}D_2$ of the negative $3$D gadget collides with $j_{1,\Lt}$
if $\norm{A_1 B_{2,\Lt}}<\norm{B_{2,\Lt}G_{2,\Lt}}=1/\sqrt{3}$.
\item The crease which is the reflection of $j_{1,\Lt}$ across $\ell_{2,\Lt}$
collides with the flap formed by kite $A_2 B_{2,\Lt}G_{2,\Lt}D_2$ of the negative $3$D gadget
if $\norm{A_1 B_{2,\Lt}}<\norm{B_{2,\Lt}G_{2,\Lt}}=1/\sqrt{3}$.
\item The crease which is the reflection of $j_{2,\Rt}$ across $\ell_{1,\Rt}$
collides with the flap formed by kite $A_2 B_{2,\Lt}G_{2,\Lt}D_2$ of the negative $3$D gadget
if $\norm{A_2 B_{1,\Rt}}<\norm{B_{1,\Rt}G_{1,\Rt}}=1/\sqrt{3}$.
\end{itemize}
Thus in any of the above four cases, there are no interferences if and only if $\norm{A_1 B_2}\geqslant 1/\sqrt{3}$.
We show the crease pattern for $\norm{A_1 B_2}=1/\sqrt{3}$ in Figure $\ref{fig:CP_prism_VI_min}$.
\end{example}
\begin{figure}[htbp]
\addtocounter{theorem}{1}
\centering\includegraphics[width=0.75\hsize]{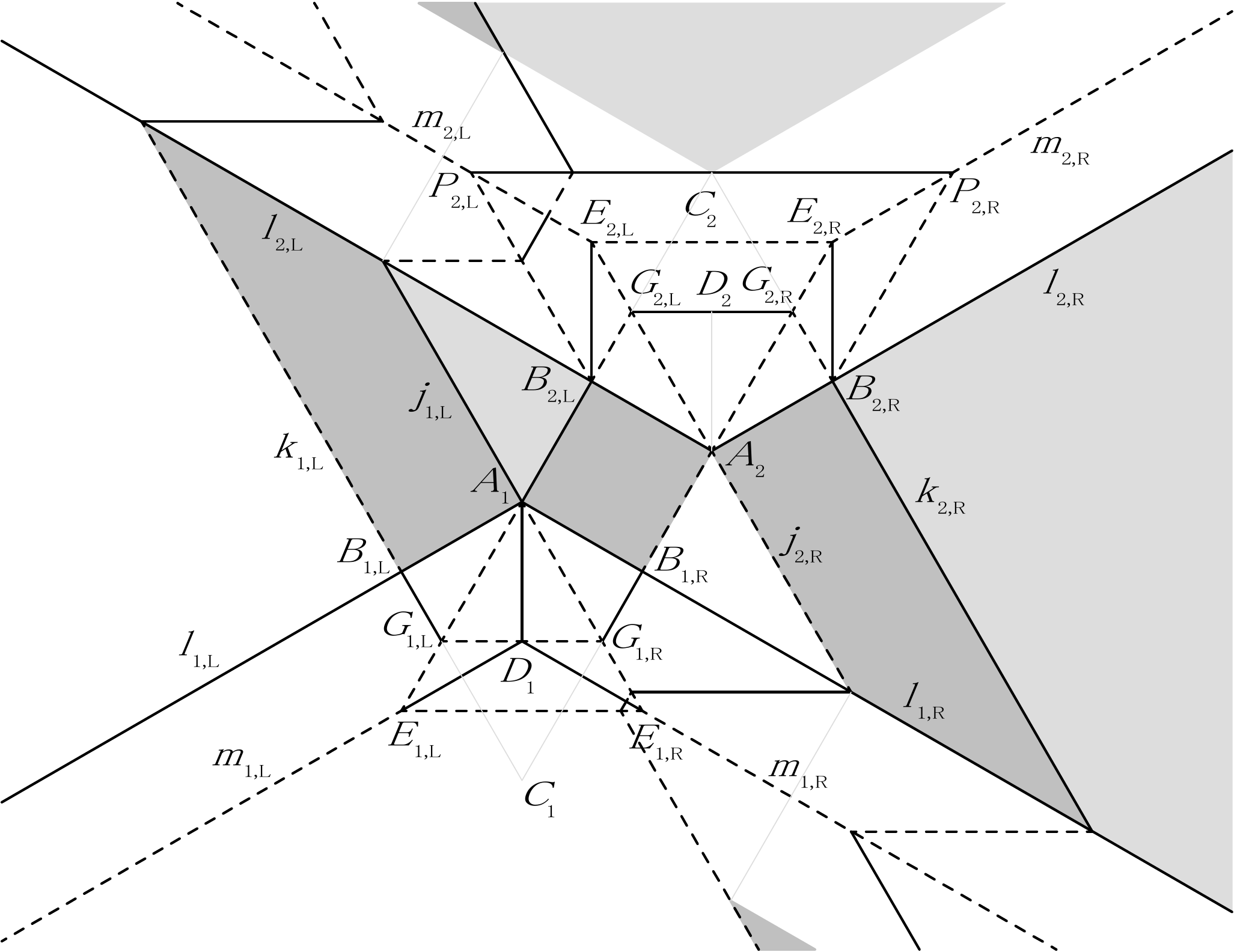}
    \caption{CP of a prism extruded with a positive and a negative $3$D gadget such that $\norm{A_1 B_2}=\norm{A_1 B_1}$}
    \label{fig:CP_prism_VI}
\end{figure}
\begin{figure}[htbp]
\addtocounter{theorem}{1}
\centering\includegraphics[width=0.75\hsize]{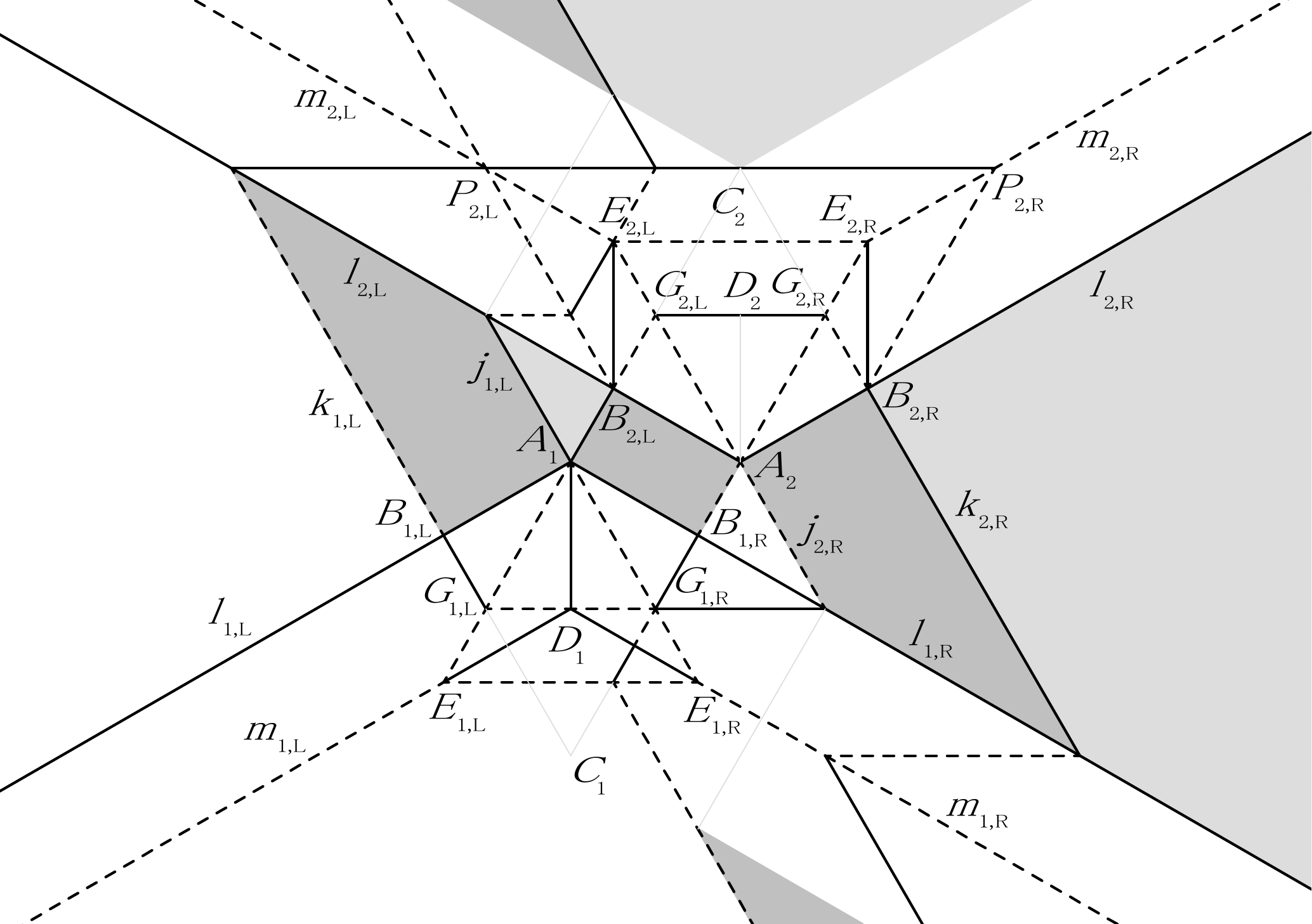}
    \caption{CP of a prism extruded with a positive and a negative $3$D gadget such that $\norm{A_1 B_2}=\norm{A_1 B_1}/\sqrt{3}$}
    \label{fig:CP_prism_VI_min}
\end{figure}
Thus in the presence of both positive and negative $3$D gadgets,
there may be not only interferences between a $3$D gadget and another one or a side face which is not necessarily adjacent to the first one,
but also ones beween a $3$D gadget and an edge of a side face reflected by an outgoing pleat as in the third and fourth cases of Example $\ref{ex:prism_VI}$.
Using the repeating gadgets given in \cite{Doi20}, Section $7$ for positive ones and those given in the next section for negative ones,
or changing the angles of the outgoing pleats, we may be able to avoid these interferences.
However, the more $3$D gadgets we have, the more difficult it becomes to handle their interferences.
For the treatment of intersections of outgoing pleats, see \cite{Doi19}, Section $9$.

\section{Repetition and division of negative $3$D gadgets}\label{sec:repetition}
As we explained in \cite{Doi19} and \cite{Doi20}, division and repetition of $3$D gadgets are the same but the gadget we start from.
The repeating $3$D gadgets presented in \cite{Doi20}, which have flat back sides, can also be used for negative $3$D gadgets with a little modification.
The creases of the negative $3$D gadget does not affect on those of the repeating gadgets in the resulting crease pattern
because the supporting triangle overlaps with the repeating gadgets only at the ridge.
As in \cite{Doi19} and \cite{Doi20}, we shall treat the cases with $\delta_\Lt =\delta_\Rt =0$.
The general cases with $\delta_\Lt ,\delta_\Rt\geqslant 0$ are rather complicated, and thus will be treated elsewhere.
In the construction below, we will use a negative $3$D gadget by our third construction as the lowest gadget.
However, we can use any of our constructions as long as it is applicable.
\begin{construction}\label{const:division_new}\rm
Consider a development as shown in Figure $\ref{fig:development_division_negative}$, 
for which we require conditions (i), (ii) and (iii) of Construction $\ref{const:condition}$.
Suppose $p_1,\dots ,p_d>0$ with $p_1+\dots +p_d=d$ satisfy
\begin{equation}\label{ineq:pn_qn}
\begin{aligned}
\frac{r+1}{2}\cdot p_n\leqslant q_n\quad\text{for }n=2,\dots ,d,\text{ where }q_n=p_1+\dots +p_n.
\end{aligned}
\end{equation}
(The meaning of $\eqref{ineq:pn_qn}$ will be explained in Remark $\ref{rem:meaning_ineq_qn}$.)
In particular, in the case of equal division $p_1=\dots =p_n=1$, $\eqref{ineq:pn_qn}$ holds if and only if
\begin{equation*}
\gamma\leqslant 2\cos^{-1}(1/3)=2\tan^{-1}(2\sqrt{2})\approx 141.06^\circ .
\end{equation*}
Then the crease pattern of the proportional division of a negative $3$D gadget by our third construction
into $d$ gadgets in the ratio $p_1:\dots :p_d$ from the bottom with $p_1+\dots +p_d=d$ is constructed as follows,
where we regard $\sigma$ as taking both $\Lt$ and $\Rt$.
\begin{enumerate}
\item Divide segment $CA$ (not $AC$) into $d$ parts proportionally in the ratio $p_1:\dots :p_d$,
letting $A^{(1)},\dots ,A^{(d-1)}$ to be the divided points in order from the side of $C$.
\item For $n=1,\dots ,d-1$, draw a ray $\ell_\sigma^{(n)}$ parallel to $\ell_\sigma$, starting from $A^{(n)}$ and going in the same direction as $\ell_\sigma$,
letting $B_\sigma^{(n)}$ be the intersection point with $B_\sigma C$.
\item Construct a negative $3$D gadget by Construction $\ref{const:negative_3}$ for $\alpha ,\beta_\Lt ,\beta_\Rt$
inside polygonal chain $\ell_\Lt^{(1)}B_\Lt^{(1)}A^{(1)}B_\Rt^{(1)}\ell_\Rt^{(1)}$,
where we add the supersctipt `$(1)$' to the resulting points and lines in the crease pattern.
\item For $n=2,\dots ,d-1$, we choose $\psi_\Lt^{(n)}\in (-\gamma/2,\gamma /2)$ such that
\begin{equation}\label{ineq:psi_n}
\frac{r^2-1}{2(r\cos\psi_\Lt^{(n)}-1)}\cdot p_n\leqslant q_n, \quad\text{or equivalently, }\quad
\frac{1}{r}\left(\frac{r^2-1}{2}\cdot\frac{p_n}{q_n}+1\right)\leqslant\cos\psi_\Lt^{(n)}
\end{equation}
where $r$ is given by $r=1/\cos (\gamma /2)$.
Note that we can \emph{always} take $\psi_\Lt^{(n)}=0$ for $n=2,\dots ,d$ by $\eqref{ineq:pn_qn}$.

Also, if we have $p_1=\dots =p_d=1$ (equal division) and $\psi_\Lt^{(2)}=\dots =\psi_\Lt^{(n)}=\psi_\Lt$, then $\eqref{ineq:psi_n}$ is simplified as
\begin{equation*}
\frac{r^2+3}{4r}\leqslant\cos\psi_\Lt .
\end{equation*}
\item For $n=1,\dots ,d$, draw a perpendicular bisector $m_\sigma^{(n)}$ of $B_\sigma^{(n-1)}B_\sigma^{(n)}$,
in which we determine a point $E_\sigma^{(n)}$ so that $\angle B_\sigma^{(n)}A^{(n)}E_\sigma^{(n)}=\phi_\sigma^{(n)}/2$,
where we set $\phi_\Lt^{(n)}=\gamma /2-\psi_\Lt^{(n)}$, $\phi_\Rt^{(n)}=\gamma /2+\psi_\Lt^{(n)}$,
$A^{(d)}=A$, $A^{(0)}=B_\sigma^{(0)}=C$ and $B_\sigma^{(d)}=B_\sigma$.
Also, redefine $m_\sigma^{(n)}$ to be a ray starting from $E_\sigma^{(n)}$ and going in the same direction as $\ell_\sigma$.
Thus we have $2d$ parallel rays $m_\sigma^{(1)},\ell_\sigma^{(1)},\dots ,m_\sigma^{(d)},\ell_\sigma^{(d)}=\ell_\sigma$ in order from the side of $C$.
\item For $n=2,\dots ,d$, let ${F'}^{(n)}$ be a point in segment $E_\Lt^{(n)}E_\Rt^{(n)}$ with $\angle E_\Lt^{(n)}A^{(n)}{F'}^{(n)}=\phi_\Lt^{(n)}$
(and so $\angle E_\Rt^{(n)}A^{(n)}{F'}^{(n)}=\phi_\Rt^{(n)}$).
\item For $n=1,\dots ,d-1$, draw a ray $k_\sigma^{(n)}$ parallel to $k_\sigma$ from $B_\sigma^{(n)}$ to the side of $m_\sigma^{(n+1)}$.
If $k_\sigma^{(n)}$ intersects segment $A^{(n)}E_\sigma^{(n+1)}$, then let ${G'_\sigma}^{\! (n)}$ be the intersection point.
If not, then let $J_\sigma^{(n+1)}$ be the intersection point with $m_\sigma^{(n)}$.
\item For $n=1,\dots ,d$, draw a ray ${k'_\sigma}^{\! (n)}$ starting from $B_\sigma^{(n)}$
which is a reflection of $k_\sigma^{(n)}$ across $\ell_\sigma^{(n)}$, where we set $k_\sigma^{(d)}=k_\sigma$.
If ${k'_\sigma}^{\! (n)}$ intersects segment $A^{(n)}E_\sigma^{(n)}$, then let $G_\sigma^{(n)}$ be the intersection point.
If not, then it ${k'_\sigma}^{\! (n)}$ intersects $m_\sigma^{(n)}$ at $J_\sigma^{(n)}$ given in $(6)$.
Note that for $2\leqslant n\leqslant d$, $G_\sigma^{(n)}$ exists if and only if ${G'_\sigma}^{\! (n-1)}$ exists.
\item For $n=2,\dots ,d$ such that $G_\sigma^{(n)},{G'_\sigma}^{\! (n-1)}$ exist,
draw a line through $G_\sigma^{(n)}$ which is a reflection of ${k'_\sigma}^{\! (n)}$ across $A^{(n)}E_\sigma^{(n)}$,
and a line through ${G'_\sigma}^{\! (n-1)}$ which is a reflection of ${k'_\sigma}^{\! (n)}$ across $A^{(n)}E_\sigma^{(n)}$,
letting $M_\sigma^{(n)}$ be their common intersection point with segment $E_\Lt^{(n)} E_\Rt^{(n)}$.
\item The desired crease pattern is shown as the solid lines in Figure $\ref{fig:CP_division_negative}$,
and the assignment of mountain and valley folds is given in Table $\ref{tbl:assignment_division_negative}$.
\end{enumerate}
\end{construction}
\begin{figure}[htbp]
\addtocounter{theorem}{1}
\centering\includegraphics[width=0.75\hsize]{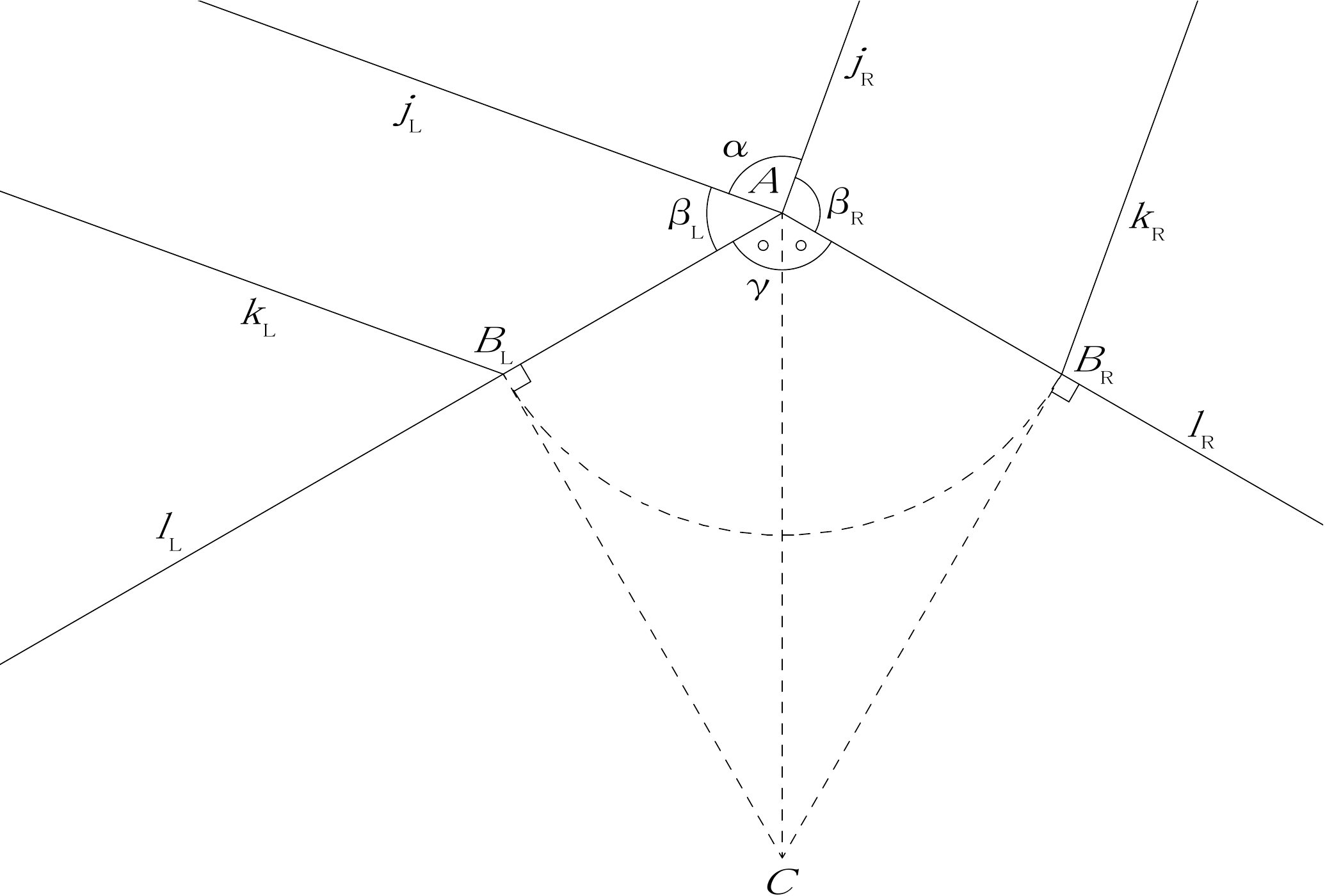}
\caption{Development of a negative $3$D gadget to be divided}
\label{fig:development_division_negative}
\end{figure}
\begin{figure}[htbp]
\addtocounter{theorem}{1}
\centering\includegraphics[width=0.75\hsize]{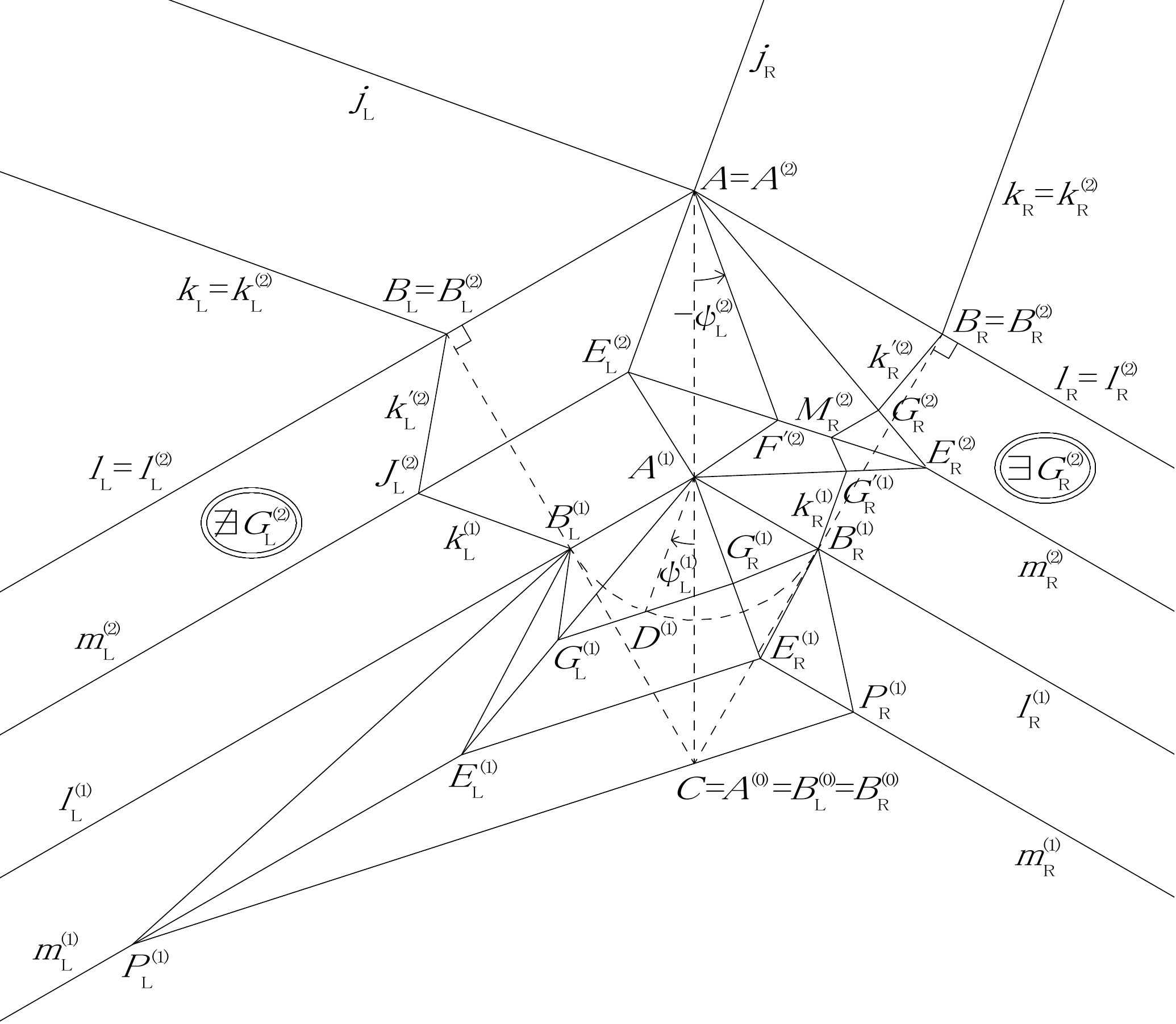}
\caption{CP of the division of a negative $3$D gadget}
\label{fig:CP_division_negative}
\end{figure}
\addtocounter{theorem}{1}
\begin{table}[h]
\begin{tabular}{c|c|c|c}
&common&\multicolumn{2}{c}{$n$ with $2\leqslant n\leqslant d$ such that}\\
\cline{3-4}&creases&$\exists G_\sigma^{(n)}$&$\nexists G_\sigma^{(n)}$\\
\hline mountain&$j_\sigma ,m_\sigma^{(n)},$&\multicolumn{2}{c}{$A^{(n-1)}E_\sigma^{(n)},A^{(n)}{F'}^{(n)}$}\\
\cline{3-4}folds&$A^{(n)}E_\sigma^{(n)},E_\Lt^{(1)}E_\Rt^{(1)},$&$B_\sigma^{(n)}G_\sigma^{(n)},$&$B_\sigma^{(n)}J_\sigma^{(n)}$\\
&$B_\sigma^{(1)}G_\sigma^{(1)},B_\sigma^{(1)}P_\sigma^{(1)}$&${G'}_\sigma^{(n-1)}M_\sigma^{(n)}$&\\
\hline valley&$k_\sigma ,\ell_\sigma^{(n)},$&\multicolumn{2}{c}{$A^{(n-1)}{F'}^{(n)}, E_\Lt^{(n)}E_\Rt^{(n)}$}\\
\cline{3-4}folds&$AB_\sigma ,B_\sigma^{(1)}E_\sigma^{(1)},$&$B_\sigma^{(n-1)}{G'}_\sigma^{(n-1)},$&$B_\sigma^{(n-1)}J_\sigma^{(n)}$\\
&$G_\Lt^{(1)}G_\Rt^{(1)},P_\Lt^{(1)}P_\Rt^{(1)},$&$G_\sigma^{(n)}M_\sigma^{(n)}$&
\end{tabular}\vspace{0.5cm}
\caption{Assignment of mountain and valley folds to the division of a negative $3$D gadget}
\label{tbl:assignment_division_negative}
\end{table}
\begin{remark}\label{rem:meaning_ineq_qn}\rm
Suppose $\eqref{ineq:psi_n}$ does not hold.
Then as we saw in \cite{Doi20}, Proposition $7.7$,
there exists a point $D^{(n)}$ in segment $A^{(n)}{F'}^{(n)}$ such that $\norm{A^{(n)}D^{(n)}}=\norm{A^{(n)}B_\sigma^{(n)}}=q_n$
(and also a point ${D'}^{(n-1)}$ in segment $A^{(n-1)}{F'}^{(n)}$ such that $\norm{A^{(n)}{D'}^{(n-1)}}=\norm{A^{(n-1)}B_\sigma^{(n-1)}}=q_{n-1}$).
Since segment $D^{(n)}{F'}^{(n)}$, which overlaps with ${D'}^{(n-1)}{F'}^{(n)}$, overlaps with both $\ell_\Lt^{(n)}$ and $\ell_\Rt^{(n)}$,
$D^{(n)}{F'}^{(n)}$ glues $\ell_\Lt^{(n)}$ and $\ell_\Rt^{(n)}$ together, which ruins Construction $\ref{const:negative_3}$.
Note that $\eqref{ineq:pn_qn}$ is a necessary and sufficient condition such that there is at least one $\psi_\Lt^{(n)}$ satisfying $\eqref{ineq:psi_n}$.
This is why we require $\eqref{ineq:pn_qn}$.

Also, we saw in \cite{Doi20}, Proposition $7.8$ that there exist $G_\sigma^{(n)}$ and ${G'}_\sigma^{(n-1)}$ for $n\geqslant 2$
in Construction $\ref{const:negative_3}$, $(7)$ and $(8)$, if and only if
\begin{equation*}
\frac{\tan (\gamma /2)}{2}\left(\frac{1}{\tan (\phi_\sigma^{(n)}/2)}-\frac{1}{\tan\beta_\sigma}\right)\cdot p_n\leqslant q_n.
\end{equation*}
Thus to avoid the appearance of $G_\sigma^{(n)}$ for both $\sigma =\Lt ,\Rt$, it is better to choose $\psi_\Lt^{(n)}$ so that
\begin{equation*}
\frac{1}{\tan (\phi_\sigma^{(n)}/2)}\leqslant\frac{2}{\tan (\gamma /2)}\cdot\frac{q_n}{p_n}+\frac{1}{\tan\beta_\sigma}\quad\text{for both }\sigma =\Lt ,\Rt
\end{equation*}
if such $\psi_\Lt^{(n)}$ exists.
In particular, for $p_1=\dots =p_d$ (equal division) and $\psi_\Lt^{(2)}=\dots =\psi_\Lt^{(d)}=\psi_\Lt$, 
there appear no $G_\sigma^{(n)}$ for both $\sigma =\Lt ,\Rt$ and all $n\geqslant 2$ if we can choose $\psi_\Lt$ so that
\begin{equation*}
\frac{1}{\tan (\phi_\sigma /2)}\leqslant\frac{4}{\tan (\gamma /2)}+\frac{1}{\tan\beta_\sigma}\quad\text{for both }\sigma =\Lt ,\Rt ,
\end{equation*}
where $\phi_\Lt =\gamma /2-\psi_\Lt$ and $\phi_\Rt =\gamma /2+\psi_\Lt$.
\end{remark}

\section{Conclusion}\label{sec:conclusion}
In this paper we studied fives constructions of negative $3$D gadgets, whose features are listed in Table $\ref{tbl:features}$.
In particular, we saw in Section $\ref{subsec:negative_new_3}$
that our first and third constructions solve Problem $\ref{prob:existence_negative}$.
These solutions enable us to treat positive and negative $3$D gadgets on the same basis as long as there are no interferences among the $3$D gadgets,
and also broaden the range of $3$D objects which can be folded in origami extrusions, although it is still limited and to be further broadened.
Each of the two constructions has its own merits, and we can select them interchangeably for our purpose.
Our first construction produces an infinite number of negative $3$D gadgets for any positive $3$D gadget in Problem $\ref{prob:existence_negative}$,
but choosing $\angle AB_\tau E_\tau$ in Construction $\ref{const:negative_1}$, $(1)$ is a little troublesome.
On the other hand, our third construction gives a unique solution, for which we do not need to care about the choice.
To prove the existence and uniqueness of the desired negative $3$D gadget by our third construcion
for any given positive $3$D gadget in Problem $\ref{prob:existence_negative}$,
we solved equation $\eqref{eq:DA_EE=BA_EE}$ \emph{numerically} and found rather complicated formulas for the unique solution
in Theorem $\ref{thm:existence_rho}$ and Proposition $\ref{prop:rho_alt}$.
However, we still do not either know how to construct $D$ \emph{geometrically} except for some cases such as $\alpha =\beta_\Lt +\beta_\Rt$,
or understand the geometric meaning of the formulas for the solution
which may possibly explain the existence and uniqueness of the desired gadget more naturally and easily.

Also, it is still a major problem in origami extrusions
to handle outgoing pleats and various interferences which increase as the extruded object becomes more complex.
Although we have many ways of handling an \emph{individual} difficulty
such as changing the angles of the outgoing pleats, composing their intersections, and using repeating $3$D gadgets,
at this point we depend more or less on our skills in designing the \emph{whole} crease pattern of an origami extrusion,
which is to be further developed in future researches.
\renewcommand{\arraystretch}{1.2}
\addtocounter{theorem}{1}
\begin{table}[h]\label{tbl:features}
\begin{tabular}{c|c|c|c|c|c}
\multirow{2}{*}{Section}&\multirow{2}{*}{constructed}&outgoing&number of&symmetry${}^{*1}$&\multirow{2}{*}{remark}\\
&&pleats&possible CPs&of CP&\\ \hline
\multirow{2}{*}{$2.1$}&\multirow{2}{*}{many paper folders}&$(b,c_\Lt )$ or&\multirow{2}{*}{$2$}&\multirow{2}{*}{no}&$\delta_\Lt$ and $\delta_\Rt$ are\\
&&$(b,c_\Rt )$&&&not prescribed\\ \hline
\multirow{2}{*}{$2.2$}&Cheng $(\delta_\Lt =\delta_\Rt =0)$,&$(\ell'_\Lt ,m_\Lt ),$&\multirow{2}{*}{$2$}&\multirow{2}{*}{no}&\multirow{2}{*}{covers all cases}\\
&author (general cases)&$(\ell'_\Rt ,m_\Rt )$&&&\\ \hline
\multirow{2}{*}{$3.1$}&\multirow{2}{*}{author}&$(\ell_\Lt ,m_\Lt ),$&\multirow{2}{*}{$\infty +\infty^{*2}$}&\multirow{2}{*}{no}&covers all cases,\\
&&$(\ell_\Rt ,m_\Rt )$&&&solves Problem $\ref{prob:existence_negative}$\\ \hline
\multirow{2}{*}{$3.2$}&\multirow{2}{*}{author}&$(\ell_\Lt ,m_\Lt ),$&\multirow{2}{*}{$0,1$ or $\infty^{*3}$}
&\multirow{2}{*}{yes}&needs $\beta_\Lt ,\beta_\Rt\leqslant\pi /2$\\
&&$(\ell_\Rt ,m_\Rt )$&&&or $\beta_\Lt ,\beta_\Rt\geqslant\pi /2$\\ \hline
\multirow{2}{*}{$3.3$}&\multirow{2}{*}{author}&$(\ell_\Lt ,m_\Lt ),$&$1$&\multirow{2}{*}{yes}&covers all cases,\\
&&$(\ell_\Rt ,m_\Rt )$&(unique)&&solves Problem $\ref{prob:existence_negative}$
\end{tabular}\\
\flushleft{\small $\ast 1$ The symmetry of the resulting crease pattern for $\beta_\Lt =\beta_\Rt ,\delta_\Lt =\delta_\Rt$.\\
$\ast 2$ There are infinitely many crease patterns for each $\tau\in\{\Lt ,\Rt\}$.\\
$\ast 3$ There are infinitely many crease patterns only for $\beta_\Lt =\beta_\Rt =\pi /2$.}\vspace{0.5cm}
\caption{Features of five constructions of negative $3$D gadgets}
\end{table}


\begin{thebibliography}{9}
\bibitem{Cheng} Herng Yi Cheng, Extruding towers by serially grafting prismoids,
presentation slides in {\it The 6th International Meeting on Origami in Science, Mathematics and Education (6OSME)}, held at The University of Tokyo,
available at http://www.herngyi.com/uploads/2/3/1/7/23170532/6osme{\_}slides.pdf, 2014.
\bibitem{Doi19}Mamoru Doi, New efficient flat-back 3D gadgets in origami extrusions compatible with the conventional pyramid-supported 3D gadgets,
arXiv:cs.CG/1908.07342, 2019.
\bibitem{Doi20}Mamoru Doi, Improved flat-back 3D gadgets in origami extrusions completely downward compatible with the conventional pyramid-supported 3D gadgets,
arXiv:cs.CG/2007.01859, 2020.
\bibitem{Natan} Carlos~Natan~Lop\'{e}z Nazario, Origami~3D~tessellations, {\it Carlos~Natan~Lop\'{e}z Nazario's~photostream~at~flickr.com},
http://www.flickr.com/photos/origamiz/sets/72157606559615966/, 2010.
\end{thebibliography}
\end{document}